\documentclass[12pt]{article}
\usepackage{amsmath}
\usepackage{graphicx}
\usepackage{enumerate}
\usepackage{natbib}
\usepackage{url} 
\usepackage{amsmath,bm,amsfonts,url}
\usepackage{amsthm} 
\usepackage{amssymb}
\usepackage[dvipsnames]{xcolor}
\usepackage{setspace}
\usepackage{graphicx}
\usepackage{booktabs}

\newcommand{\widgraph}[2]{\includegraphics[keepaspectratio,width=#1]{#2}}

\newcommand{\eps}{\epsilon}

\newcommand{\PE}{\text{PE}}
\newcommand{\RPE}{\text{RPE}}
\newcommand{\pPE}{\overline{\PE}}
\newcommand{\pRPE}{\overline{\RPE}}
\newcommand{\pFDR}{\overline{\text{FDR}}}
\newcommand{\pFNR}{\overline{\text{FNR}}}

\newcommand{\xbar}{\bar{x}}
\newcommand{\Xbar}{\bar{X}}
\newcommand{\ybar}{\bar{y}}
\newcommand{\Ybar}{\bar{Y}}

\newcommand{\Gh}{{G}}  
\newcommand{\Ghat}{\widehat{G}}
\newcommand{\Gbar}{\bar{G}}

\newcommand{\Fh}{{F}} 
\newcommand{\Fhat}{\widehat{F}}
\newcommand{\NN}{\mathcal{N}}
\newcommand{\GG}{\mathcal{G}}
\newcommand{\VV}{\mathcal{V}}
\newcommand{\EE}{\mathcal{E}}
\newcommand{\SSS}{\mathcal{S}}
\newcommand{\SSb}{\bar{\SSS}}

\renewcommand{\Re}{\mathbb{R}}
\newcommand{\PP}{\mathcal{P}}

\newcommand{\eIP}{p_e} 
\newcommand{\dt}{\widetilde{d}}

\newcommand{\blind}{1}

\newtheorem{theorem}{Theorem}

\newtheorem{lemma}{Lemma}
\newtheorem{proposition}{Proposition}

\newtheorem{assumption}{Assumption}

\addtocounter{page}{-1}
\begin{document}

\def\spacingset#1{\renewcommand{\baselinestretch}%
{#1}\small\normalsize} \spacingset{1}


\if1\blind
{
  \title{\bf Bayesian Multivariate Density-Density Regression}
  \author{Khai Nguyen$^1$, Yang Ni$^1$, and Peter Mueller$^{1,2}$ \\
    $^1$Department of Statistics and Data Sciences, University of Texas at Austin \\
    $^2$Department of Mathematics, University of Texas at Austin
}
  \maketitle
  \thispagestyle{empty}
} \fi

\if0\blind
{
  \thispagestyle{empty}
  \bigskip
  \bigskip
  \bigskip
  \begin{center}
    {\LARGE\bf Bayesian Density-Density Regression with Application to Cell-Cell Communications}
\end{center}
  \medskip
} \fi

\bigskip
\begin{abstract}
We introduce a novel and scalable Bayesian framework for multivariate-density-density regression (DDR), designed to model relationships between multivariate distributions. Our approach addresses the critical issue of distributions residing in spaces of differing dimensions. We utilize a generalized Bayes framework, circumventing the need for a fully specified likelihood by employing the sliced Wasserstein distance to measure the discrepancy between fitted and observed distributions. This choice not only handles high-dimensional data and varying sample sizes efficiently but also facilitates a Metropolis-adjusted Langevin algorithm (MALA) for posterior inference. Furthermore, we establish the posterior consistency of our generalized Bayesian approach, ensuring that the posterior distribution concentrates around the true parameters as the sample size increases. Through simulations and application to a population-scale single-cell dataset, we show that Bayesian DDR provides robust fits, superior predictive performance compared to traditional methods, and valuable insights into complex biological interactions.
\end{abstract}

\noindent%
\vfill

\newpage

\spacingset{1.9} 

\section{Introduction}
\label{sec:introduction}

We propose the first approach for regressing random distributions on random distributions in \textit{multivariate settings}. 
A key feature of our approach, which is Bayesian, is its ability to
provide a full probabilistic characterization of all relevant
unknowns, including data and model parameters. However, this feature
can also become a limitation, particularly when dealing with complex
relationships in which specifying a complete probabilistic model is
challenging. To address this, we adopt a generalized Bayes framework
that circumvents the need to
 start with a 
fully specify likelihood. Our approach
leverages the sliced Wasserstein distance to measure the discrepancy
between the observed response distribution and a predicted
distribution, allowing for a principled update of beliefs while
maintaining flexibility.

In the usual regression setup, both the predictors and responses are typically scalars or
vectors.
When distributional variables serve as the predictors and real-valued variables as the response, the density
regression problem 
arises~\citep{poczos2013distribution}, e.g., predicting health
indicators from a distribution of clinical
outcomes~\citep{szabo2016learning} or estimating the effects
of physical activity on
biomarkers~\citep{matabuena2023distributional}. 
Another type of density regression arises with real-valued
predictors and density-valued responses. This is studied, for example, 
in \cite{tokdar2004bayesian,dunson2007bayesian} or
\cite{shen2016adaptive}. Recently, the density-density regression
(DDR) problem, which deals with distributional predictor and
distributional response 
variables, has been studied. For example, regression with
one-dimensional predictor and one-dimensional response is considered
in~\cite{zhao2023density} using a Riemannian representation of density
functions. The method of~\cite{zhao2023density} is hard to extend to
\textit{multivariate} density functions due to the usage of a wrapping
function, which is only well-defined in one dimension. 
To the best of our knowledge, there are currently no methods for DDR with \textit{multivariate} distributions. Extending existing univariate methods to this multivariate setting is non-trivial, as most of them depend on the notion of the quantile function, which lacks a universally accepted definition for multivariate distributions.

We propose a DDR model, where both
predictors and responses are random 
distributions of arbitrary dimensions. The motivation comes from the
challenge of modeling molecular dependence relationships among
different cell types in population-scale single-cell data. In this
type of data, gene expressions are measured for a large number of
cells donated by a large number of subjects. 
Exploiting that cell types can be
identified based on known cellular markers, this provides an
opportunity to decipher the communication patterns between different
cell types by way of expressions of ligand-receptor pairs. Because of 
the availability of a large number of cells from each donor, we, in
essence, observe for each subject the empirical
distributions of the gene expressions for each cell type. Therefore,
we can cast the inference problem of cell-cell communications as a DDR
problem,
where the predictor for each subject
is the distribution of gene expression for the
ligands from one cell type (sender cells produce ligands as signaling molecules), and the response is the distribution of
receptors from another cell type (receiver cells express receptors and respond to the signals).
Conventional bulk gene expression methods collapse the
distributional representation by taking the mean as the representative
statistic for downstream regression. In contrast, the proposed
Bayesian DDR approach allows us to model complex regression dependencies
between distributions while preserving the full variability and
uncertainty of the data.

Our inference procedure is based on
the posterior distribution for a 
regression function that maps from the space of 
the predictor distribution to the space of 
the response distribution. Avoiding  to state  an explicit sampling model, 
we rely instead on a \textit{generalized
  likelihood}, which is constructed based on a discrepancy between the
pushforward (predicted) response distribution and the observed
response distribution. Hence, the proposed Bayesian DDR 
approach falls within the generalized Bayesian inference
framework~\citep{bissiri2016general}.

One challenge in implementing DDR is that we do not directly observe
distributions through their density functions. Instead, we
observe \textit{finite} samples from distributions. Moreover, the
number of observed samples can differ between distributions, and
distributions may lie in high-dimensional spaces. Therefore,
 we base  the construction of the generalized likelihood 
 on a 
discrepancy between distributions that can be reliably
estimated with \textit{finite} samples in high dimensions.
We choose the sliced Wasserstein
distance~\citep{rabin2012wasserstein,bonneel2015sliced,peyre2020computational},
which is a sliced variant of the
Wasserstein distance through random projection ~\citep{villani2009optimal}. Sliced Wasserstein
does not suffer from the curse of dimensionality and is
computationally scalable with respect to the number of support points
in the distributions. Moreover, the use of sliced Wasserstein allows
us to derive a Metropolis-adjusted Langevin algorithm
(MALA,~\citealt{girolami2011riemann}) for efficient posterior inference and to guarantee the consistency of the posterior.

Using the proposed Bayesian DDR, we
construct a directed cell-cell communication network of different cell types
based on ligand–receptor interactions for each ordered pair of cell types. 
We  carry out 
Bayesian DDR for each pair of cell types, i.e., regressing
distributions over receptors onto distributions over ligands,
 implemented as a separate instance of DDR for each pair of cell
types. 
After fitting the Bayesian DDR models, we propose
two ways to infer the cell-cell communication network. The first option is
a weighted 
fully-connected graph with edge weights derived from the expected
relative errors of the regressions. The second approach
casts the graph selection as a decision problem, which we
implement by way of a simple decision boundary. 
For simplicity, we parameterize the regression as a
linear function and put a shrinkage prior on the linear 
coefficient, which can be easily extended to nonlinear functions.
We report edges by thresholding the posterior probabilities in
these linear functions, with the threshold determined by minimizing
posterior expected false negative rate (FNR) under a bound on
posterior expected false discovery rate (FDR). 

The remainder of the article is organized as follows. In
Section~\ref{sec:preliminaries}, we review the definitions of
Wasserstein and sliced Wasserstein distance, and their
statistical and computational properties. In
Section~\ref{sec:Bayesian_DDR}, we propose a generalized Bayesian DDR
framework using sliced Wasserstein distance to define a 
generalized likelihood, and we derive posterior simulation using
a MALA sampler.
In Section~\ref{sec:DDR_B_to_T}, we apply the proposed
Bayesian DDR approach to infer cell-cell communication between
two cell types. 
In Section~\ref{sec:graph_discovery},
we build on the regression setup and propose discovery of a 
directed graph by way of forming Bayesian DDRs for ordered pairs of cell
types.
In Section~\ref{sec:cell_graph_exp}, we
apply the proposed graph discovery for 
graphs among multiple cell types in
a population-scale single-cell dataset, 
and we validate inference in a semi-simulation 
with real predictors and a simulated true regression
function.

For notations, we use $\delta_x$ to denote the Dirac delta function, which takes the value 1 at $x$ and 0 otherwise. For any $d \geq 2$, we denote $\mathbb{S}^{d-1} = \{\theta \in \mathbb{R}^{d} \mid ||\theta||_2^2 = 1\}$ as the unit hypersphere. We use $I$ for the identity matrix and $\mathbf{1}$ for the vector of all 1's, with the size depending on the context. For any two sequences $a_{n}$ and $b_{n}$, the notation $a_{n} = \mathcal{O}(b_{n})$ means that $a_{n} \leq C b_{n}$ for all $n \geq 1$, where $C$ is some universal constant. Other notations will be introduced later when they are used.

\section{Background on Sliced Wasserstein Distance}
\label{sec:preliminaries}
By way of a brief review of sliced Wasserstein distance, we introduce
some notation and definitions. 
Given $p\geq 1$, let $G_1, G_2 \in \PP_p(\Re^d)$ where
$\PP_p(\Re^d)$ be the set of all distributions supported on
$\Re^d$ 
with finite $p$-th moment.
Wasserstein-$p$ ($p \geq 1$)
distance~\citep{villani2009optimal} between  $G_1$ and $G_2$ is
defined as:
\begin{align}
\label{eq:Wasserstein}
    W_p^p(G_1, G_2) = \inf_{\pi \in \Pi(G_1, G_2)} \int_{\Re^d \times \Re^d} \|x- y\|_p^p \, \mathrm{d}\pi(x, y),
\end{align}
where $\Pi(G_1, G_2) = \left\{\pi \in \PP(\Re^d \times \Re^d) \mid \pi(A, \Re^d) = G_1(A), \ \pi(\Re^d, B) = G_2(B) \ \forall A, B \subset \Re^d \right\}$ is the set of all transportation plans/couplings. When observing $x_1,\ldots,x_{m_1} \overset{i.i.d.}{\sim} G_1$ and $y_1,\ldots,y_{m_2} \overset{i.i.d.}{\sim} G_2$, a plug-in estimation of the Wasserstein-$p$ distance is given by:
\begin{align}
\label{eq:empirical_Wasserstein}
    W_p^p(\Ghat_{1}, \Ghat_{2}) = \min_{\gamma \in \Gamma(\mathbf{1}/m_1, \mathbf{1}/m_2)} \sum_{i=1}^{m_1} \sum_{j=1}^{m_2} \|x_i- y_j\|_p^p \gamma_{ij},
\end{align}
where  $\Ghat_{1}=\frac{1}{m_1}\sum_{i=1}^{m_1} \delta_{x_i}$ and
$\Ghat_{2}=\frac{1}{m_2}\sum_{j=1}^{m_2} \delta_{y_j}$ are
corresponding empirical distributions over i.i.d. samples from $G_1$
and $G_2$, and the set of discrete transportation plans is
$\Gamma(\mathbf{1}/m_1, \mathbf{1}/m_2) = 
\left\{\gamma \in \Re_+^{m_1 \times m_2} \mid \gamma \mathbf{1}
=\mathbf{1}/m_1 \text{ and }
  \gamma^\top \mathbf{1} =  \mathbf{1}/m_2\right\}$. Computing $ W_p^p(\Ghat_{1},\Ghat_{2})$ is a linear programming problem, which costs $\mathcal{O}((m_1+m_2)^3\log(m_1+m_2))$ in time complexity~\citep{peyre2020computational}. From~\cite{fournier2015rate}, we have the following sample complexity: $\mathbb{E}\left[\left|W_p(\Ghat_{1}, \Ghat_{2})-W_p(G_1, G_2)\right|\right] = \mathcal{O}(m_1^{-1/d}+m_2^{-1/d})$, which implies that we need $m_1$ and $m_2$ to be large to achieve small approximation errors. However, the time complexity to compute the Wasserstein distance scales poorly with $m_1$ and $m_2$. 

 This conundrum is addressed by introducing 
sliced Wasserstein (SW)
distance~\citep{rabin2014adaptive,bonneel2015sliced}. SW relies on the
fact that one-dimensional Wasserstein distance can be computed
efficiently. In particular, let $G_1,G_2 \in
\mathcal{P}_p(\mathbb{R})$, computing $W_p^p(\Ghat_{1},
\Ghat_{2})$ in Equation~\eqref{eq:empirical_Wasserstein} costs only
$\mathcal{O}((m_1+m_2)\log (m_1+m_2))$ in time complexity  by using
the north-west corner solution~\citep{peyre2020computational}.

SW distance between two distributions $G_1,G_2 \in
\mathcal{P}_p(\mathbb{R}^d)$ is defined as:
\begin{align}
        \label{eq:SW}
SW_p^p(G_1, G_2) = \mathbb{E}_{\theta \sim
  \mathcal{U}(\mathbb{S}^{d-1})}[W_p^p(\theta \sharp G_1, \theta
\sharp G_2)],
\end{align}
where  $\mathcal{U}(\mathbb{S}^{d-1})$ is the uniform distribution
over the unit hypersphere in $d$ dimension, and $\theta \sharp G_1$
and $\theta \sharp G_2$ denote the pushforward distribution of $G_1$
and $G_2$ through the function $f_\theta(x) =  \theta^\top x$. In
particular, given two measurable spaces $(\mathcal{X}_1,\Sigma_1)$ and
$(\mathcal{X}_2,\Sigma_2)$, a measurable function $f:\mathcal{X}_1 \to
\mathcal{X}_2$, and a measure $\mu:\Sigma_1 \to [0,\infty)$, the
push-forward of $\mu$ through $f$ is $f\sharp \mu(B) = \mu(f^{-1}(B))$
for any $B\in \Sigma_2$  -- we will use the definition with
probability measures only. The push-forward measure is a general measure-theoretic concept. When the mapping is bijective, typically between spaces of the same dimension, the change of variables formula can be used to explicitly compute the density of the push-forward measure, provided the measures involved are absolutely continuous.
Similar to  Wasserstein, we have the following plug-in estimation of SW: 
\begin{align}
    \label{eq:empirical_SW}
    SW_p^p(\Ghat_{1}, \Ghat_{2}) = \mathbb{E}_{\theta \sim \mathcal{U}(\mathbb{S}^{d-1})}[W_p^p(\theta \sharp \Ghat_{1}, \theta \sharp \Ghat_{2})].
\end{align}
From~\cite{nadjahi2020statistical,nguyen2021distributional,nietert2022statistical,boedihardjo2025sharp},
we have that $\mathbb{E}\left[\left|SW_p(\Ghat_{1},
    \Ghat_{2})-SW_p(G_1, G_2)\right|\right] =
\mathcal{O}(m_1^{-1/2}+m_2^{-1/2})$.
 In particular, there is no dimension dependence. 
Combining with the linear-time complexity of empirical SW, we have
a both, computationally and statistically, scalable solution for
comparing distributions. The last step to compute SW is approximating
the expectation with respect to
$\mathcal{U}(\mathbb{S}^{d-1})$~\citep{nguyen2024quasimonte,leluc2024slicedwasserstein,nguyen2024control,sisouk2025user}. For
example, we can use simple Monte Carlo estimation: 
\begin{align}
    \label{eq:MC_empirical_SW}
    \widehat{SW}_p^p(\Ghat_{1}, \Ghat_{2};L) = \frac{1}{L}\sum_{l=1}^LW_p^p(\theta_l \sharp \Ghat_{1}, \theta_l \sharp \Ghat_{2}),
\end{align}
where $\theta_1,\ldots,\theta_L \overset{i.i.d.}{\sim} \mathcal{U}(\mathbb{S}^{d-1})$ with $L$ being the number of Monte Carlo samples or the number of projections. The overall time complexity for computing this approximation is $\mathcal{O}(L (m_1+m_2)\log(m_1+m_2)+Ld(m_1+m_2))$ where $\mathcal{O}(Ld(m_1+m_2))$ is for projecting supports of $\Ghat_{1}$ and $\Ghat_{2}$, and $\mathcal{O}(L (m_1+m_2)\log(m_1+m_2))$ is for computing  $W_p^p(\theta_l \sharp \Ghat_{1}, \theta_l \sharp \Ghat_{2})$ with $\theta_l \sharp \Ghat_{1}=\frac{1}{m_1}\sum_{i=1}^{m_1} \delta_{\theta_l^\top x_i}$ and $\theta_l \sharp \Ghat_{2}=\frac{1}{m_2}\sum_{j=1}^{m_2} \delta_{\theta_l^\top y_j}$ for $l=1,\ldots,L$ .

\section{Generalized Multivariate Bayesian Density-Density Regression}
\label{sec:Bayesian_DDR}

\subsection{Generalized Multivariate Bayesian Density-Density Regression}
\label{subsec:Bayesian_DDR}


We consider the inference problem of relating a distribution-valued
 response $G_i \in \PP_2(\Re^{d_2})$ ($d_2\geq 2$) to a
 distribution-valued predictor $F_i \in \PP_2(\Re^{d_1})$ ($d_1\geq
 2$) for $i=1,\ldots,N$ ($N>0$). We represent the data as
 $\SSS=\{(F_{i},G_{i})\}_{i=1}^N$ and define the following DDR model.
The model is specified as a {\em generalized likelihood}
\citep{bissiri2016general}   $\ell(f;\Gh_i,\Fh_i)$ based on a loss function for a fitted
approximation of $\Gh_i$: 
\begin{align}
    \label{eq:DDRmodel}
    \ell(f;\Gh_i,\Fh_i) = \exp\left(-w SW_2^2(f \sharp \Fh_i, \Gh_i)\right),
\end{align}
where $f$ is a measurable function that maps from $\mathbb{R}^{d_1}$ to
$\mathbb{R}^{d_2}$, $w > 0$, and $f\sharp
  F_i$ denotes the push-forward measure of $F_i$ through $f$.

 In \eqref{eq:DDRmodel} 
we are using SW distance to define a loss function that serves as
generalized (negative log) likelihood.
SW is an attractive choice for the loss function as it naturally
defines a distance between random distributions $f \sharp
\Fh_i$ and $\Gh_i$. 
One of the
advantages of using the generalized likelihood is that it
 sidesteps the need  to start with 
a full probabilistic description of all relevant
unknowns.
 Let $Y_i = (y_{i1},\ldots,y_{im})$,
when $G_i = \frac{1}{m}\sum_{j=1}^m \delta_{y_{ij}}$ for
$i=1,\ldots,N$ is an empirical distribution with $m$ atoms,
\eqref{eq:DDRmodel} implies 
a hypothetical sampling model 
\begin{align}
  Y_i \mid \Fh_i,f \sim p(Y_i \mid \Fh_i, f)  \propto
  \exp\left(-w SW_2^2 \left(f \sharp \Fh_i, \frac{1}{m}\sum_{j=1}^m
  \delta_{y_{ij}}\right)\right).
  \label{eq:splg}
\end{align}
Since $SW_2^2\left(f_\sharp F_i,\frac{1}{m}\sum_{j=1}^m
\delta_{y_{ij}}\right)$ is bounded due to the assumption of finite
second moments of distributions, the normalizing constant in  \eqref{eq:splg}
is bounded. 
 However, as a likelihood function of $f$ such normalization constant cannot be ignored, making \eqref{eq:splg} distinct from the generalized likelihood $\ell(f;G_i,F_i)$ in \eqref{eq:DDRmodel}.



To assess goodness of fit, we use the following  residual  error
between the  fitted distribution and the recorded  response
distributions: 
$PE(\SSS)=  \frac{1}{N} \sum_{i=1}^{N}SW_2^2(f \sharp \Fh_{i}, \Gh_{i})$.
 Averaging with respect to  a 
posterior sample $f_{1},\ldots,f_{T} \sim p(f \mid \mathcal{S})$, we
have the following approximated mean predictive error: 
$ 
    \pPE(\mathcal{S}) = \frac{1}{T}\frac{1}{N} \sum_{t=1}^T  \sum_{i=1}^{N}SW_2^2(f_{t} \sharp \Fh_{i}, \Gh_{i}).
    $ 
 To calibrate $\pPE$ we introduce      
a reference model $f_{0}$ that serves as a worst-case reference (see
later for examples),
 in a similar spirit as using an  intercept-only model in a
linear regression.
We normalize the relative error to obtain the following
relative  residual  error (RPE):
\begin{align}
\label{eq:RPE2}
    \RPE(\mathcal{S}) = \frac{1}{N} \sum_{i=1}^{N}\frac{SW_2^2(f \sharp \Fh_{i}, \Gh_{i})}{SW_2^2(f_{0} \sharp \Fh_{i}, \Gh_{i})}.
\end{align}
When $\RPE(\mathcal{S})=0$, it means a perfect fit as SW is a distance between distributions. When $\RPE(\mathcal{S})=1$, it means a bad fit. Similarly, the mean RPE can be estimated using posterior samples, i.e., 
\begin{align}
\label{eq:RPE}
    \pRPE(\mathcal{S}) = \frac{1}{T} \frac{1}{N}\sum_{t=1}^T  \sum_{i=1}^{N}\frac{SW_2^2(f_{t} \sharp \Fh_{i}, \Gh_{i})}{SW_2^2(f_{0t} \sharp \Fh_{i}, \Gh_{i})},
\end{align}
where posterior samples $f_{1}, \ldots, f_T \sim p(f \mid \mathcal{S})$ and posterior samples $f_{01}, \ldots, f_{0T} \sim p(f_{0} \mid \mathcal{S})$. Later, in the case where $f$ takes a linear form, we will define $f_{0}$ as an intercept model. 

\subsection{Posterior Consistency}
\label{subsec:posterior_inference}
 The generalized likelihood \eqref{eq:DDRmodel} implies a
generalized posterior 
\begin{align}
  \label{eq:pseudo_posterior} p(f\mid \mathcal{S}) \propto p(f)
  \exp\left(-w \sum_{i=1}^N SW_2^2(f \sharp \Fh_i, \Gh_i)\right),
\end{align}
where $p(f)$ is a prior distribution of $f$. 
We now discuss the posterior consistency when $f$ takes a parametric
form, i.e., $f_\phi$ with parameters $\phi \in \Phi$ and
$(F_1,G_1),\ldots,(F_N,G_N) \overset{i.i.d}{\sim} P$ for an unknown
 truth  $P$.  Under  a prior $p(\phi)$, we have the
generalized posterior:
\begin{equation}
  \label{eq:posterior_parameter}
  p(\phi \mid \mathcal{S}) \propto
  p(\phi) \exp\left(-w \sum_{i=1}^N SW_2^2(f_\phi \sharp
  \Fh_i, \Gh_i)\right). 
\end{equation}
 Defining empirical risk and population risk 
\begin{align}
    &R_N (\phi) = \frac{1}{N}\sum_{i=1}^n SW_2^2(f_\phi \sharp F_i,G_i), \quad 
    R(\phi) = \mathbb{E}_{(F,G)\sim P}[SW_2^2(f_\phi \sharp F,G)],
\end{align}
we make the following assumptions for posterior consistency.

\begin{assumption}[Identifiability]\label{assumption:identifiability}
There exists $\phi_0 \in \Phi$ such that $R(\phi)$ attains its unique minimum at $\phi_0$. Moreover, for every $\epsilon > 0$,  
$
    \Delta(\epsilon) \;=\; \inf_{\{\phi \in \Phi : \|\phi - \phi_0\|_2 \geq \epsilon\}} \big( R(\phi) - R(\phi_0) \big) \;>\; 0.
$
\end{assumption}

\begin{assumption}[Compactness]\label{assumption:compact}
The parameter space $\Phi$ is compact.
\end{assumption}

\begin{assumption}[Bounded second moments]\label{assumption:secondmoment}
For all $(F,G) \in \mathrm{supp}(P)$,
$
    \mathbb{E}_{X\sim F}[\|X\|_2^2] \leq M_1, \quad  \mathbb{E}_{Y\sim G}[ \|Y\|_2^2]\leq M_2,
$
for some constants $M_1, M_2 < \infty$.
\end{assumption}

\begin{assumption}[Prior positivity]\label{assumption:prior}
The prior $\pi$ assigns positive mass to every neighborhood of
$\phi_0$; that is, for every $\epsilon > 0$,   
$
    \pi\big\{B_\epsilon(\phi_0)\big\} > 0,  
    \quad \text{where} \quad  
    B_\epsilon(\phi_0) = \{\phi \in \Phi : \|\phi - \phi_0\|_2 < \epsilon\}.
$
\end{assumption}

\begin{assumption}[Regularity of Regression Function]\label{assumption:continuity}
The regression function $f_\phi$ admits $\omega: \Re_+\to\Re_+$ ($\lim_{t\to 0} \omega(t)=0$) as a modulus of continuity: $\mathbb{E}_{X\sim F}\|f_\phi(X) - f_{\phi'}(X)\|_2^2 \leq \omega(\|\phi -\phi'\|_2^2)$ for all $(F,G) \in \mathrm{supp}(P)$ and $\phi,\phi' \in \Phi$. For all $(F,G) \in \mathrm{supp}(P)$ and $\phi \in \Phi$,
$
    \mathbb{E}_{X\sim F} [\|f_\phi(X)\|_2^2] \leq C,
$
for a constant $ C < \infty$.
\end{assumption}

\begin{theorem}\label{theoremm:posterior_consistency}
 Under 
Assumptions~\ref{assumption:identifiability}–\ref{assumption:continuity},
for every $\epsilon>0$, the posterior measure $\pi_N$
of~\eqref{eq:posterior_parameter} satisfies 
\begin{align}
    \pi_N\big(\{ \phi \in \Phi : \|\phi - \phi_0\|_2 \geq \epsilon \}\big) \xrightarrow{a.s} 0 ,
\end{align}
as $N \to \infty$ under i.i.d sampling.
\end{theorem}

The proof of Theorem~\ref{theoremm:posterior_consistency} is given in
Appendix~\ref{subsec:proof:theoremm:posterior_consistency}, which
requires proving the uniform law of a large number for $R_N(\phi)$ and
$R(\phi)$. 

\subsection{ Posterior Simulation}
We implement the posterior inference
by Markov chain Monte Carlo (MCMC) simulation from
\eqref{eq:pseudo_posterior}. Again, we
restrict the function $f$ to a parametric form, i.e., $f_\phi$ with
parameters $\phi \in \Phi$ for computational efficiency. 
We restrict the function $f_\phi$ and  the prior $p(\phi)$  to be
differentiable with respect to $\phi$ to facilitate the use of
gradient-based MALA (one-step hybrid Monte Carlo) transition
probabilities  using  the following proposal distribution:
$q(\phi' \mid \phi) \propto \exp\left(-\frac{1}{4\eta} \|\phi' -
 \phi - \eta \nabla_\phi \log p(\phi \mid \mathcal{S})\|_2^2\right)$,
for a fixed step size $\eta > 0$.
We sample $\phi' \sim q(\phi' \mid \phi)$ as
$\phi' = \phi + \eta \nabla_\phi \log p(\phi \mid \mathcal{S}) +
 \sqrt{2\eta} \epsilon_0$,
with $\epsilon_0 \sim \mathcal{N}(0, I)$, a standard multivariate
Gaussian distribution  of  dimension matching the parameter
$\phi$. We accept $\phi'$ with probability: 
$
\alpha = 
  1 \wedge
  {p(\phi \mid \mathcal{S}) q(\phi \mid \phi')}\big/
  {\left(p(\phi' \mid \mathcal{S}) q(\phi' \mid \phi)\right)}.
$

In practice, we cannot evaluate the generalized likelihood since $SW_2^2(f_\phi \sharp
  \Fh_i, \Gh_i)$
  is intractable  due to two reasons.
 First, we only  observe $F_i$ and $G_i$ through their samples
  i.e., $x_{i1},\ldots,x_{im_{i1}} \sim F_i 
  $ and $y_{i1},\ldots,y_{im_{i2}} \sim G_i$.
Second, the expectation in the definition of SW \eqref{eq:SW} is
intractable. As discussed in Section~\ref{sec:preliminaries}, we can
use plug-in estimate \eqref{eq:empirical_SW} as a simple solution for
the first issue. For the second issue, we can use Monte Carlo estimate
as a solution \eqref{eq:MC_empirical_SW}. For the posterior inference
purpose, we replace the generalized likelihood $p(\phi\mid \SSS)$ by
its estimation $\hat{p}(\phi\mid \SSS) \propto  p(\phi) \exp\left(-w
\sum_{i=1}^N \widehat{SW}_2^2(f_\phi \sharp \Fhat_i, \Ghat_i;L)\right)$.

 Under these approximations  we derive an
  estimate of the gradient of the log posterior density
$\nabla_\phi \log p (\phi \mid \mathcal{S}) \approx \nabla_\phi \log p(\phi) + \nabla  \hat{p}(\mathcal{S} \mid \phi)= \nabla_\phi \log p(\phi)
- w \sum_{i=1}^N \nabla_\phi \widehat{SW}_p^p(f_\phi \sharp \Fhat_i, \Ghat_i;L)$,
where
$\nabla_\phi \widehat{SW}_2^2(f_\phi \sharp \Fhat_i, \Ghat_i; L) =
\frac{1}{L} \sum_{l=1}^L \nabla_\phi W_2^2(\theta_l \sharp f_\phi
\sharp \Fhat_i, \theta_l \sharp \Ghat_i)$
with
\begin{equation}
  \nabla_\phi W_2^2(\theta_l \sharp f_\phi \sharp \Fhat_i, \theta_l
  \sharp \Ghat_i) = \nabla_\phi \left(\min_{\gamma \in
      \Gamma(\mathbf{1}/m_{i1}, \mathbf{1}/m_{i2})}
    \sum_{i'=1}^{m_{i1}} \sum_{j'=1}^{m_{i2}}
    (\theta_l^\top f_\phi(x_{ii'}) - \theta_l^\top
    y_{ij'})^2 \gamma_{i'j'}\right). 
\end{equation}
Since the function
$\sum_{i'=1}^{m_{i1}} \sum_{j'=1}^{m_{i2}} (\theta^\top
f_\phi(x_{ii'}) - \theta^\top y_{ij'})^2$
is continuous in $\phi$ and the set of discrete transportation plans
$\Gamma(\mathbf{1}/m_{i1}, \mathbf{1}/m_{i2})$ is a compact set, we
can apply Danskin's envelope theorem with $\gamma^\star =
\text{argmin}_{\gamma \in \Gamma(\mathbf{1}/m_{i1},
\mathbf{1}/m_{i2})} \sum_{i'=1}^{m_{i1}} \sum_{j'=1}^{m_{i2}}
(\theta^\top f_\phi(x_{ii'}) - \theta^\top y_{ij'})^2 \gamma_{i'j'}$
and obtain the gradient:
\begin{multline}
  \nabla_\phi \left(\min_{\gamma \in \Gamma\left(\frac{\mathbf{1}}{m_{i1}},
      \frac{\mathbf{1}}{m_{i2}}\right)} \sum_{i'=1}^{m_{i1}} \sum_{j'=1}^{m_{i2}}
    (\theta^\top f_\phi(x_{ii'}) - \theta^\top y_{ij'})^2
    \gamma_{i'j'}\right)   \\
  = \nabla_\phi \left(\sum_{i'=1}^{m_{i1}} \sum_{j'=1}^{m_{i2}}
    (\theta^\top f_\phi(x_{ii'}) - \theta^\top y_{ij'})^2
    \gamma_{i'j'}^\star \right)  
  = \sum_{i'=1}^{m_{i1}} \sum_{j'=1}^{m_{i2}} \gamma_{i'j'}^\star
  \nabla_\phi (\theta^\top f_\phi(x_{ii'}) - \theta^\top y_{ij'})^2  \\
  =2 \sum_{i'=1}^{m_{i1}} \sum_{j'=1}^{m_{i2}}
  \gamma_{i'j'}^\star(\theta^\top f_\phi(x_{ii'}) - \theta^\top
  y_{ij'}) \nabla_\phi (\theta^\top f_\phi(x_{ii'})).
  \label{eq:gradient}
\end{multline}
The calculation of $\nabla_\phi (\theta^\top f_\phi(x_{ii'}))$ depends
on the form of $f_\phi$.

In the upcoming examples and simulations,
we utilize a linear function for simplicity, $f_{A,b}(x) = Ax + b$,
 i.e., $\phi=(A,b)$ with 
$A=[A_{ij}]\in \Re^{d_2\times d_1}$ and $b\in
 \Re^{d_2}$. It is easy to check that this function
   satisfies Assumption~\ref{assumption:continuity}, i.e., it is
   Lipchitz in $\phi$ (see
   Proposition~\ref{proposition:linear_function} in
   Appendix~\ref{subsec:proof:theoremm:posterior_consistency}). While
   any  alternative  functions can be used e.g., deep neural
   network~\citep{wilson2020bayesian}, we choose a linear form as 
   the most parsimonious model,
\ and recommend more complex choices only if residual errors so
indicate.    
We put a horseshoe
prior~\citep{carvalho2009handling} on $A$, which will be later discussed in detail in the context of graph discovery:
\begin{align} \label{eq:HS}
    A_{ij} \mid \lambda_{ij}, \tau &\sim \mathcal{N}(0, \lambda_{ij}^2
                                     \tau^2), \quad 
    \lambda_{ij} \sim C^+(0, 1), \quad 
    \tau \sim C^+(0, 1),
\end{align}
and $b \sim \mathcal{N}(0, I)$. 
We implement posterior simulation using
 standard MCMC transition probabilities using the inverse gamma
complete conditional posterior distributions for
$\lambda_{ij}^2$,
$\nu_{ij}$,
$\tau^2$,
and
$\zeta$.  See Supplementary Materials ~\ref{sec:MCMC} for more details of the MCMC sampler.
To sample from $p(A, b \mid \lambda^2, \tau^2, \mathcal{S})$, we use
the MALA algorithm as discussed above.  Substituting  in $\nabla_A
(\theta^\top (Ax_{ii'} + b)) = \theta x_{ii'}^\top$ and $\nabla_b
(\theta^\top (Ax_{ii'} + b)) = \theta$ into
Equation~\eqref{eq:gradient}, we obtain the gradient for the
log-likelihood. For the gradient of the log-prior, we have
$\nabla_{A_{ij}}\log p(A_{ij} \mid \lambda_{ij}^2, \tau^2) =
\nabla_{A_{ij}} \frac{-A_{ij}^2}{2 \lambda_{ij}^2 \tau^2} =
\frac{-A_{ij}}{\lambda_{ij}^2 \tau^2}$ and $\nabla_b \log p(b) =
\nabla_b \frac{-\|b\|_2^2}{2} = -b$.

\subsection{  Simulation}
\label{subsec:simulation}
\begin{figure}[!t]
\begin{center}
    \begin{tabular}{cc}
  \widgraph{0.5\textwidth}{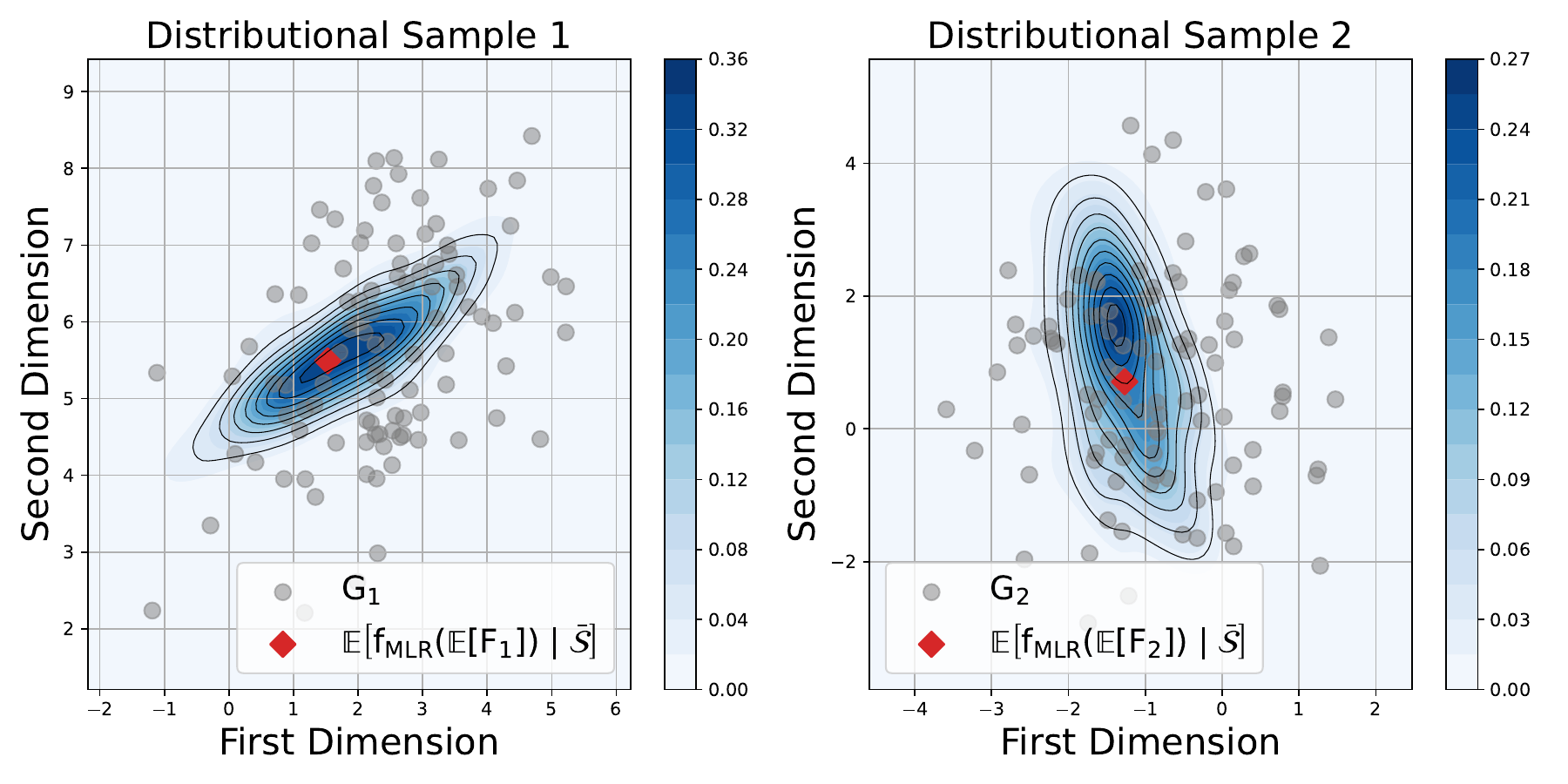} &
  \widgraph{0.5\textwidth}{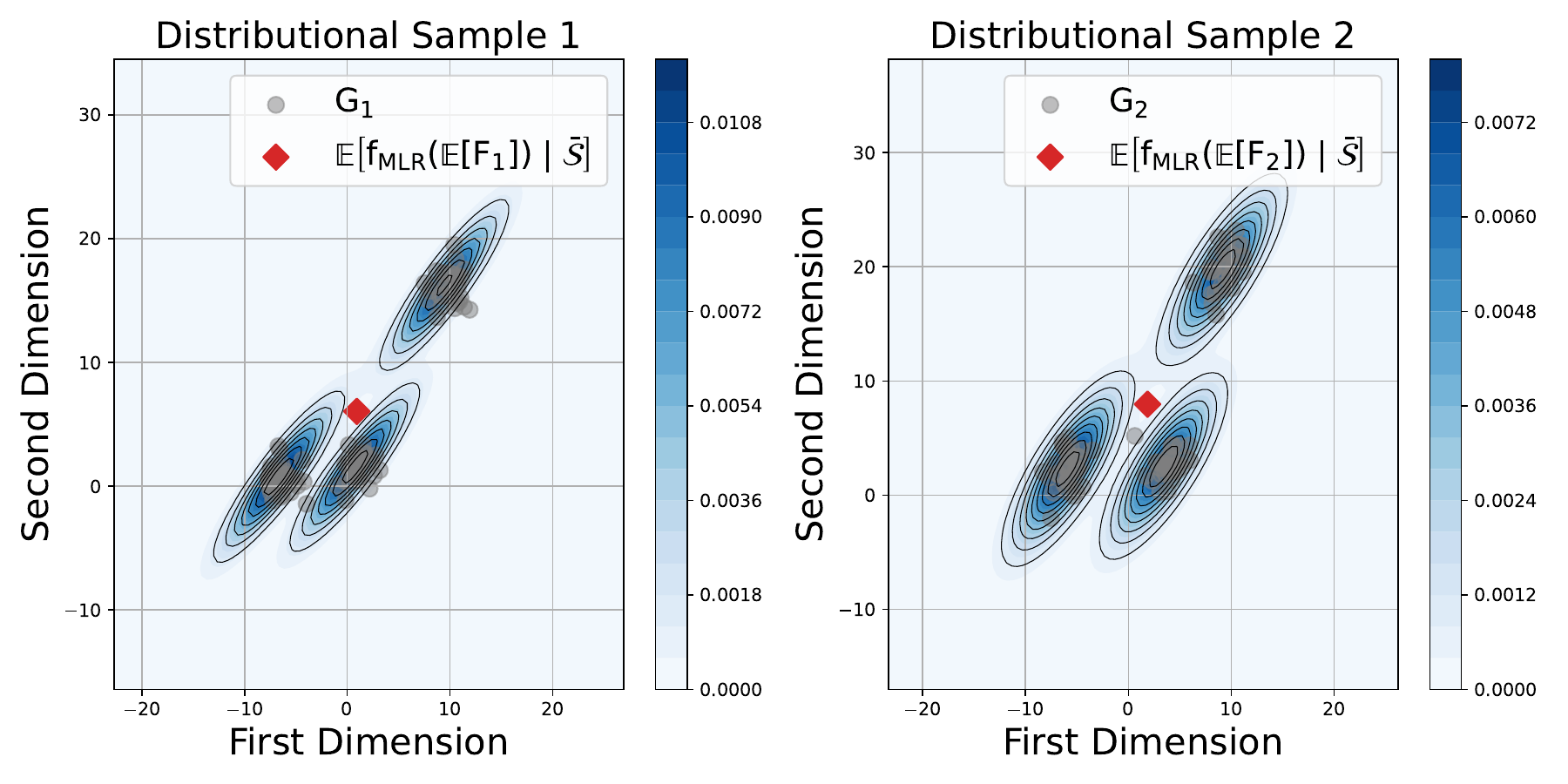}  \\
   \widgraph{0.5\textwidth}{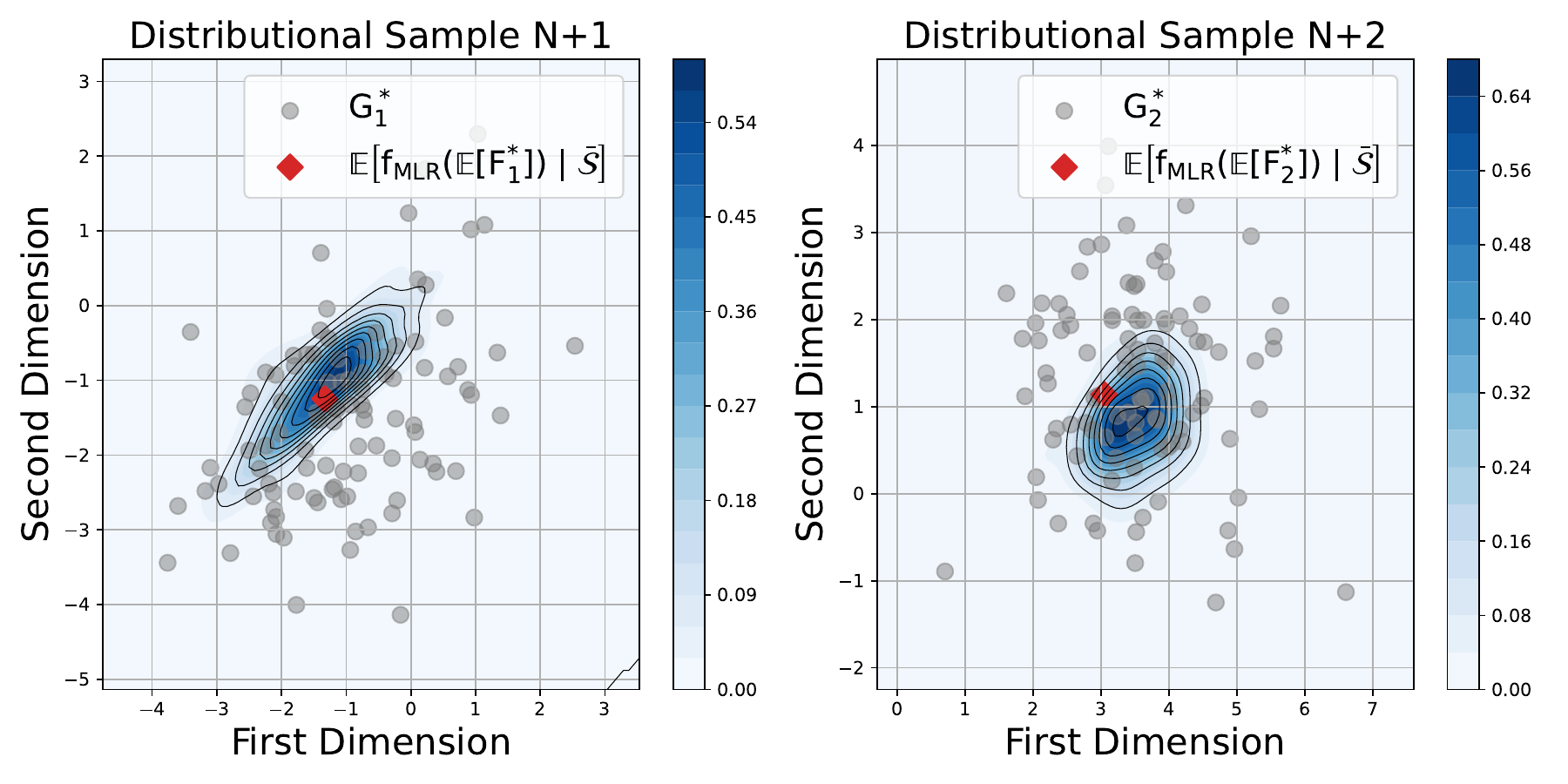} &
  \widgraph{0.5\textwidth}{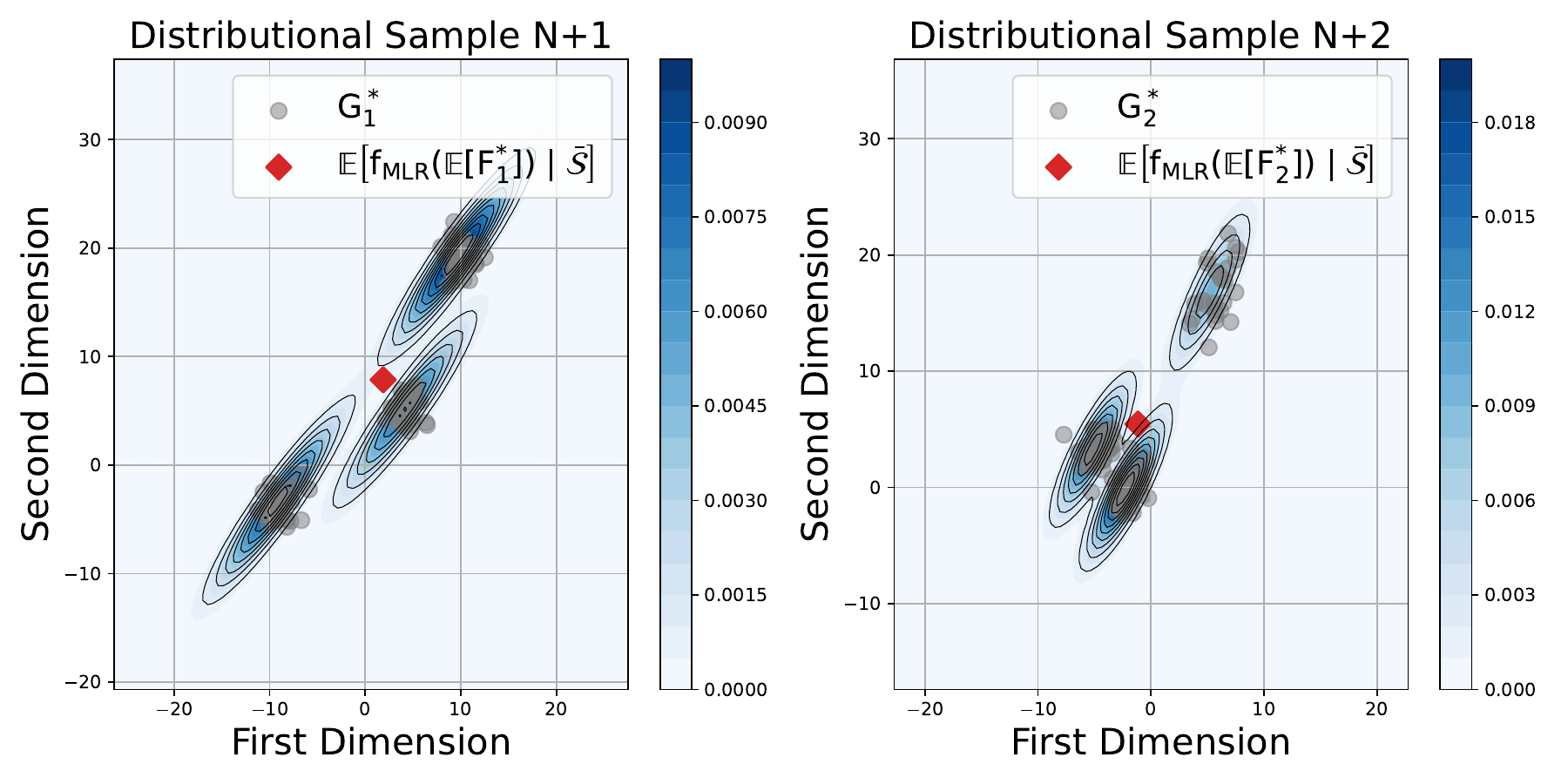} \\
 (a) \underline{Scenario 1} & (b) \underline{Scenario 2}
  \end{tabular}
  \end{center}
  \vspace{-0.2 in}
  \caption{
    \footnotesize{
      Fitted and predictive 
      densities $\Gbar_i$ under Bayesian DDR and
      expected fitted means under  Bayesian MLR.
      The first two columns shows results for two pairs of predictor
      and response distributions  under \underline{scenario 1}. 
    The last  two columns shows   the same 
    under \underline{scenario 2}. The first and second row show fitted
     (for $i=1,2$) and predicted (for  $i=N+1,N+2$)
     densities, respectively.
     Here and in all later figures showing fitted densities $\Gbar_i$, 
    the grey dots
    are the samples from the response distributions, the red diamonds
    show the fitted mean under Bayesian MLR, and the blue
    contours are fitted densities  under Bayesian DDR. 
}
} 
  \label{fig:prediction_simulation}
\end{figure}

\begin{figure}[!t]
  \begin{center}
    \begin{tabular}{cc}
      \widgraph{0.5\textwidth}{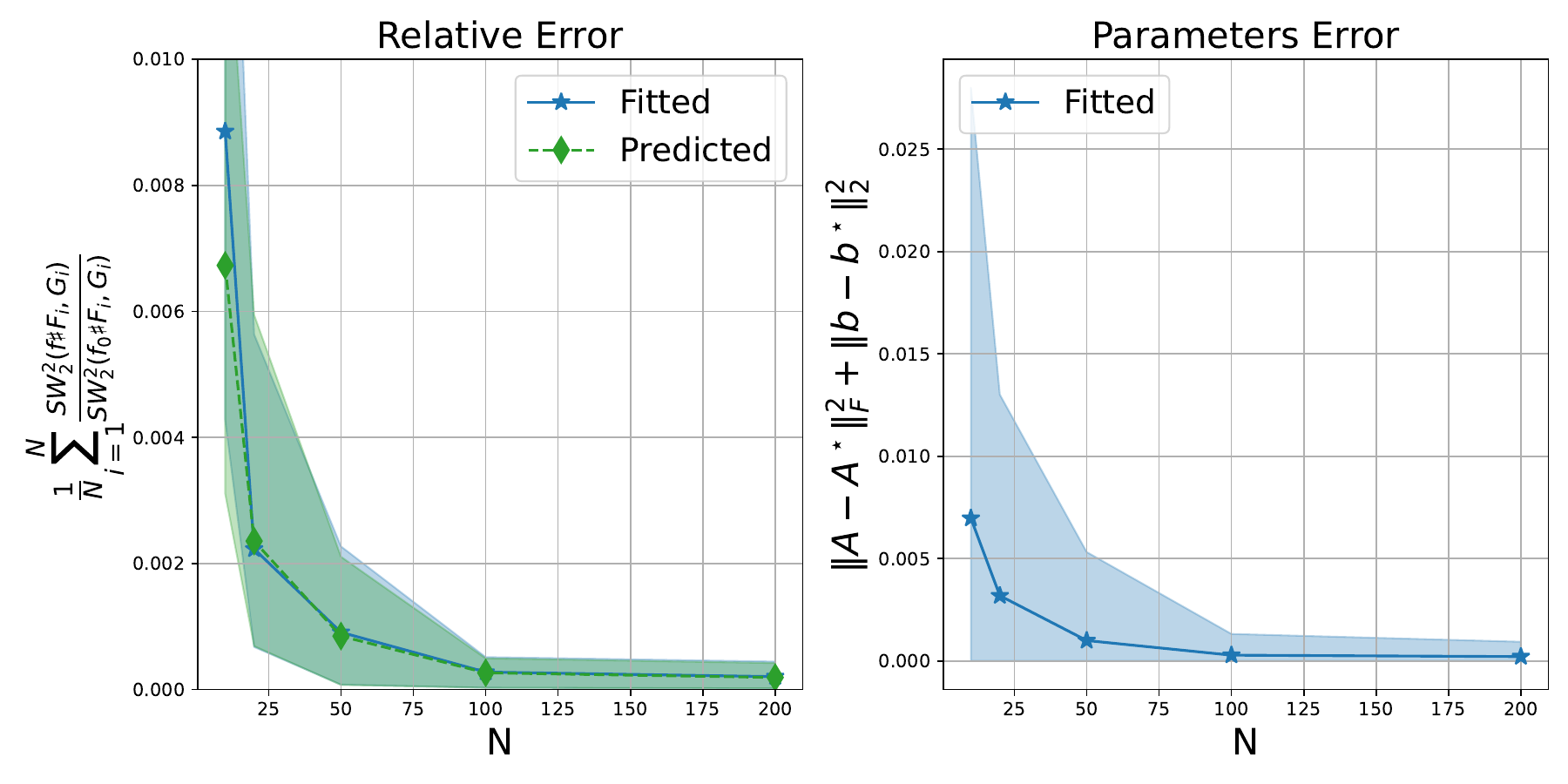}
      &
      \widgraph{0.5\textwidth}{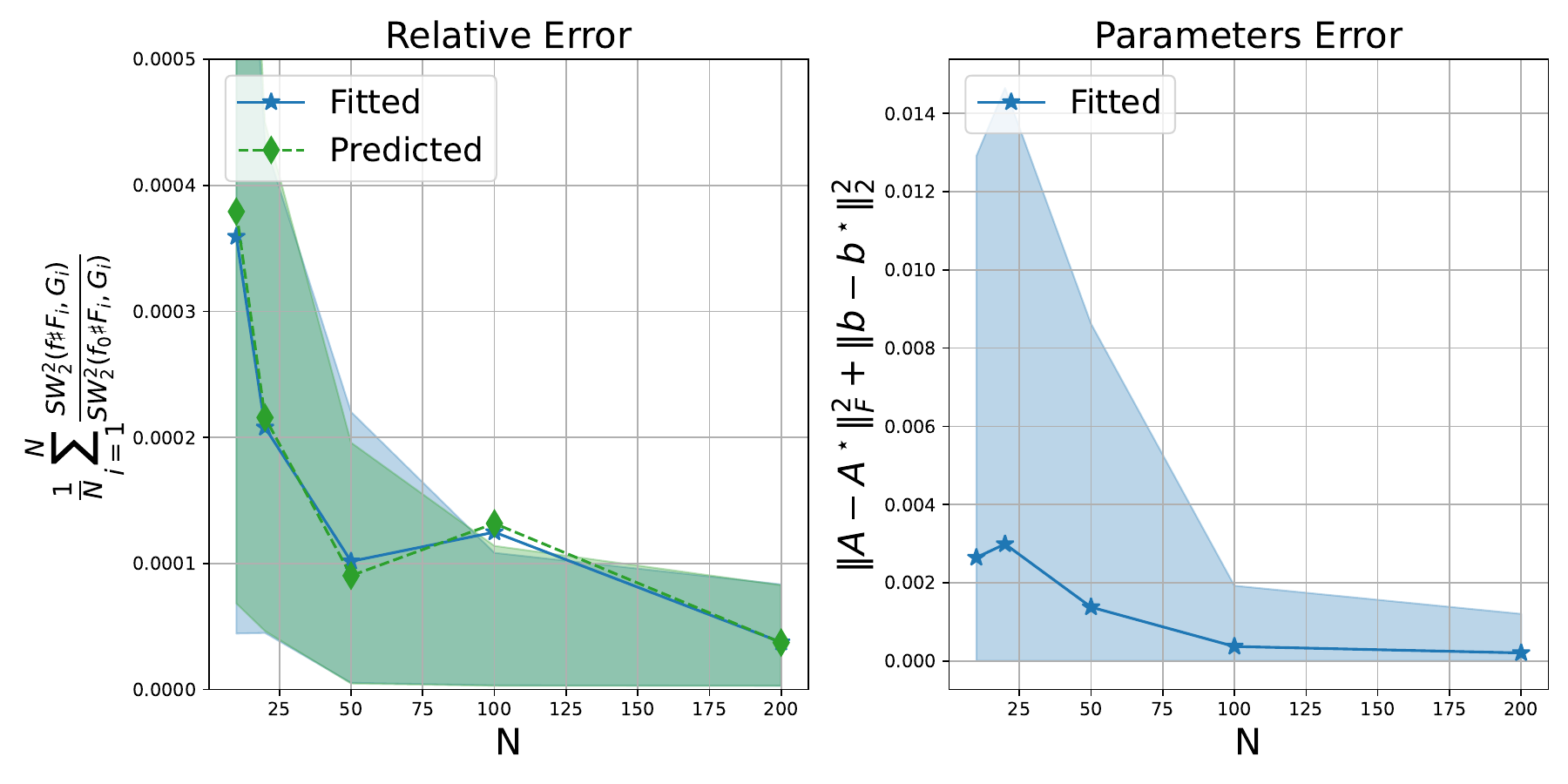}\\
       \hspace{1cm} RPE \hspace{2.5cm} $L(\phi,\phi^\star)$ &
       \hspace{1cm} RPE \hspace{2.5cm} $L(\phi,\phi^\star)$\\[4pt]
      \footnotesize{(a)
     Errors
      under \underline{scenario 1}.} &
      \footnotesize{(b)  Errors
      under \underline{scenario 2}.} 
    \end{tabular}
  \end{center}
  \vspace{-0.2 in}
  \caption{
    \footnotesize{Mean in-sample relative
      predictive errors~\eqref{eq:RPE} (fitted computed with $\mathcal{S}$ and predicted computed with $\mathcal{S}^*$),  and posterior mean squared error $\|A-A^\star\|_F^2 +\|b-b^\star \|_2^2$
      with corresponding 95\%   highest credible interval. 
    }}
\label{fig:error_simulation}
\end{figure}


We generate distributional samples $\SSS = \{(\Fh_1, \Gh_1),
\ldots, (\Fh_N, \Gh_N)\}$ with $N \in \{10, 20, 50, 100, 200\}$. In
particular,
letting
 $A^\star = \left[\begin{smallmatrix}
    1 & 0 \\
    1 & 1
\end{smallmatrix}\right]$ and $b^\star =  [1, 1]^\top$, 
we construct: 
$$
    \Fh_i = \frac{1}{m} \sum_{j=1}^{m} \delta_{x_{ij}}, \quad 
    \Gh_i = \frac{1}{m} \sum_{j=1}^{m} \delta_{A^\star x_{ij} +
            b^\star + \epsilon_0}
$$
where $\epsilon_0 \sim \mathcal{N}([0, 0]^\top, I_2)$, and $m =
100$.
 See below for the choice of $x_{ij}$. 
    We aim to examine whether the proposed model is able to
    recover the true parameters under a well-specified scenario.
For easier visualization we consider the bivariate setting.
We consider two
 simulation scenarios. 
In \underline{scenario 1},  we sample $x_{i1}, \ldots, x_{im}
\overset{i.i.d.}{\sim} \mathcal{N}(\mu_i, \Sigma_i)$ with $\mu_i \sim
\mathcal{N}([0, 0]', 4I_2)$ and $\Sigma_i \sim IW(I_2,6)$ (inverse
Wishart distribution with scale matrix $I_2$ and degrees of freedom 6)
for $i = 1, \ldots, N$. Under  \underline{scenario 2}, 
we sample $x_{i1}, \ldots,
x_{im} \overset{i.i.d.}{\sim} \frac{1}{3}\mathcal{N}(\mu_{i1},
\Sigma_{i1}) + \frac{1}{3}\mathcal{N}(\mu_{i2}, \Sigma_{i2}) +
\frac{1}{3}\mathcal{N}(\mu_{i3}, \Sigma_{i3})$, with $\mu_{i1} \sim
\mathcal{N}([0, 0]^\top, 4I_2)$, $\mu_{i2} \sim \mathcal{N}([8, 8]^\top,
4I_2)$, $\mu_{i3} \sim \mathcal{N}([-8, 8]^\top, 4I_2)$, and $\Sigma_{i1},
\Sigma_{i2}, \Sigma_{i3} \sim IW(I_2,6)$ for $i = 1, \ldots, N$.


We compare the proposed Bayesian DDR with
Bayesian multivariate linear regression (MLR) 
 replacing the distributions $\Gh_i$ and $\Fh_i$ by
sample-specific average gene expressions as one would observe in 
a pseudo-bulk gene expression dataset.
 That is, the data are  
$\SSb = \{(\xbar_i,\ybar_i)\}$ where $\xbar_i = \frac{1}{m_{i1}} \sum_{j=1}^{m_{i1}}
x_{ij}$ and $\ybar_i = \frac{1}{m_{i2}} \sum_{j=1}^{m_{i2}}
y_{ij}$. Let $\Xbar = (\xbar_1, \ldots, \xbar_N)$ and $\Ybar =
(\ybar_1, \ldots, \ybar_N)$.  The Bayesian MLR
model is 
\begin{align} 
    &\Ybar \mid \Xbar, A, \Sigma \sim \mathcal{MN}(\Xbar A, \Sigma, I), \quad  A \mid \Sigma \sim \mathcal{MN}(0, \Sigma, I), \quad \Sigma \sim IW(I, d_2).
\end{align}
where $\mathcal{MN}(M, U, V)$ denotes the matrix normal distribution
with location matrix $M$ and scale matrices $U$ and $V$, and the intercepts are included in $A$. 
Inference under the Bayesian MLR model is implemented by
 a Gibbs sampler alternating draws from the complete conditional
posterior distributions for $A$ and $\Sigma$: 
\begin{equation}
\Sigma \mid \Ybar, \Xbar, A \sim IW(V_N, \nu_N), \quad
A \mid  \Ybar, \Xbar, \Sigma \sim \mathcal{MN}(B_N, \Lambda_N^{-1}, \Sigma)
\label{MLR}
\end{equation}
with $V_N = I + (\Ybar - \Xbar A)^\top (\Ybar - \Xbar A) + (A - A_0)^\top (A - A_0)$,
$B_N = (\Xbar^\top \Xbar + I)^{-1}(\Xbar^\top \Ybar + A_0)$, 
$\Lambda_N = \Xbar^\top \Xbar + I$, and
$\nu_N = d_2 + N$.

For both, Bayesian DDR and Bayesian MLR, we  save  1000 MCMC
samples. In the MALA algorithm
 for posterior simulation under the Bayesian DDR, 
we use a step size $\eta = 10^{-5}$, and
 we use $w = 10$ in the generalized likelihood. 
We discard the first 500 samples
as burn-in samples. With the final $T = 500$ posterior samples
$\{A_1, \dots, A_T; b_1, \dots, b_T\}$ we evaluate distributional
predictions $\Gbar(z)$ for any predictor distribution $\Fh = \frac{1}{m}
\sum_{i=1}^m \delta_{x_i}$  as
$\Gbar(z) = \mathbb{E}[(f_{A,b} \sharp \Fh)(z) \mid \mathcal{S}]$.

 In the first row of Figure~\ref{fig:prediction_simulation} 
we visualize the fitted densities
(i.e., in-sample prediction)  $\Gbar$ under the Bayesian DDR  and
the fitted means under MLR for 2 randomly chosen distributional
samples  under both scenarios,  with $N = 10$.
 Here and in later figures we use kernel density estimation to
smooth the shown contours. 
 Inference under the Bayesian DDR reports 
distributions whereas inference under the Bayesian MLR yields a
single vector. For 
the mixture of Gaussians case  (\underline{scenario 2}), the
expected means under the Bayesian MLR are not even in a high
probability region.  

 For an assessment of inference as a function of sample size, we
evaluate 
the change in RPE and the change in posterior expected squared error
 $L(\phi, \phi^\star) = ||A-A^\star||^2_F + ||b-b^\star||^2_2$,
for parameters $\phi=(A,b)$ 
as functions of sample size $N$.  Results are shown  in
Figure~\ref{fig:error_simulation}. To compute RPE, we
 define the reference model $f_0$ in \eqref{eq:RPE}
as intercept model, with $A=0$, leaving only the intercept $b$. 
From Figure~\ref{fig:error_simulation}, we see that
both, RPE (in blue) and  $L(\phi,\phi^\star)$ 
decrease with $N$ as expected. 
From this simulation, we observe that the generalized
posterior under Bayesian DDR  already concentrates around the true
parameters for  practically  feasible sample sizes $N$.

Next, we consider out-of-sample prediction. We use a test
sample
$\mathcal{S}^*=\{(\Fh_{N+1}, \Gh_{N+1}),$\\$ \ldots, (\Fh_{N+N^*}, \Gh_{N+N^*})\}$,
which are generated under the same simulation truths (scenarios 1
and 2) as before. We use $N^* = 200$ (the large $N^\star$ is
needed for the evaluation of $\pRPE$).
The second row of Figure~\ref{fig:prediction_simulation} shows out-of-sample
predicted densities $\Gbar_i(z) = \mathbb{E}[(f_{A,b} \sharp
\Fh_i)(z) \mid \SSS]$ for two randomly chosen distributional samples
in the test set, conditional on a
training sample of size $N=10$, under both scenarios. 
Predictive inference under the Bayesian DDR matches the observed  
densities quite well, confirming the desired generalizability.
We further evaluate generalizability 
by evaluating out-of-sample RPE in
Figure~\ref{fig:error_simulation} in green color.
Again, we find decreasing out-of-sample RPE with increasing sample
size, as expected.

\section{ Example: Cell to Cell Communication}
\label{sec:DDR_B_to_T}

\begin{figure}
    \centering
    \includegraphics[width=0.5\linewidth]{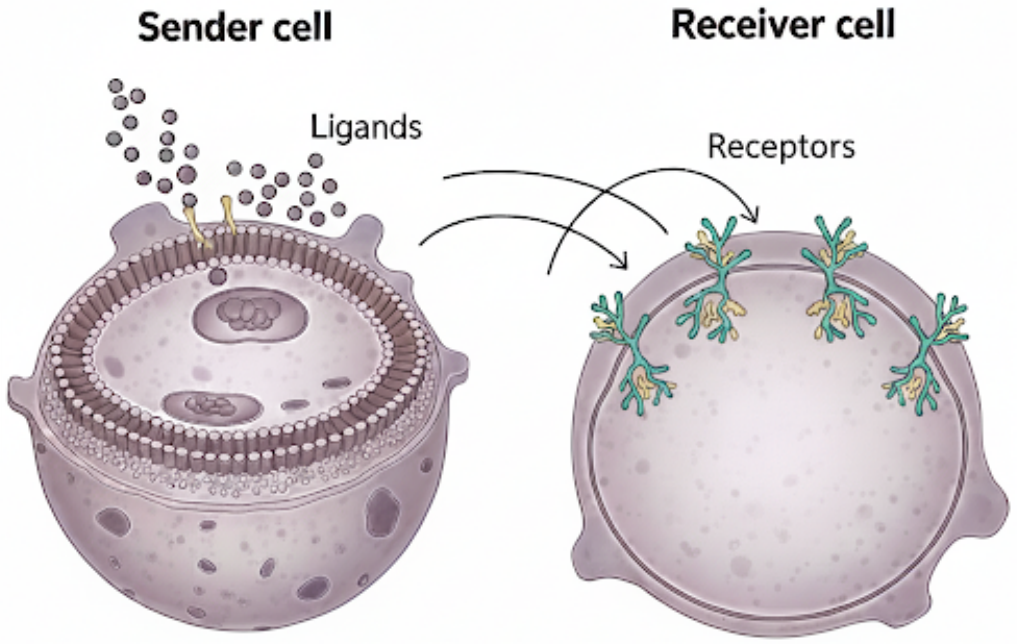}
    \caption{Cell-cell communication: the sender cell produces and sends ligands (signaling molecules) that bind to receptors expressed by the receiver cell.}
    \label{fig:cellcell}
\end{figure}

Cell-cell communication plays a fundamental role in understanding
complex biological systems, as interactions between cells drive
various physiological and pathological processes. See Figure \ref{fig:cellcell} for an illustration of a sender cell sending signals (ligands), which are received by a receiver cell via the expressed receptors. A natural way to
quantify these interactions is by modeling how the distribution of
ligands of  one cell type influences the distribution of
receptors of another cell type.  
We demonstrate the proposed generalized Bayesian DDR with an
application to inference on the communication between B cells and
T cells using the population-scale single-cell data, OneK1K
\citep{yazar2022single}. We select 439 donors who have at least 100
observed cells for both cell types.
 We formalize cell to cell communication as a regression of a
distribution of gene expressions of T cell receptors (as response) on
gene expression of B cell ligands (as predictor) to represent
``B cell to T cell'' communication, and vice versa for ``T cell to
B cell.''
Letting $\Gh_i$ denote the distribution of T cell receptor gene
expression for sample $i$ (as empirical distribution), and 
$\Fh_i$ the same for B cell ligands, the regression of T cell
on B cell gene expression becomes a density-density regression
as in \eqref{eq:DDRmodel} for the pairs
$(\Fh_i, \Gh_i)$. And similarly for T to B. 
For the regression of T cell gene expression on B cell gene expression
(B cells to T cells) the predictor is the distribution  $\Fh_i$
 of B cell ligands \{CD40, CD86, ICOSL, IL-6, BAFF, APRIL\} 
and the response is the distribution  $\Gh_i$  of T cell receptors
\{CD3D, CD3E, CD3G, TRBC1, TRBC2, CD28, ICOS, IL-2R, IFNGR, IL-21R,
PD-1, CTLA-4, CXCR5, CCR7\}. 
For T cells to B cells, the predictor is the distribution of
T cell ligands \{CD40L, IFN-$\gamma$ \}
and the response is the distribution of B cell receptors
\{IGHM, IGHD, IGHG, IGHA, CD40, ICOSL, IL-21R, IL-6R, BAFFR, CXCR5,
CCR7\}.
 The orthogonal nature of the ligand and receptor sets makes the
regressions of T on B and B on T complementary. 
For each task, we split the set of donors into a training set with 80\%
of the  subjects  and a test set with  the remaining
 20\%.
We standardize all 
ligands and receptors based on the mean and standard deviation from
all cells and all donors in the training set.

\begin{figure}[!t]
\begin{center}
    \begin{tabular}{c}
  \widgraph{1\textwidth}{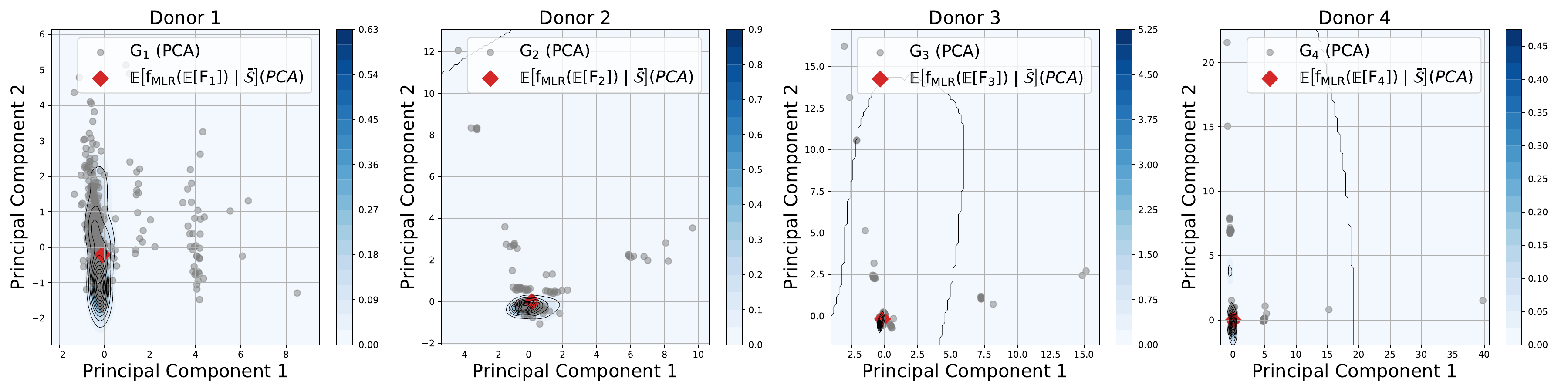} \\
\widgraph{1\textwidth}{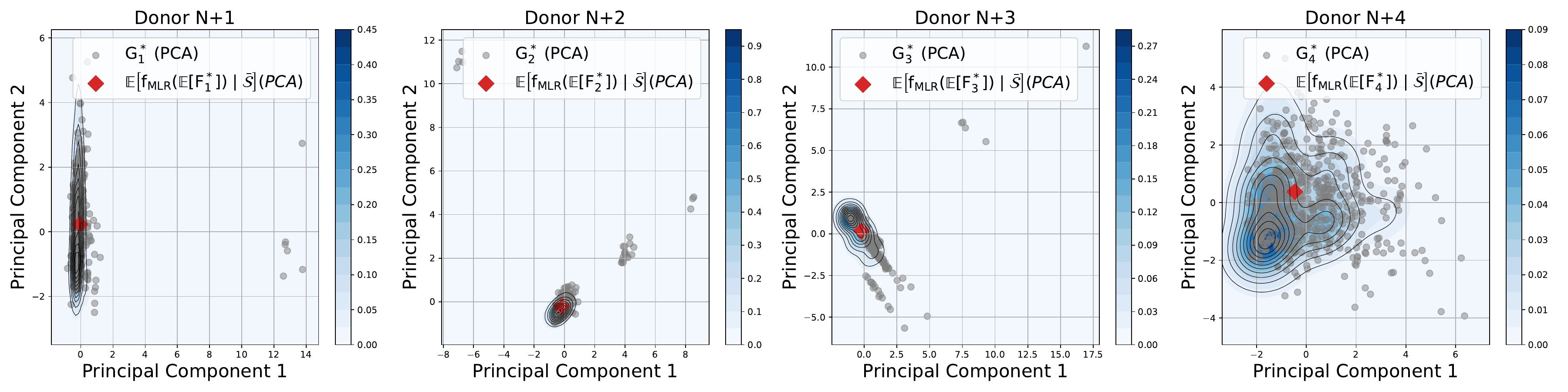} \\
  \end{tabular}
  \end{center}
  \vspace{-0.2 in}
  \caption{
    \footnotesize{Density-density regression of
      T cell gene expression on B cell gene expression under the
      Bayesian DDR model.
      The first row shows posterior fitted densities
       $\Gbar_i$ for four randomly chosen samples in the training
      set  (in-sample).
      The second row shows predictive densities  $\Gbar_i$ for four
      randomly chosen samples in the test set   (out-of-sample).
       Compare with Figure \ref{fig:prediction_simulation} for
      an explanation of the symbols and contours. 
      Here and in all later figures illustrating high dimensional gene
      expression features, 
      we use PCA on the atoms of the observed response distribution
      for each donor to plot the first two projections.
}
} 
  \label{fig:B_T_prediction}
\end{figure}


 We implement inference under the proposed Bayesian DDR and MLR
models. 
For the MALA algorithm, we use a step size $\eta = 10^{-4}$.
In the definition of the generalized log-likelihood, we use a factor 
$w = 100$, and we use $L = 1000$ for approximating the SW
distance. We obtain 1000 MCMC samples and discard the first 500
samples for both the Bayesian DDR and the Bayesian MLR.
Figure~\ref{fig:B_T_prediction} 
shows the fitted densities $\Gbar$ for 4 randomly selected donors in the
training set (first row of panels) and the predicted densities $\Gbar$
for 4 random donors in the test set (second row of panels) under the
proposed Bayesian DDR of T cell gene
expression on B cell gene expression. 
For comparison, we also indicate the fitted and predicted values for
 $\Ybar_i$  under the Bayesian MLR \eqref{MLR}. 
For the Bayesian DDR, we show
contour plots of fitted and predicted 
$\Gbar = \mathbb{E}[(f_{A,b} \sharp \Fh) \mid \Fh, \SSS]$ 
for predictor $\Fh_i=\Fh$.

For each donor,
we plot the first two projections under PCA of the atoms of the 
observed response distribution $\Gh_i$. 
As a density-density regression, Bayesian DDR 
reports an entire distribution, while
Bayesian MLR naturally only reports a single point estimate of
$\Ybar_i$.
We use RPEs as defined in \eqref{eq:RPE} to assess the fit and
prediction. 
For the regression of T cell gene expression on B cell
gene expression shown in Figure~\ref{fig:B_T_prediction},
under the Bayesian DDR we find posterior 
mean RPEs of $0.233$ with a 95\% credible
interval of $[0.232, 0.235]$
for the fitted densities in the training set, and 
$0.217~ (0.216 \text{ to } 0.219)$ for the predicted densities
for the test set.
For the regression of B cell gene expression on T cell gene
expression (not shown), we find the posterior mean
RPE of $0.617~ (0.616 \text{ to } 0.619)$ for
the training set and $0.606~ (0.603 \text{ to } 0.609)$ for the
test set.

In summary, T cell receptor expression can be
fitted well using $f_{A,b}\sharp \Fh_i$, via the linear mapping
$f_{A,b}$, while the distribution of B cell receptor expression
can be partially linearly explained by T cell ligand expression.


\section{ Application  to Graph Discovery}
\label{sec:graph_discovery}
We extend DDR to a setting with multiple pairs of
distributions. 
The goal is to discover a directed network or graph of
these distributions. A directed graph is a pair 
$\GG=(\VV,\EE)$ where $\VV = \{v_1, \ldots, v_V\}$ is a set of nodes
and $\EE \subseteq \EE_0 = \{(\ell, r) \mid \ell, r \in \VV, \ell \neq r\}$
is a set of directed 
edges. For each pair $e = (\ell, r) \in \EE_0$, 
we observe a sample of $N_e$ pairs of distributions $\SSS_{e} = \{(\Fh_{ei}, \Gh_{ei})\}_{i=1}^{N_{e}}$. 
We introduce two approaches 
to discover a directed graph of distributional nodes.
The first  option is to report  a fully connected graph with
edges labeled by weights generated from DDR.
The second approach
 casts the graph discovery as a decision problem, which we address
by way of a simple decision boundaries. 
In both cases, the graph is reported as a posterior summary based on
the generalized posterior \eqref{eq:pseudo_posterior}.


To define a {\bf weighted fully connected graphs} 
we start by setting up a Bayesian DDR model for each edge $ e \in \EE_0$.
Using inference under the DDR model, we evaluate 
the relative predictive error as defined in \eqref{eq:RPE2} with approximation of SW,
\begin{align}
    RPE_e(\mathcal{S}_e,L) = \frac{1}{N_e} \sum_{i=1}^{N_e}\frac{\widehat{SW}_2^2(f_{e} \sharp \Fh_{ei}, \Gh_{ei};L)}{\widehat{SW}_2^2(f_{e0} \sharp \Fh_{ei}, \Gh_{ei};L)},
\end{align}
which is estimated as 
\begin{align}
\label{eq:RPEedges}
    \pRPE_e(\mathcal{S}_e,L) = \frac{1}{T}\frac{1}{N_e} \sum_{t=1}^T  \sum_{i=1}^{N_e}\frac{\widehat{SW}_2^2(f_{et} \sharp \Fh_{ei}, \Gh_{ei};L)}{\widehat{SW}_2^2(f_{e0t} \sharp \Fh_{ei}, \Gh_{ei};L)},
\end{align}
using posterior samples
$f_{e1}, \ldots, f_{eT} \sim p(f_{e} \mid \SSS_e)$ and 
$f_{e01}, \ldots, f_{e0T} \sim p(f_{e0} \mid \SSS_{e})$
 under  a reference model $f_{e0}$. For example, if $f_{e}$ is a
linear mapping, as we shall use it later,
the reference model $f_{e0}$ could be the intercept-only model.
Given the posterior mean RPE for all pairs, we report a fully connected graph
with the estimated RPEs as weights for directed edges.

For an alternative approach, we define a
 {\bf sparse graph based on FDR and FNR control.} 
We augment the inference model to a decision problem by adding a loss
function.   As before,  we use  a linear mapping 
$f_{A,b}(x) = Ax + b$ with parameters $\phi=(A,b)$.
 Again,   we place a horseshoe prior on $A$, as in 
\eqref{eq:HS} before.
Let $A_e$ be
the linear coefficient matrix in $f_{A,b}$ that maps from
$\Fh_{ei}$ to $\Gh_{ei}$ for $i = 1, \ldots, N_{e}$.
We define the \textit{$\epsilon$-inclusion probability} ($\epsilon >
0$):
\begin{align}
\label{eq:eIP}
  \eIP = \mathbb{P}
  \Big(\overbrace{\max_{ij}\{|A_{eij}|\} > \eps}^{R_e} ~\Big|~
  \SSS_{e}\Big)\quad i=1,\ldots,d_{e1};\; j=1,\ldots,d_{e2}, 
\end{align}
where $d_{e1}$ and $d_{e2}$ are the dimensions of $\hat G_{ei}$ and
$\hat F_{ei}$.
The posterior probability \eqref{eq:eIP} 
can be estimated using the MCMC samples for $A_e$.
We interpret the  event $R_e$ 
as true  
presence of edge $e$, and 
use the posterior probability $\eIP$ to define a parametrized decision rule.
We include edge $e=(\ell,r)$ if $\eIP > 0.5$. that is, we use the rule 
$\dt_e = I(\eIP > 0.5)$.
The rule $\dt_e$ reduces the decision to the choice of a single
decision parameter $\eps$.
 A true positive happens if $d_e=1_{R_e} = 1$. 
We use a loss function based on posterior expected false
discovery rate ($\pFDR$) and false negative rate ($\pFNR$) to choose $\eps$.  In
short, for a fixed bound on $\pFDR$, we choose $\eps$ to
minimize $\pFNR$. 
The two error rates are defined as
\begin{equation}
\label{eq:FDR_FNR}
\pFDR = \frac{\sum_{e \in \EE_0} (1 - \eIP)\, \dt_e}
             {\sum_{e \in \EE_0} \dt_e + 0.001} \text{~ and~ }
\pFNR = \frac{\sum_{e \in \EE_0} \eIP\, (1 - \dt_e)}
             {|\EE_0| - \sum_{e \in \EE_0} \dt_e + 0.001},
\end{equation}
where $|\EE_0|$ is the total number of  possible edges (and
$0.001$ in the denominator avoids $0$-denominators).
See, for example,  \cite{muller2004optimal} for a discussion of
$\pFDR$ and $\pFNR$.
We choose $\eps$ to minimize $\pFNR$ subject to $\pFDR \le 0.10$.

\section{ Example: A Cell-Cell Communication Graph}
\label{sec:cell_graph_exp}

\subsection{ Results for the OneK1K single-cell data}
\label{subsec:cell-cell}
\begin{table}[!t]
    \centering
    \scalebox{0.7}{
    \begin{tabular}{|l|l|p{5cm}|p{9cm}|}
    \toprule
         Ordered Pairs& N& Ligands $\Fh_i$&Receptors $\Gh_i$ \\
         \midrule
         T Cells to B Cells& 439&  CD40L, IFN-$\gamma$ & IGHM, IGHD, IGHG, IGHA, CD40, ICOSL, IL-21R, IL-6R, BAFFR, CXCR5, CCR7 \\
         \midrule
         B Cells to T Cells& 439& CD40, CD86, ICOSL, IL-6, BAFF, APRIL & CD3D, CD3E, CD3G, TRBC1, TRBC2, CD28, ICOS, IL-2R, IFNGR, IL-21R, PD-1, CTLA-4, CXCR5, CCR7\\
         \midrule
T Cells to monocytes& 130& CD40L, IFN-$\gamma$, CD70 & CD86, IL-12R, TNFR1, TNFR2, IL-6R, CXCR4, PD-1, LAG-3\\
\midrule
monocytes to T Cells&130&  CD86, ICOSL, IL-23, TNF-$\alpha$, IL-1$\beta$, IL-6, CD70 & CD3D, CD3E, CD3G, TRBC1, TRBC2, CD28, ICOS, IL-2R, IFNGR, IL-21R, PD-1, CTLA-4, CXCR3, CXCR4\\
\midrule
B Cells to monocytes& 59& CD40L, CD70, IL-6, TNF-$\alpha$, Tim-3, BAFF, APRIL & CD40, CD86, IL-12R, TNFR1, TNFR2, CXCR3, PD-1, LAG-3\\
\midrule
monocytes to B Cells&59&  CD40L, IL-6, TNF-$\alpha$, Tim-3, BAFF, APRIL & CD40, IL-6R, IL-10R, BAFF-R, CXCR4, CCR7, PD-1\\
\midrule
monocytes to NK Cells&102&  CD40L, TNF-$\alpha$, IL-15, IL-18, Tim-3 & NKG2D, NKp30, NKp46, DNAM-1, NKG2A, IL-2R, IL-15R\\
\midrule
NK Cells to monocytes&102&  IL-15, IFN-$\gamma$, FasL & IL-12R, IL-15R, IFN-$\gamma$R, CXCR3, Fas, CD40, PD-1, Tim-3\\
\midrule
NK Cells to T Cells&700&  IL-15, IFN-$\gamma$, CD40L, FasL & IL-12R, IL-15R, IFN-$\gamma$R, CD28, ICOS, PD-1, Tim-3\\
\midrule
T Cells to NK Cells&700&  IFN-$\gamma$, CD86 & IL-2R, IL-4R, IL-21R, IFN-$\gamma$R, CCR1, CXCR3, CD28, ICOS, PD-1, Tim-3, NKG2D\\
\midrule
B Cells to NK Cells& 319& IL-6, BAFF, TNF-$\alpha$, CD40L & IL-6R, BAFF-R, TNFR, IL-10R, CCR1, CXCR3, CD40, ICOS, PD-1, Tim-3, NKG2D\\
\midrule
NK Cells to B Cells&319& IL-15, IFN-$\gamma$, CXCL8, CD40L & IL-2R, IL-15R, IFN-$\gamma$R, CCR1, CXCR2, CD40, ICOS, PD-1, Tim-3\\
\bottomrule
    \end{tabular}
    }
    \vskip 0.1in
    \caption{Summary of a dataset of 4 cell
      types: B cells, T cells, and NK cells, and monocytes. For each
      pair the table lists the number of
      donors and the list of ligands and receptors. 
}
    \label{tab:data}
\end{table}

We consider four cell types from the OneK1K dataset: B cells, T cells,
NK cells, and monocytes, and aim to infer
a graph to represent cell-cell communication between these
four subtypes. 
For each ordered pair $e$ of cell types, we extract the ligand genes from
the “sender” cell and the receptor genes from the “receiver” cell
and regress the distribution $G_{ei}$ of gene expression for 
the latter onto the distribution $F_{ei}$ of the former.
The data are the observed distributions of single-cell RNA-seq
gene expression for donors $i = 1, \ldots, N$. 
We select only donors who have at least 100 cells.
Table~\ref{tab:data} shows 
the number of donors $N$ for each ordered pair of cell types and the list
of ligands and receptors. We standardize all
ligands and receptors across all cells and donors.
As before, we implement
Bayesian DDR using a linear function $f_{A,b}(x) = Ax + b$ and
Bayesian MLR.
Under both models, we 
generate 1000 MCMC samples with the first 500 samples discarded as
burn-in. We use a step size of $\eta = 10^{-4}$ for the MALA
algorithm, and a scale factor of $w = 100$ in the generalized log-likelihood function.


\begin{figure}[!t]
\begin{center}
\begin{tabular}{c}
\widgraph{0.5\textwidth}{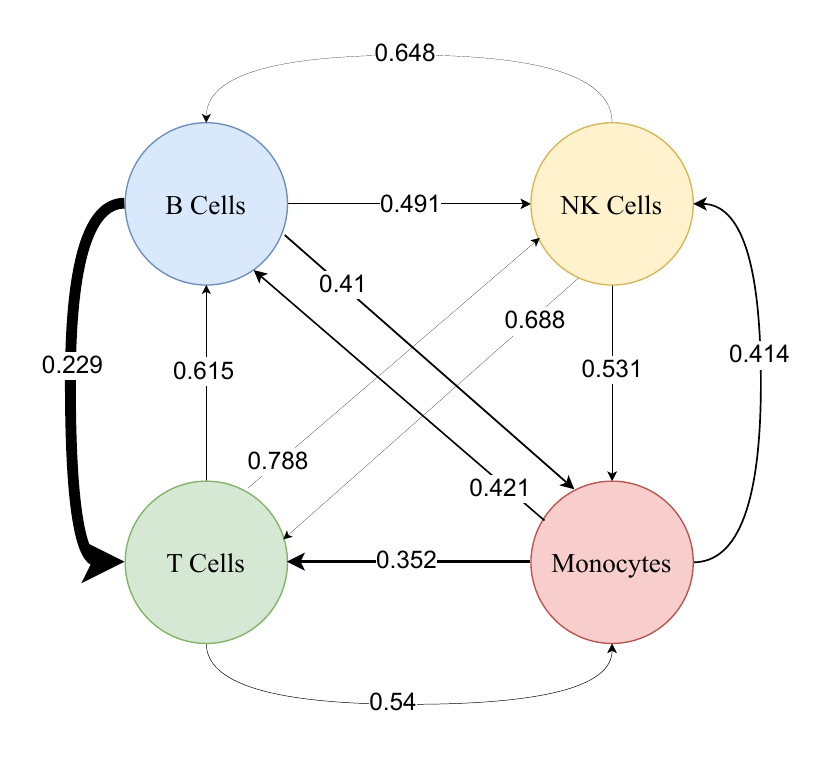} \\
(a) Graph discovered by RPEs~\eqref{eq:RPE}
\end{tabular}
    \begin{tabular}{cc}
 \widgraph{0.4\textwidth}{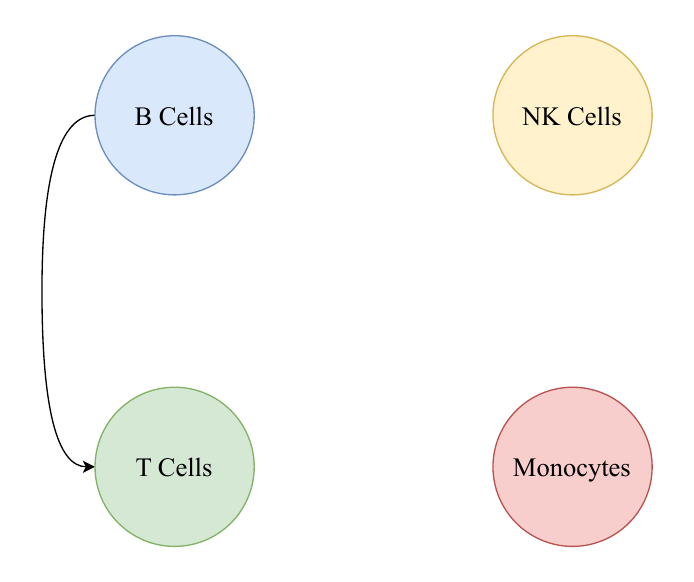} & \widgraph{0.45\textwidth}{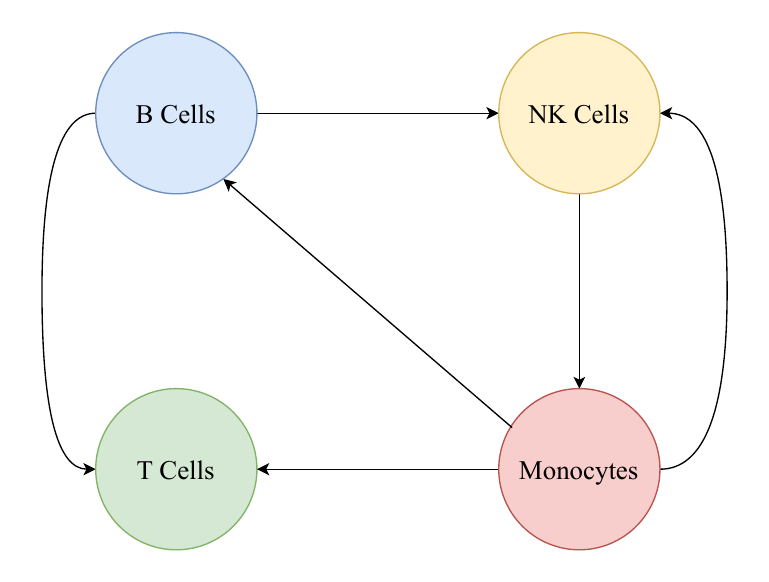} \\
\footnotesize{(b) Graph discovered by MLR} & \footnotesize{(c) Graph discovered by DDR}
  \end{tabular}
  \end{center}
  \vspace{-0.15 in}
  \caption{
  \footnotesize{ (a) Weighted graph using mean RPE under the Bayesian
    DDR model.  
    Edge thickness is controlled by an exponential kernel of 
    mean RPE,   showing thicker edges for lower mean RPE.  Estimated graphs under  (b) Bayesian MLR
    and, (c) Bayesian DDR.  Both are based on FDR and FNR
    control. 
}
} 
  \label{fig:FDR_graph}
\end{figure}




 As before, 
for visualization we use PCA on the response distribution of each
donor to obtain the projection matrix in two dimensions.  
Figures~\ref{fig:prediction_1}-\ref{fig:prediction_3} in Supplementary
Materials~\ref{subsec:graph_discovery_appendix} shows fitted densities $\Gbar$ under Bayesian DDR
for  ordered pairs of cell types.
 For comparison, as before, we also include  fitted means
under Bayesian MLR.
Bayesian DDR reports
entire fitted densities $\Gbar$. In contrast,
Bayesian MLR provides a single fitted value for $\ybar$,
stopping short  of representing the entire distributional information
of the response. 

Figure~\ref{fig:FDR_graph}a
shows an estimated graph of cell-cell communication, 
with edge weights based on mean RPE.
Lower mean RPE means a higher level of association. We use an
exponential kernel based on mean RPE to control the thickness of the edges. We refer the reader to  Figure~\ref{fig:RPEs} in Supplemental Materials~\ref{subsec:graph_discovery_appendix} for boxplots of RPE for all ordered pairs
of cell types. 
 For an alternative graph estimate, we
show in Figure~\ref{fig:FDR_graph}b and Figure~\ref{fig:FDR_graph}c
the graph based on FDR and FNR control as discussed before. 
 For comparison, we also report an estimated graph based on 
inference under Bayesian MLR (using the same FDR and FNR control).
For Bayesian MLR, we found the optimal threshold $\epsilon = 0.858$,
yielding posterior expected FDR and FNR of $(\pFDR, \pFNR) = (0.1, 0.039)$. 
Under the Bayesian DDR, the optimal threshold is $\epsilon = 0.974$,
corresponding to $\pFDR = 0.1, \pFNR = 0.108$.
The graph based on Bayesian MLR includes only one edge from B
to T cells. Under the Bayesian DDR model we find 6 edges: B to T
cells, B to NK cells, monocytes to T cells, monocytes to NK
cells, NK cells to monocytes, and monocytes to B cells. Comparing with
the mean RPE weights in Figure~\ref{fig:FDR_graph}a, we note that the
estimated graph under FDR and FNR control seems to select edges with
the mean RPE smaller than $0.5$.




 Besides facilitating the graph construction,
another benefit of choosing the linear function $f_{A,b}$ is the
opportunity to explore the 
interaction between ligands and receptors through
 inference on the corresponding coefficients in $A$. 
Figures~\ref{fig:T_to_B_coeff}-\ref{fig:NK_to_B_coeff} in
Supplementary Materials~\ref{subsec:graph_discovery_appendix}
 show boxplots for the linear regression coefficients $A$ for
all ordered pairs of cell types. 
From these summaries, one could 
determine which ligands are significant in
predicting specific receptors. We  discuss more details
in the Supplementary Materials~\ref{subsec:graph_discovery_appendix}. 

\subsection{Semi-Simulation}
\label{subsec:semi_simulation}
We  validate the proposed  graph discovery framework
 by setting up a simulation using simulation truths inspired by
the  OneK1K data. Particularly,  we use the real predictor
distributions taken from the data and a simulation  truth for
$f_{A,b}$ 
estimated from the real data  to set up a simulation truth for the
response distributions.
 Assume then that 
we have a posterior Monte Carlo sample with $T$
posterior samples under the Bayesian DDR model 
for all pairs of cell types.
 Recall that $\dt_e \in \{0,1\}$ is an indicator for including
edge $e$ in the inferred graph under Bayesian DDR
(Figure~\ref{fig:FDR_graph}c). 
We define the following empirical distributions: 
$$
    p_1 \propto\sum_{e\in \EE_0} \sum_{t=1}^T
      \sum_{i=1}^{d_{e1}}\sum_{j=1}^{d_{e2}}\delta_{A_{eijt}}
      I(\dt_e=1) \text{ and }
    p_0 \propto \sum_{e\in \EE_0} \sum_{t=1}^T
      \sum_{i=1}^{d_{e1}}\sum_{j=1}^{d_{e2}}\delta_{A_{eijt}}
      I(\dt_e=0), 
$$
where $A_{eijt}$ is  the $(i,j)$ element of
$A_e$ in the  $t$-th MCMC sample for cell types $e=(\ell,r)$.
 That is,  $p_1$ is the empirical distribution
of coefficients for included edges in Figure
\ref{fig:FDR_graph}b, and $p_0$ is the  same for missing edges. 
We consider the following three scenarios: 
\begin{figure}[!t]
\begin{center}
\begin{tabular}{ccc}
 \widgraph{0.25\textwidth}{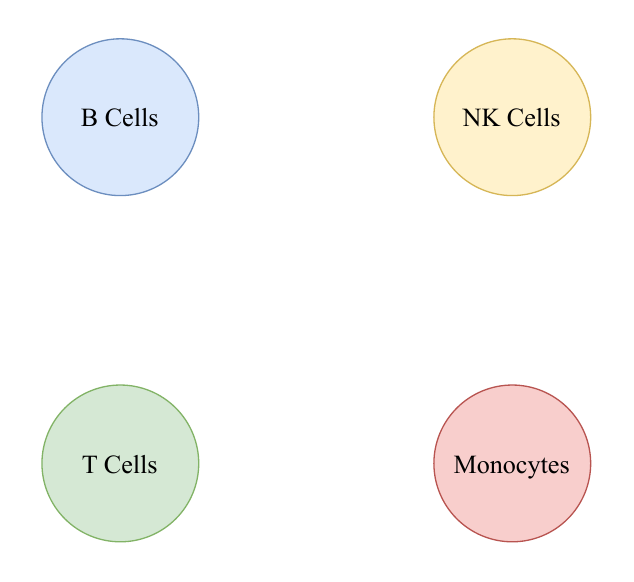}  & \widgraph{0.25\textwidth}{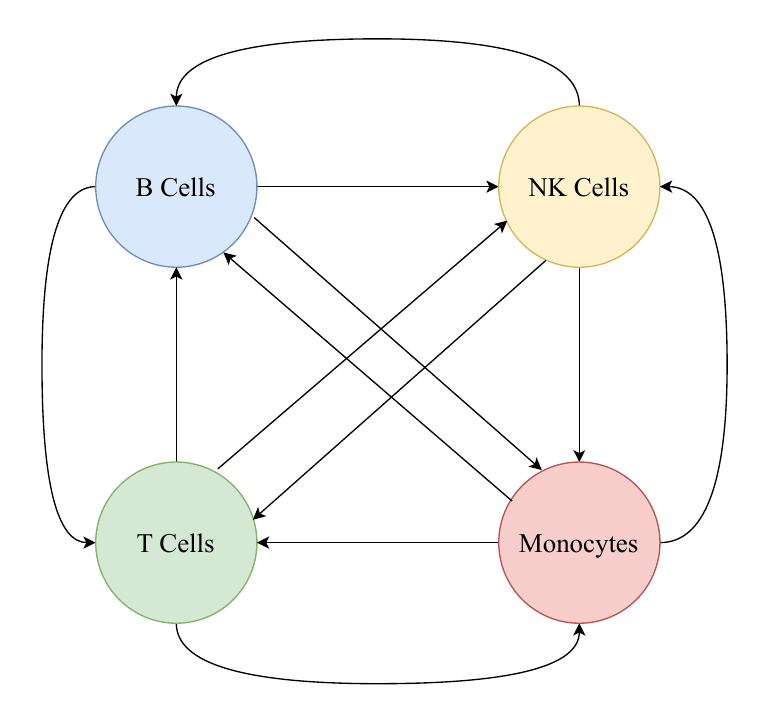} & \widgraph{0.25\textwidth}{figs/Yangsdata/Cell_Graph.pdf} \\
 & (a) Simulation Truth \\
 \widgraph{0.25\textwidth}{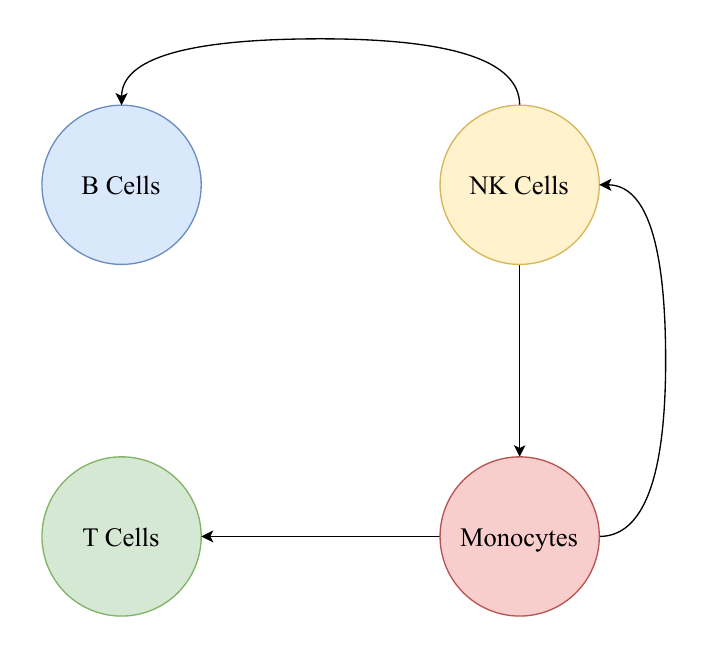} &
    \widgraph{0.25\textwidth}{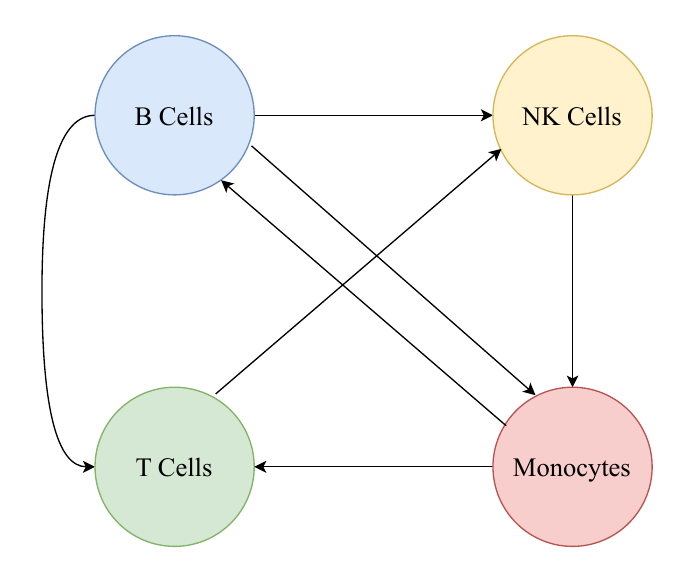} 
 & \widgraph{0.25\textwidth}{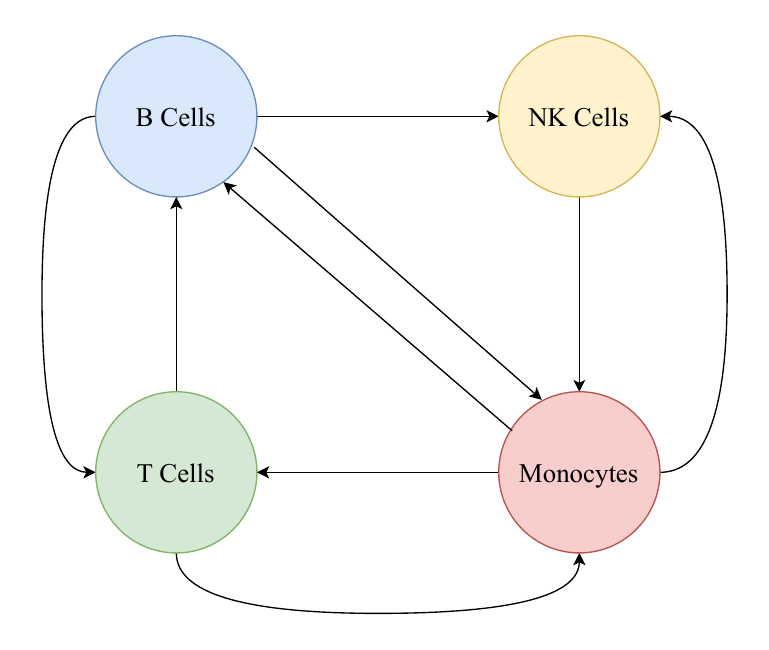}
\\
& (b) Estimated graphs  under DDR \\
\widgraph{0.25\textwidth}{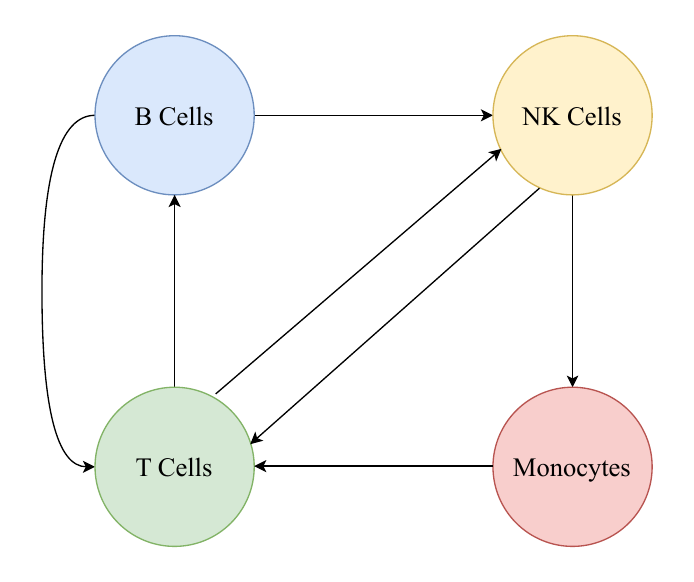}& \widgraph{0.25\textwidth}{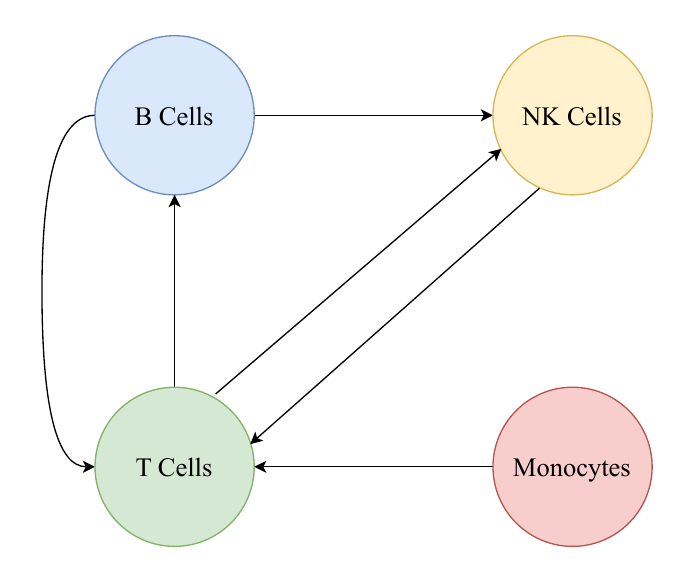}  & \widgraph{0.25\textwidth}{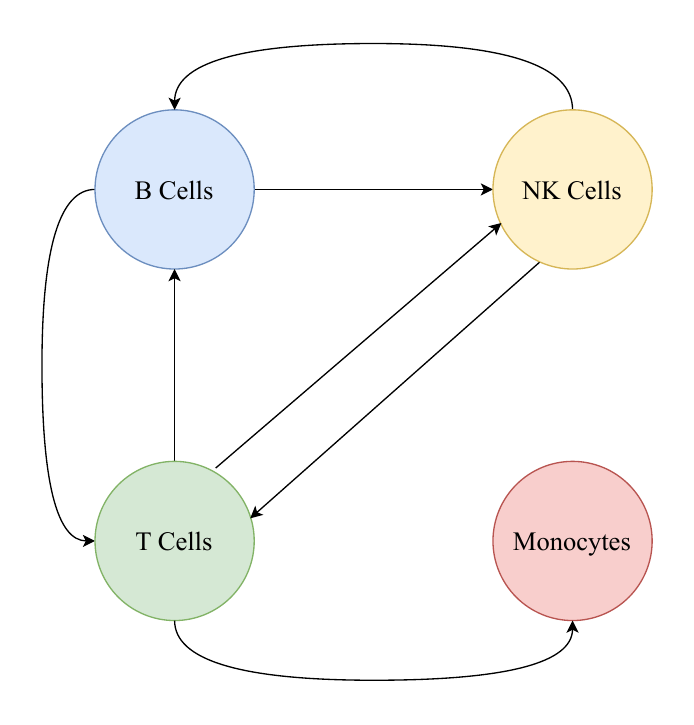}\\
& (c) Estimated graphs  under MLR  \\
 \underline{Scenario 1}&    \underline{Scenario 2} &    \underline{Scenario 3}

  \end{tabular}
  
  \end{center}
  \vspace{-0.15 in}
  \caption{
  \footnotesize{Simulated true graphs  (first row) and
    estimated graphs  under the Bayesian DDR model (second row) and the Bayesian MLR model (third row)  using FDR and
    FNR control  in the
    ``no edge''-\underline{scenario 1} (left column),
    ``full graph''-\underline{scenario 2} (center column), and ``sparse graph'' -\underline{scenario 3} (right column) scenarios.}}

  \label{fig:Semi_simulated_graph_null}
\end{figure}
Under \underline{scenario 1} (``no edge") we
sample  for each pair $e=(\ell, r)$ of cell types the
simulation truth 
$A_{eij}^\star\sim p_0$ and 
$b_{ej}^\star \sim p_0$ ($i=1,\ldots,
d_{e1},~j=1,\ldots,d_{e2}$).
Under \underline{scenario 2} (``full graph")
 we sample the simulation truth for all pairs $e=(\ell,r)$ as 
$A_{eij}^\star\sim p_1$ and
$b_{ej}^\star \sim p_1$ ($i=1,\ldots, d_{e1}),~j=1,\ldots,d_{e2}$).
Finally, under \underline{scenario 3} (``sparse graph") we use
the inferred graph under Bayesian DDR (Figure~\ref{fig:FDR_graph}b) as
the ground truth. For included edges $e$, we sample
 the simulation truth like under scenario 2. For excluded edges
we sample like under scenario 1. 
Under all three scenarios, after  generating the simulation truths
$(A_e,b_e)$  for the linear mappings, we simulate response
distributions   
$
\Gh_{ei} = \frac{1}{m_{ei}} \sum_{j=1}^{m_{ei}} \delta_{A^\star_e x_{ij} + b_e^\star + \epsilon_0}
$
where 
$x_{ij}$ is the observed ligands gene expression from the $i$-th
donnor in the OneK1K data, and $\epsilon_0 \sim \mathcal{N}(0, 0.01I)$
for $i=1, \ldots, N_{e}$.
 This concludes the definition of the simulation truth. 
We then implement inference under the Bayesian DDR and the Bayesian
MLR model,
with the same hyperparameters as in the earlier data analysis.
 Finally, we  use FDR and FNR control to report a graph. 

{\em Inference under Bayesian DDR.}
 Figure~\ref{fig:Semi_simulated_graph_null}
shows  the simulated true graphs and the reported graphs  under
the Bayesian DDR model for all three scenarios. 
Under \underline{scenario 1} (``no edge")
 we find the optimal threshold as $\epsilon=1$, 
yielding $(\pFDR, \pFNR) = (0.035, 0.064)$,
and 5 edges are  reported. 
Under \underline{scenario 2} (``full graph"), 
$\epsilon=1$ with $(\pFDR, \pFNR) = (0.083, 0.125)$, and 7
edges are discovered.
Finally, under \underline{scenario 3} (``sparse graph''),
$\epsilon= 1$, and $(\pFDR, \pFNR) =
(0.036, 0.17)$. In this setup, Bayesian DDR yields 9 edges, with 6
correct edges and 3 incorrect edges.  
%
%
 Boxplots of RPE for all ordered pairs of cell types,
and under all three scenarios
is shown in  
in Figure~\ref{fig:Semi_simulated_graph_null_RPEs} in Supplementary
Materials~\ref{subsec:semi_simulation_appendix}.
The relatively low RPE values show a good fit for almost all
ordered pairs of cell types across all three settings.
Additionally,
Figure~\ref{fig:coefficient_semi_simulation} in Supplementary
Materials~\ref{subsec:semi_simulation_appendix} shows
boxplots
 of 
squared errors for the parameters ($A_e,b_e$) 
for the three settings and all pairs of cell types. 
Overall we find  relatively small parameter errors in all cases.   




  

\begin{figure}[!t]
\begin{center}
    \begin{tabular}{c}
  \widgraph{1\textwidth}{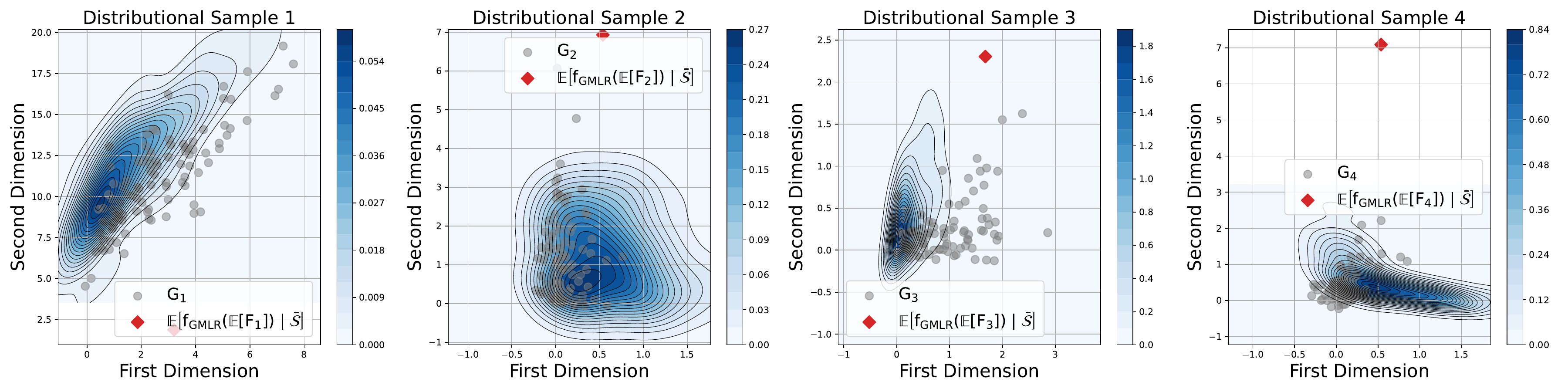} 
  \end{tabular}
  \end{center}
  \vspace{-0.2 in}
  \caption{
  \footnotesize{ Fitted densities $\Gbar$ under 
    Bayesian DDR with quadratic  mapping $f_{A,b}$ 
    fitted means under a Bayesian  multivariate nonlinear
    regression,  with quadratic mean function.
}
} 
  \label{fig:prediction_simulation_nonlinear}
\end{figure}

{\em Inference under Bayesian MLR:}
 Note that under all three scenarios the simulation truth is based
on a linear mapping of the atoms in $F_i$, implying 
a linear relationship also between
$\xbar_{i}=\frac{1}{m_{i1}}\sum_{j=1}^{m_{i1}} x_{ij}$ and
$\ybar_{i}=\frac{1}{m_{i2}}\sum_{j=1}^{m_{i2}} y_{ij}$.
 Those are the 
measurements reported by experimental platforms that record bulk gene
expression in samples.
 Naturally, this setting favors Bayesian MLR, 
which yields the following $(\pFDR, \pFNR)$ values: $(0.083,0.007)$
for \underline{scenario 1} (``no edge'');
$(0.007,0.008)$
for \underline{scenario 2} (``full graph"); and $(0.071,0.039)$
for \underline{scenario 3} (``sparse graph").
 The estimated graphs are shown in Figure 
\ref{fig:Semi_simulated_graph_null}. 
In all three cases, the FNR under DDR was higher, but remained within
practically useful bounds.

The previous simulation included a linear (in the atoms) simulation
truth, favoring MLR (since it also implies linearity of mean bulk gene
expression).
To highlight the limitations of linear regression for bulk gene
expression $(\xbar,\ybar)$ we consider another simulation, with: 
$
    \Fh_i = \frac{1}{m} \sum_{j=1}^{m} \delta_{x_{ij}}, \quad  \Gh_i = \frac{1}{m} \sum_{j=1}^{m} \delta_{ (A^\star x_{ij} +b^\star)^2 + \epsilon_0}, 
$
where $A^\star = \left[\begin{smallmatrix} 1 & 0 \\ 1 & 1 \end{smallmatrix}\right]$, $b^\star = [1,1]'$, $\epsilon_0 \sim \NN([0,0]',
0.01I)$, and $m=100$. We sample $x_{i1}, \ldots, x_{im}
\overset{i.i.d.}{\sim} \NN(\mu_i, \Sigma_i)$ with $\mu_i \sim
\NN([0,0]', I)$ and $\Sigma_i \sim IW(10, I)$ for $i = 1,
\ldots, N=10$. We implement Bayesian DDR with $f_{A,b}(x) =
(Ax + b)^2$ (element-wise quadratic), and  compare with
inference under  a Bayesian
multiple nonlinear regression (MNR) assuming $\ybar \mid \xbar,
A, \Sigma \sim \mathcal{MN}((\xbar A)^2, \Sigma, I)$ (element-wise
quadratic).
 Figure~\ref{fig:prediction_simulation_nonlinear} shows 
the fitted densities under  Bayesian DDR and fitted means under
the Bayesian MNR model. 
While inference under the Bayesian DDR can still often capture the
high-density regions,  
Bayesian MNR naturally gives inaccurate fits of the mean,
 due to model mis-specification. 
Similarly, in cases where $\mathbb{E}[X_i] = 0$
or $\mathbb{E}[X_i]$ does not exist
(e.g., Cauchy distribution), 
Bayesian MLR based on the mean statistic is not feasible, while
Bayesian DDR remains well-defined.



\section{Conclusion}
\label{sec:conclusion} e introduced Bayesian Density-Density
Regression (Bayesian DDR) as a novel framework for modeling complex
regression relationships between multivariate random
distributions.

 Some limitations remain.  First, the framework can suffer
from model misspecification due to the use of  the parametric
mapping $f_{\phi}$. However, in simulations and in the application we found
it sufficiently flexible.  The MALA sampling algorithm may not be
optimal for ensuring  a well mixing Markov chain, due to  the
non-convex nature of sliced
Wasserstein~\citep{tanguy2025properties}. Moreover, using the
conventional sliced Wasserstein distance for generalized likelihood
might not be suitable for all types of distributions with
non-Euclidean structures. For the graph discovery framework, the
$\epsilon$-inclusion probability is an ad-hoc choice, allowing
many alternatives. 

Future work will focus on addressing these limitations. For instance,
we may extend the framework to Bayesian nonparametric density-density
regression. In contrast to recent frequentist nonparametric
density-density regression
approaches~\citep{chen2023wasserstein,ghodrati2022distribution}, which
focus on one-dimensional regression, we  envision a Bayesian
nonparametric prior  for the regression function in multivariate
settings. One could  explore alternative geometric variants of the
SW distance for the generalized likelihood, and other
loss functions or policies for graph discovery. Finally, to deal with a more diverse set of distribution, Orlic-Wasserstein~\citep{sturm2011generalized} can be considered for the generalized likelihood.

\bigskip
\begin{center} {\large\bf SUPPLEMENTARY MATERIALS}
\end{center}

The Supplementary Materials include additional materials and visualizations, as
mentioned in the main text.

\bibliographystyle{apalike}

\bibliography{bib}

\newpage
\appendix
\begin{center} \Large Supplementary Materials for ``Bayesian
  Multivariate Density-Density Regression"
\end{center}

\section{Proofs}

\subsection{Proof of Theorem~\ref{theoremm:posterior_consistency}}
\label{subsec:proof:theoremm:posterior_consistency}

\begin{lemma}
\label{lemma:bound}
    For any $G_1 \in \mathcal{P}_2(\mathbb{R}^d)$ and $G_2 \in \mathcal{P}_2(\mathbb{R}^d)$, we have:
    \begin{align}
        SW_2(G_1,G_2)\leq W_2(G_1,G_2).
    \end{align}
\end{lemma}
\begin{proof}
From the definitions, we have:
    \begin{align}
    SW_2^2 (G_1,G_2) &= \mathbb{E}_{\theta \sim \mathcal{U}(\mathbb{S}^{d-1})}[W_2^2(\theta\sharp G_1,\theta \sharp G_2)] \\
    &=\mathbb{E}_{\theta \sim \mathcal{U}(\mathbb{S}^{d-1})}\left[\inf_{\pi \in \Pi(G_1,G_2)} \mathbb{E}_{(X,Y)\sim \pi}[ (\theta^\top X - \theta^\top Y)^2] \right] \\
    &=\mathbb{E}_{\theta \sim \mathcal{U}(\mathbb{S}^{d-1})}\left[\inf_{\pi \in \Pi(G_1,G_2)} \mathbb{E}_{(X,Y)\sim \pi}[ (\theta^\top ( X - Y))^2] \right] \\
    &\leq \mathbb{E}_{\theta \sim \mathcal{U}(\mathbb{S}^{d-1})}\left[\inf_{\pi \in \Pi(G_1,G_2)}  \mathbb{E}_{(X,Y)\sim \pi}[  \|\theta\|_2^2 \|X-Y\|_2^2 \right]   \quad \text{(Cauchy–Schwarz)}\\
    &= \mathbb{E}_{\theta \sim \mathcal{U}(\mathbb{S}^{d-1})}\left[\inf_{\pi \in \Pi(G_1,G_2)} \mathbb{E}_{(X,Y)\sim \pi}[\|X-Y\|_2^2 \right] \quad (\|\theta\|_2^2=1) \\
    &= W_2^2(G_1,G_2),
\end{align}
which completes the proof.
\end{proof}

    \begin{lemma}
\label{lemma:lipchitz}
Under Assumptions~\ref{assumption:secondmoment} and~\ref{assumption:continuity},  we have
\begin{align}
   \big| SW_2^2(f_\phi \sharp F, G) - SW_2^2(f_{\phi'} \sharp F, G) \big|
    \leq \gamma(\|\phi - \phi'\|_2),
\end{align}
for a  modulus of continuity $\gamma:\Re_+\to \Re_+$ such that $\lim_{t\to 0}\gamma(t)=0$.
\end{lemma}
\begin{proof}

From Lemma~\ref{lemma:bound}, we first bound the difference $| SW_2(f_\phi \sharp F, G) - SW_2(f_{\phi'} \sharp F, G) \big|$ for $\phi,\phi'\in \Phi$. Using the triangle inequality of $SW_2$, we have:
\begin{align}
\label{eq:SW_diff}
    &| SW_2(f_\phi \sharp F, G) - SW_2(f_{\phi'} \sharp F, G)|  \nonumber \\
    & \leq | SW_2(f_\phi \sharp F, f_{\phi'} \sharp F)+ SW_2(f_{\phi'} \sharp F, G) - SW_2(f_{\phi'} \sharp F, G)| \nonumber\\
    &= SW_2(f_\phi \sharp F, f_{\phi'} \sharp F) 
    \ \leq\ W_2(f_\phi \sharp F, f_{\phi'} \sharp F).
\end{align}
By definition of Wasserstein distance, 
\begin{align}
    W_2(f_\phi \sharp F, f_{\phi'} \sharp F)  &= \inf_{\pi \in \Pi(F,F)} \left( \mathbb{E}_{(X,Y)\sim \pi} \| f_\phi(X) - f_{\phi'}(Y) \|_2^2 \right)^{1/2}  \nonumber \\ 
    &\leq  \left( \mathbb{E}_{(X,Y)\sim (Id,Id)\sharp F} \| f_\phi(X) - f_{\phi'}(Y) \|_2^2 \right)^{1/2}  \nonumber\\
    &=\left( \mathbb{E}_{X\sim F} \| f_\phi(X) - f_{\phi'}(X) \|_2^2 \right)^{1/2} \nonumber \\
    &\leq \sqrt{\omega (\|\phi-\phi'\|_2^2)}
\end{align}
due to  Assumption~\ref{assumption:continuity}
Combining with~\eqref{eq:SW_diff}, we obtain
\begin{equation}
\label{eq:SW_Lip}
    \big| SW_2(f_\phi \sharp F, G) - SW_2(f_{\phi'} \sharp F, G) \big|
    \ \leq\ \sqrt{\omega (\|\phi-\phi'\|_2^2)}.
\end{equation}
We now extend~\eqref{eq:SW_Lip} to $SW_2^2$. By the factorization
\[
    \big| a^2 - b^2 \big| = |a-b|\, (a+b), 
\]
for $a>0,b>0$, 
\begin{align}
    &\big| SW_2^2(f_\phi \sharp F, G) - SW_2^2(f_{\phi'} \sharp F, G) \big| \nonumber \\
    &\leq  \big( SW_2(f_\phi \sharp F, G) - SW_2(f_{\phi'} \sharp F, G) \big) \big( SW_2(f_\phi \sharp F, G) + SW_2(f_{\phi'} \sharp F, G) \big) \nonumber\\
    &\quad \leq \sqrt{\omega (\|\phi-\phi'\|_2^2)} \, \big( SW_2(f_\phi \sharp F, G) + SW_2(f_{\phi'} \sharp F, G) \big).
\end{align}
Using $SW_2 \leq W_2$ and the standard bound
\begin{align}
    &W_2(f_\phi \sharp F, G) = \sqrt{\inf_{\pi \in \Pi(F,G} \mathbb{E}_{(X,Y) \sim \pi } [ \|f_\phi(X)-Y\|_2^2]}\\
    &\leq  \sqrt{\inf_{\pi \in \Pi(F,G)} \mathbb{E}_{(X,Y) \sim \pi } [ \|f_\phi(X)\|_2^2+\|Y\|_2^2]}\\
    &\leq \sqrt{\mathbb{E}_{X\sim F}\|f_\phi(X)\|_2^2} + \sqrt{\mathbb{E}_{Y\sim G}\|Y\|_2^2},
\end{align}
due to the inequality $\sqrt{a+b} \leq \sqrt{a} +\sqrt{b}$.
Together with  Assumption~\ref{assumption:secondmoment} and Assumption~\ref{assumption:continuity} yield
\begin{align}
    SW_2(f_\phi \sharp F, G) + SW_2(f_{\phi'} \sharp F, G)
    \leq \sqrt{C} + \sqrt{M_2}
\end{align}
Hence,
\begin{align}
    \big| SW_2^2(f_\phi \sharp F, G) - SW_2^2(f_{\phi'} \sharp F, G) \big|
    &\leq 2(\sqrt{C} + \sqrt{M_2}) \sqrt{\omega(\|\phi-\phi'\|_2^2)} \\& :=\gamma(\|\phi-\phi'\|_2),
\end{align}
where $\gamma(t) =2(\sqrt{C} + \sqrt{M_2})  \sqrt{\omega(t^2)}$. We can verify that $\lim_{t\to 0} \gamma(t)=0$. We complete the proof of the Lemma.
\end{proof}

\begin{lemma}[Uniform Law of Large Numbers]
\label{lemma:uniform_law}
Under assumptions~\ref{assumption:compact}, ~\ref{assumption:secondmoment}, and~\ref{assumption:continuity}, 
\begin{align}
    \sup_{\phi \in \Phi} |R_N(\phi) - R(\phi)| \overset{a.s}{\to} 0 \quad \text{as } N \to \infty.
\end{align}
\end{lemma}

\begin{proof}
Since $\Phi$ is compact by Assumption~\ref{assumption:compact}, there exists a $\delta$-net $\{\phi_1, \ldots, \phi_K\} \subset \Phi$ such that for any $\phi \in \Phi$, there exists $\phi^{(k)} \in \{\phi_1, \ldots, \phi_K\}$ with 
\begin{align}
\|\phi - \phi^{(k)}\| \leq \delta.
\end{align}

Using Lemma~\ref{lemma:lipchitz}, we have for each $i = 1, \ldots, N$:
\begin{align}
    |SW_2^2(f_\phi \sharp F_i, G_i) - SW_2^2(f_{\phi^{(k)}} \sharp F_i, G_i)| \leq \gamma(\delta),
\end{align}
where $\lim_{\delta\to 0}\gamma(\delta)=0$.
Averaging over $i$ and taking expectation, this implies
\begin{align}
   &\frac{1}{N}\sum_{i=1}^N |SW_2^2(f_\phi \sharp F_i, G_i) - SW_2^2(f_{\phi^{(k)}} \sharp F_i, G_i)| \leq \gamma(\delta), \\
   & \mathbb{E}[|SW_2^2(f_\phi \sharp F_i, G_i) - SW_2^2(f_{\phi^{(k)}} \sharp F_i, G_i)|] \leq \gamma(\delta)
\end{align}
By  Jensen's inequality, we have:
\begin{align}
   &\left|\frac{1}{N}\sum_{i=1}^N SW_2^2(f_\phi \sharp F_i, G_i) - \frac{1}{N}\sum_{i=1}^N SW_2^2(f_{\phi^{(k)}} \sharp F_i, G_i) \right| \leq \gamma(\delta), \\
   & \left|\mathbb{E}[ SW_2^2(f_\phi \sharp F_i, G_i)] - \mathbb{E}[ SW_2^2(f_{\phi^{(k)}} \sharp F_i, G_i) \right|] \leq \gamma(\delta),
\end{align}
which is equivalent to
\begin{align}
    |R_N(\phi) - R_N(\phi^{(k)})| \leq \gamma(\delta), \quad |R(\phi) - R(\phi^{(k)})| \leq \gamma(\delta).
 \end{align}
Using triangle inequality, we have:
\begin{align}
    &\sup_{\phi \in \Phi} |R_N(\phi) - R(\phi)|  \nonumber \\
    &\leq \sup_{\phi \in \Phi} \left(|R_N(\phi) -R_N(\phi^{(k)}))| +|R_N(\phi^{(k)}) - R(\phi)|  \nonumber \right)\\ 
    &\leq \nonumber   \sup_{\phi \in \Phi} \left( |R_N(\phi) -R_N(\phi^{(k)}))|+ |R_N(\phi^{(k)}) - R(\phi^{(k)})|  +|R(\phi^{(k)}) - R(\phi)| \right) \\
    &\leq \max_{k = 1,\ldots,K} |R_N(\phi^{(k)}) - R(\phi^{(k)})| + 2 \gamma( \delta). \label{eq:key-net-bound}
\end{align}


Let
$
Z_N \;:=\; \sup_{\phi\in\Phi}\,|R_N(\phi)-R(\phi)| $ and $
A_N^\varepsilon \;:=\; \{\,Z_N>\varepsilon\,\},\ \ \varepsilon>0 .
$
Almost sure convergence $Z_N\overset{a.s}{\to} 0$ is equivalent to
\begin{align}
\forall\,\varepsilon>0:\quad \mathbb{P}\!\left(\limsup_{N\to\infty} A_N^\varepsilon\right)=0,
\end{align}
where $\displaystyle \limsup_{N\to\infty} A_N^\varepsilon:=\bigcap_{m=1}^\infty\bigcup_{N\ge m}A_N^\varepsilon$. For a fixed $\varepsilon>0$,  we choose $\delta>0$ so small that $2\gamma(\delta)<\varepsilon/2$. Then \eqref{eq:key-net-bound} implies the event inclusion
\begin{align}
A_N^\varepsilon
\;\subseteq\;
\bigcup_{k=1}^K
\left\{\,|R_N(\phi^{(k)})-R(\phi^{(k)})|>\frac{\varepsilon}{2}\right\}
=: \bigcup_{k=1}^K B_{N,k}^{\varepsilon/2}.
\end{align}
Therefore,
\begin{align}
\limsup_{N\to\infty} A_N^\varepsilon
\;\subseteq\;
\bigcup_{k=1}^K \left( \limsup_{N\to\infty} B_{N,k}^{\varepsilon/2} \right).
\end{align}

By the strong law of large numbers (using Assumption~\ref{assumption:secondmoment}) applied to each fixed $\phi^{(k)}$,
\begin{align}
|R_N(\phi^{(k)})-R(\phi^{(k)})|\overset{a.s}{\to}0,
\end{align}
hence, for every $k$ and every $\eta>0$,
\begin{align}
\mathbb{P}\!\left(\limsup_{N\to\infty} \{\,|R_N(\phi^{(k)})-R(\phi^{(k)})|>\eta\,\}\right)=0.
\end{align}
Taking $\eta=\varepsilon/2$ and using the finite union bound, we get
\begin{align}
\mathbb{P}\!\left(\limsup_{N\to\infty} A_N^\varepsilon\right)
\;\le\;
\sum_{k=1}^K
\mathbb{P}\!\left(\limsup_{N\to\infty} B_{N,k}^{\varepsilon/2}\right)
=0.
\end{align}

Since this holds for every $\varepsilon>0$, it follows that
\begin{align}
\forall\,\varepsilon>0:\quad \mathbb{P}\!\left(\limsup_{N\to\infty} A_N^\varepsilon\right)=0,
\end{align}
which is equivalent to $Z_N\to 0$ almost surely, i.e.,
\begin{align}
\sup_{\phi\in\Phi}|R_N(\phi)-R(\phi)| \overset{a.s}{\to} 0.
\end{align}

\end{proof}

 Back to the main proof of posterior consistency, for all $\epsilon>0$, we define the “bad” set
\begin{align}
S_\epsilon(\phi_0) := \{\phi \in \Phi : \|\phi-\phi_0\|_2 \ge \epsilon\},
\end{align}
where $\phi_0$ is defined in Assumption~\ref{assumption:identifiability}. We want to show
\begin{align}
\pi_N(S_\epsilon(\phi_0)) \overset{a.s}{\to} 0,
\end{align}
for all $\epsilon>0$. For $\varepsilon>0$, let
\begin{align}
A_N := \big\{\pi_N(S_\epsilon(\phi_0)) > \varepsilon\big\}.
\end{align}
By the definition of almost sure convergence in terms of the limit superior of events, it suffices to show
\begin{align}
\mathbb{P}\Big(\limsup_{N\to\infty} A_N \Big) = 0,
\end{align}
where
$
\limsup_{N\to\infty} A_N := \bigcap_{m=1}^\infty \bigcup_{N \ge m} A_N.
$
By the uniform law of large numbers (Lemma~\ref{lemma:uniform_law}), for any $\eta>0$ define the event
\begin{align}
E_N := \Big\{\sup_{\phi\in\Phi}|R_N(\phi)-R(\phi)| \le \eta \Big\}.
\end{align}
Then $E_N$ occurs eventually almost surely. On $E_N$:

\begin{itemize}
\item For $\phi \in S_\epsilon(\phi_0)$,
\begin{align}
R_N(\phi) \ge R(\phi) - \eta \ge R(\phi_0) + \Delta - \eta,
\end{align}
where $ \Delta=\Delta(\epsilon) \;=\; \inf_{\{\phi \in \Phi : \|\phi - \phi_0\|_2 \geq \epsilon\}} \big( R(\phi) - R(\phi_0) \big) \;>\; 0$ (Assumption~\ref{assumption:identifiability}).
\item For $\phi \in B_\delta(\phi_0) := \{\phi : \|\phi-\phi_0\|_2<\delta\}$, choose $\delta$ small so that
\begin{align}
\sup_{\phi \in B_\delta(\phi_0)} (R(\phi)-R(\phi_0)) \le \eta,
\end{align}
then
\begin{align}
R_N(\phi) \le R(\phi) + \eta \le R(\phi_0) + 2\eta.
\end{align}
\end{itemize}

On $E_N$, the posterior mass ratio
\begin{align}
    \frac{\pi_N(S_\epsilon(\phi_0))}{\pi_N(B_\delta(\phi_0))}
    &= \frac{\int_{S_\epsilon(\phi_0)} \exp(-w N R_N(\phi)) \, d\pi(\phi)}
           {\int_{B_\delta(\phi_0)} \exp(-w N R_N(\phi)) \, d\pi(\phi)} \notag \\
    &\leq \frac{\pi(\Phi) \, \exp(-w N (R(\phi_0) + \Delta - \eta))}{\pi(B_\delta(\phi_0)) \, \exp(-w N (R(\phi_0) + 2\eta))} \notag \\
    &= \frac{1}{\pi(B_\delta(\phi_0))} \exp(-w N (\Delta - 3\eta)),
\end{align}
which implies
\begin{align}
    \pi_N(S_\epsilon(\phi_0)) \leq C \exp(-w N (\Delta - 3\eta)), 
\end{align}
where $C=\frac{1}{\pi(B_\delta(\phi_0))} >0$ due to Assumption~\ref{assumption:prior}.
Choose $\eta < \Delta/4$ so that $\Delta-3\eta > \Delta/4>0$., then the right-hand side decays exponentially in $N$, implying that for almost every sample path $\omega$ there exists $N_0(\omega)$ such that for all $N \ge N_0(\omega)$,
\begin{align}
\pi_N(S_\epsilon(\phi_0)) < \varepsilon.
\end{align}

By the definition of $\limsup$ of events,
\begin{align}
\limsup_{N\to\infty} A_N = \bigcap_{m=1}^\infty \bigcup_{N\ge m} \{\pi_N(S_\epsilon(\phi_0)) > \varepsilon\} = \emptyset \quad \text{a.s.}
\end{align}
Hence,
\begin{align}
\mathbb{P}\Big(\limsup_{N\to\infty} A_N \Big) = 0.
\end{align}
Since the choice of $\varepsilon>0$ is arbitrary, we conclude
\begin{align}
\pi_N(S_\epsilon(\phi_0)) \xrightarrow{\text{a.s.}} 0.
\end{align}
Thus the posterior concentrates around $\phi_0$ for any $\epsilon>0$, proving Bayesian posterior consistency.

\begin{proposition}[Linear function]
    \label{proposition:linear_function}
    Let $f_{\phi}(x)=Ax+b$ ($\phi=(A,b)$), the following holds:
\begin{align}
    \mathbb{E}_{X\sim F}\|f_\phi(X) - f_{\phi'}(X)\|_2^2 \leq  (M_1+1)\|\phi -\phi'\|_2^2,
\end{align}
 for all $(F,G) \in \mathrm{supp}(P)$, $\phi,\phi' \in \Phi$, and $M_1>0$ in Assumption~\ref{assumption:secondmoment}. In addition, for all $(F,G) \in \mathrm{supp}(P)$ and $\phi \in \Phi$,
\begin{align}
    \mathbb{E}_{X \sim F} [\|f_\phi(X)\|_2^2 ]\leq C,
\end{align}
for a constant $ C < \infty$.
\end{proposition}
\begin{proof}
    Using the Cauchy–Schwarz inequality,
    \begin{align}
         \mathbb{E}_{X\sim F}[\|f_\phi(X) - f_{\phi'}(X)\|_2^2] &=  \mathbb{E}_{X\sim F}[\|AX+b-A'X+b'\|_2^2 ]\nonumber \\
         &=\mathbb{E}_{X\sim F}[\|(A-A')X+(b-b')\|_2^2] \nonumber \\
         &\leq \mathbb{E}_{X\sim F}[(\|A-A'\|_F^2 + \|b-b'\|_2^2)( \|X\|_2^2+1)] \nonumber \\
         &\leq  (\|A-A'\|_F^2 + \|b-b'\|_2^2)(\mathbb{E}_{X\sim F}[\|X\|_2^2+1]) \nonumber \\
         &= (\|A-A'\|_F^2 + \|b-b'\|_2^2) (M_1+1) \nonumber \\
         &= (M_1+1) \|\phi-\phi'\|_2^2.
    \end{align}
    due to Assumption~\ref{assumption:secondmoment}. 

    We further have:
    \begin{align}
        \mathbb{E}_{X \sim F} [\|f_\phi(x)\|_2^2 ] &=  \mathbb{E}_{X \sim F} [\|AX+b\|_2^2 ] \nonumber\\
        &\leq \mathbb{E}_{X \sim F} [(\|A\|_F^2+\|b\|_2^2)(\|X\|_2^2+1) ] \\ 
        &\leq (\|A\|_F^2+\|b\|_2^2)(M_1+1),
    \end{align}
     due to Assumption~\ref{assumption:secondmoment}.  Since the parameter space $\Phi$ is compact (Assumption~\ref{assumption:compact}), $(\|A\|_F^2+\|b\|_2^2)$ is bounded, hence,
     \begin{align}
          \mathbb{E}_{X \sim F} [\|f_\phi(x)\|_2^2 ] \leq C,
     \end{align}
     for a constant $C=(\|A\|_F^2+\|b\|_2^2)(M_1+1)>0$, which completes the proof.
\end{proof}

\color{black}

\section{MCMC}
\label{sec:MCMC}
We present the detail MCMC sampler for Bayesian DDR with the linear regression function:
$\lambda_{ij}^2 \mid A_{ij}, \tau \sim \text{Inverse-Gamma}\left(1,
  \frac{1}{\nu_{ij}} + \frac{A_{ij}^2}{2 \tau^2}\right)$,
$\nu_{ij} \mid \lambda_{ij} \sim \text{Inverse-Gamma}\left(\frac{1}{2}, 1 + \frac{1}{\lambda_{ij}^2}\right)$,
$\tau^2 \mid A, \lambda_{ij}, \zeta \sim
\text{Inverse-Gamma}\left(\frac{d_1 + d_2}{2}, \frac{1}{\zeta} +
  \sum_{i=1}^{d_1}\sum_{j=1}^{d_2} \frac{A_{ij}^2}{2
    \lambda_{ij}^2}\right)$, 
$\zeta \mid \tau \sim \text{Inverse-Gamma}\left(\frac{1}{2}, 1 +
  \frac{1}{\tau^2}\right)$,
and $(A, b) \sim p(A, b \mid \lambda^2,\tau^2, \mathcal{S})$.

\section{Graph Discovery}
\label{subsec:graph_discovery_appendix}

In Figures~\ref{fig:prediction_1}-\ref{fig:prediction_3} we show
fitted means  under  Bayesian MLR and  
fitted densities  under  Bayesian DDR for all ordered pairs of
cell types.
As before, we use PCA for the response distribution of each donor to
obtain  a 2-dimensional projection. 
We observe that the predictions align with the reported RPEs
shown in Figure~\ref{fig:RPEs} in the main text. Notably, Bayesian DDR
provides more precise fits compared to Bayesian MLR, especially in
cases where the response distributions are multimodal or exhibit high
uncertainty.

\begin{figure}[!t]
  \begin{center}
    \begin{tabular}{c}
      \widgraph{1\textwidth}{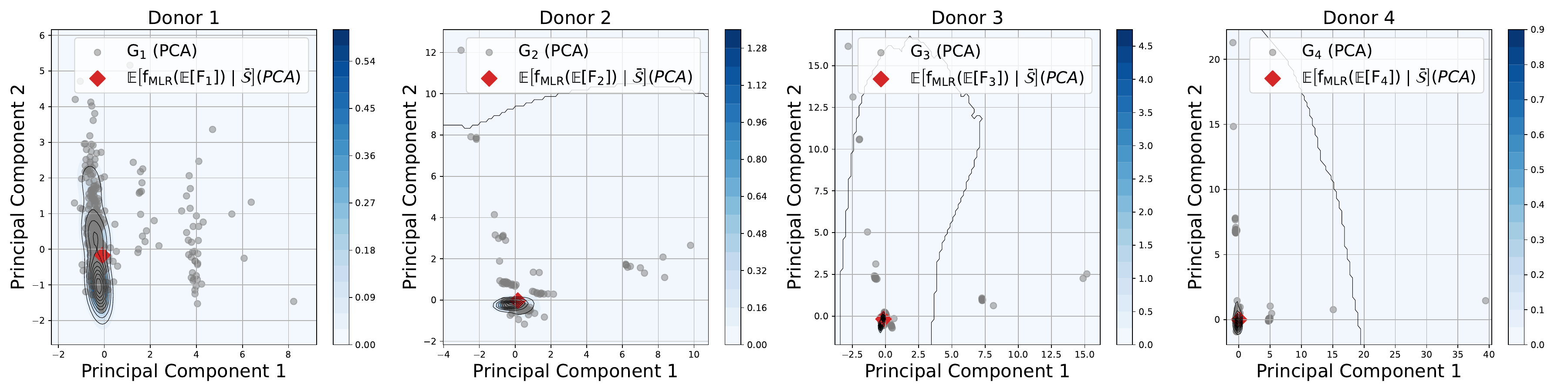}
      \\ (a) B to T cells \\
      \widgraph{1\textwidth}{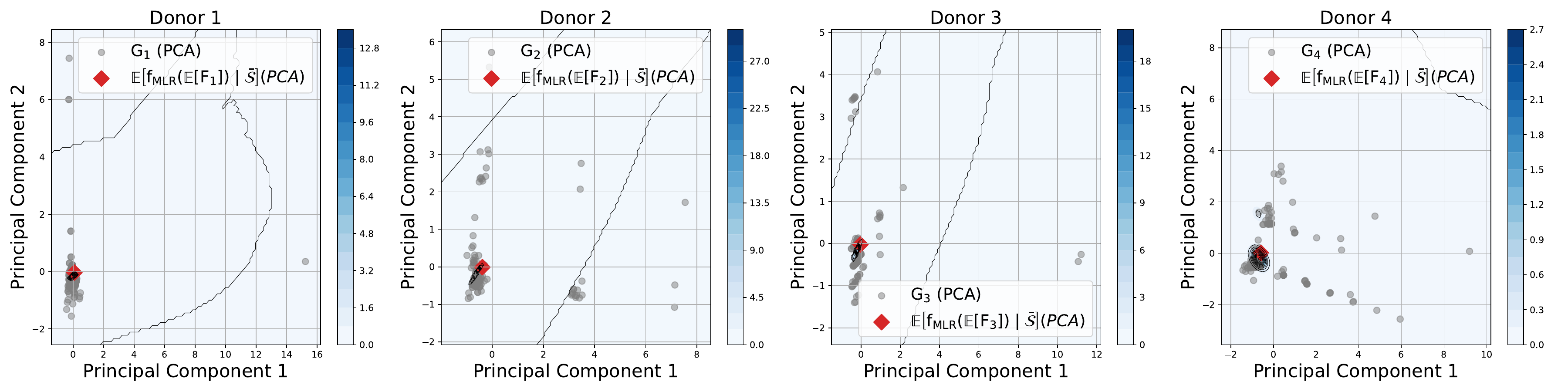}
      \\ (b) T to B cells \\
      \widgraph{1\textwidth}{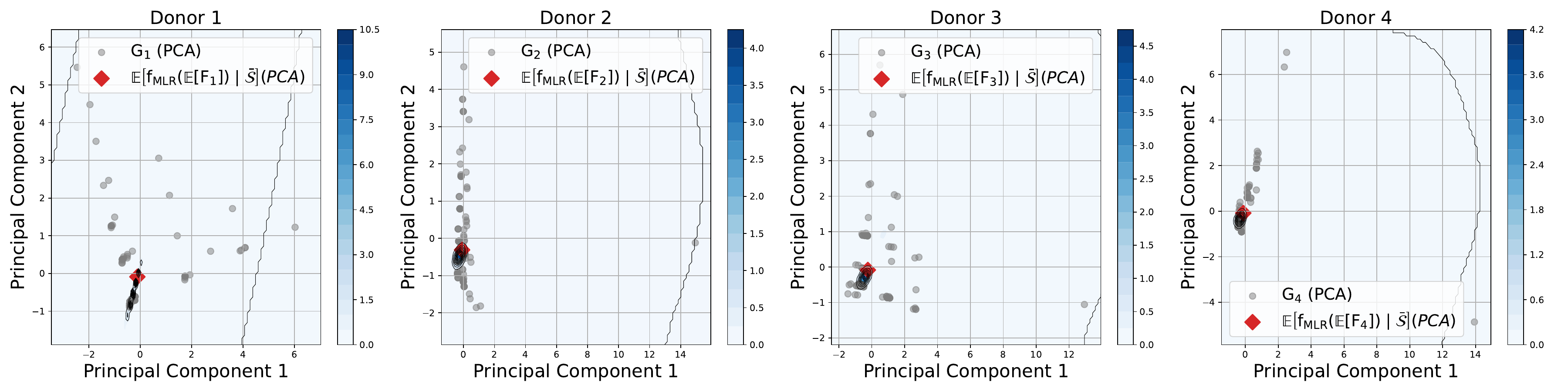}
      \\ (c) T cells to monocytes \\
      \widgraph{1\textwidth}{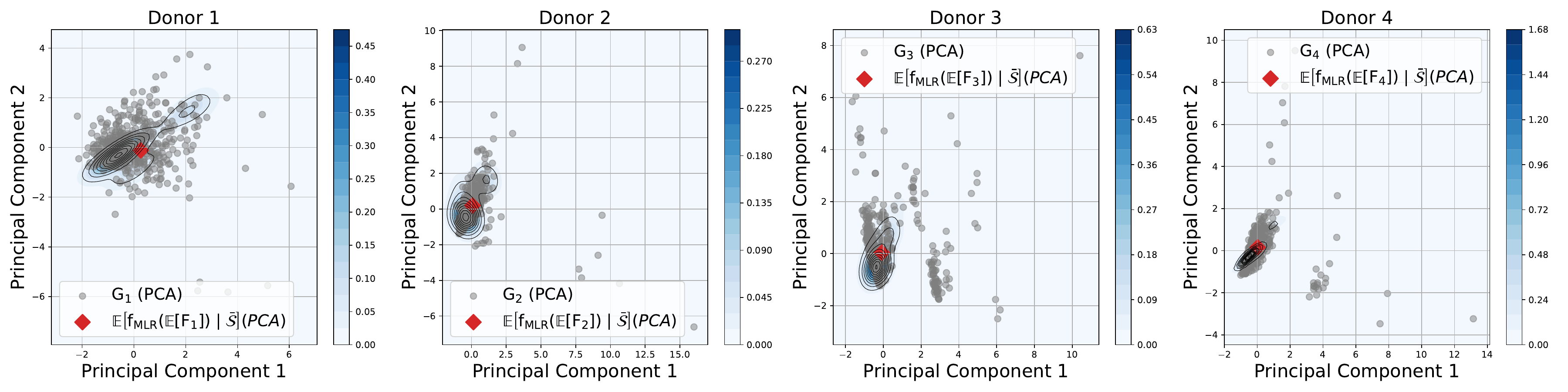}
      \\ (d) monocytes to T cells
    \end{tabular}
  \end{center} 
  \caption{ \footnotesize{Fitted
      densities under Bayesian DDR and fitted means under
      Bayesian MLR  for edges as indicated by subtitles (a)
      through (d).
     Compare with Figure \ref{fig:prediction_simulation} for
    an explanation of the symbols and contours. 
  }}
  \label{fig:prediction_1}
\end{figure}

 Figure~\ref{fig:RPEs} shows boxplots of RPE for all ordered pairs
of cell types. 
As before,  we use the
intercept model  as the reference $f_0$ 
for computing RPEs.
Figures~\ref{fig:T_to_B_coeff}-\ref{fig:NK_to_B_coeff}
show boxplots for the linear regression coefficients of all
ordered pairs of cell types.
 Such inference allows us to  identify which ligands are
significant in predicting a receptor.  This is done by inspecting
the posterior distribution for  the corresponding coefficient
 $A_{eij}$. Values far from zero indicate significant
receptor/ligand pairs. 
Consider, for example,  the regression of B cell receptor expression
on T cell ligands. Both  ligands, CD40L and IFN-$\gamma$  are
significant in predicting receptors such as IGHM, IGHD, IGHG, and
IL-6R. Specifically, the ligand IFN-$\gamma$ appears  critical  to
predict  gene expression for the receptor  IGHA, while CD40L
 is required  to predict receptors CD40, ICOSL, IL-21R, BAFFR,
CXCR5, and CCR7. Similar 
conclusions can be drawn for other  ligand/receptor
relations. 
For a more  formal assessment one could 
use pseudo inclusion probabilities
like \eqref{eq:eIP} to quantify the evidence for 
a ligand being significant in predicting a receptor. 
\begin{figure}[!t]
\begin{center}
  \begin{tabular}{c}
\widgraph{0.9\textwidth}{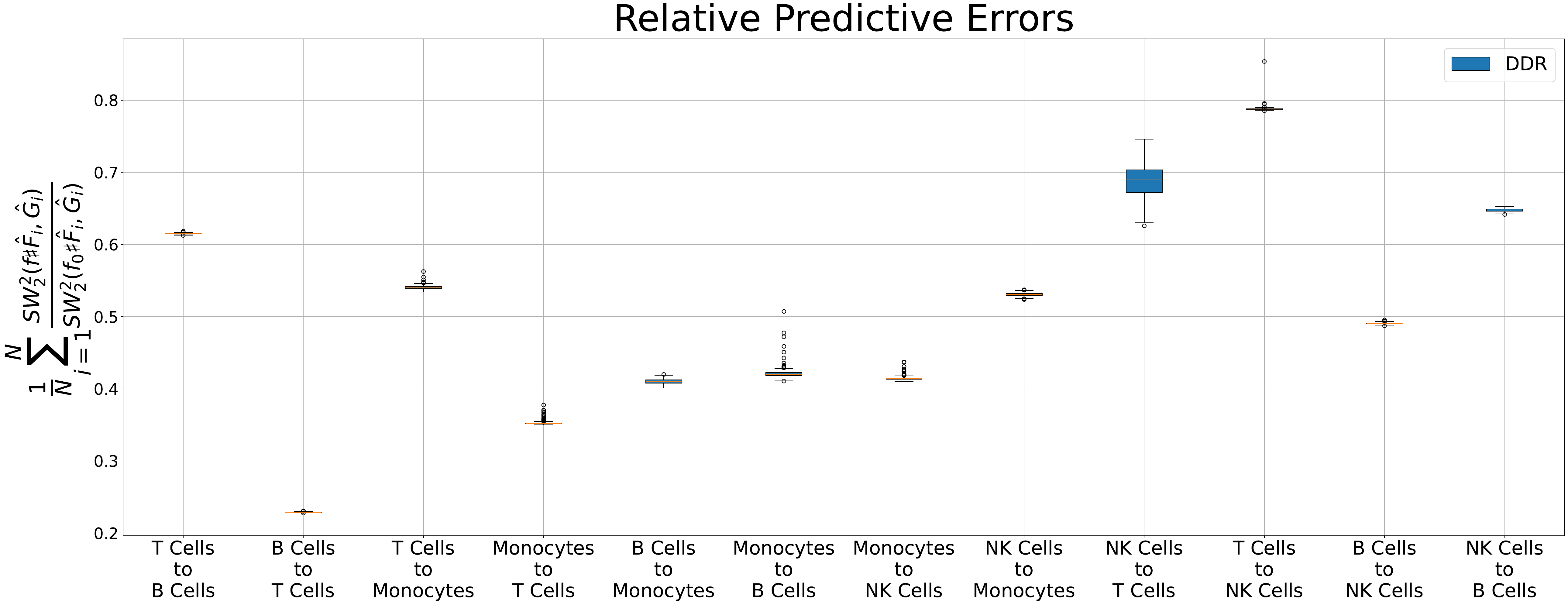} 
  \end{tabular}

  \end{center}
  \caption{
  \footnotesize{Boxplots of relative predictive errors for Bayesian
    DDR for each ordered pair of cell types on the x-axis. 
}
} 
  \label{fig:RPEs}
\end{figure}
\begin{figure}[!t]
  \begin{center}
    \begin{tabular}{cccc}
      \widgraph{1\textwidth}{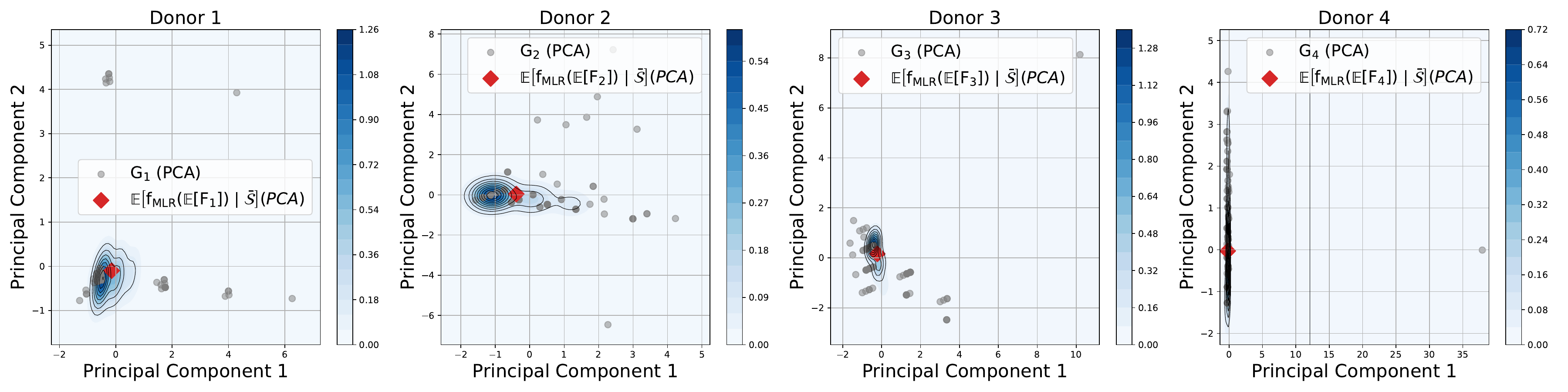}
      \\ (a) B cells to monocytes \\
      \widgraph{1\textwidth}{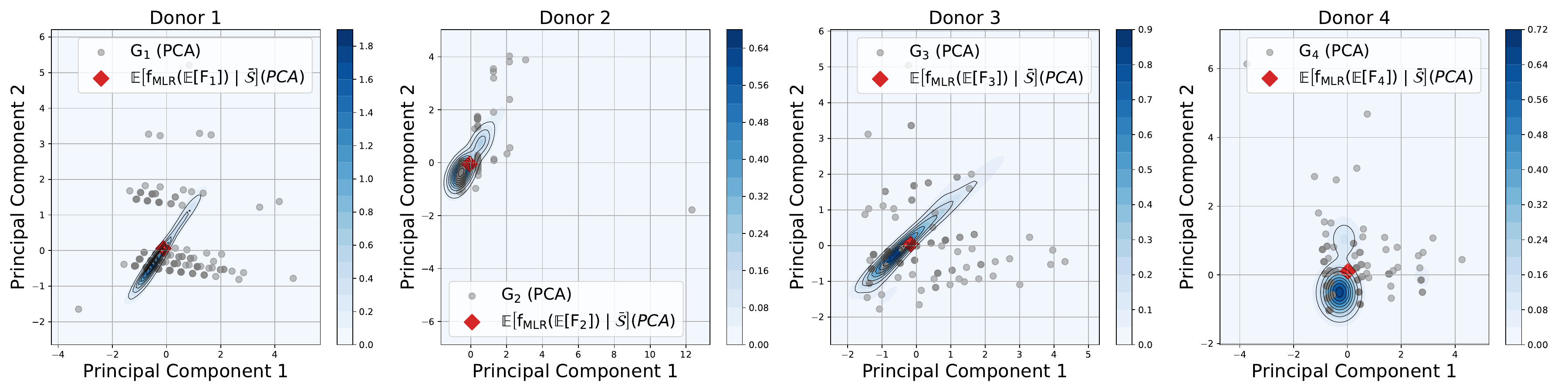}
      \\ (b) monocytes to B cells \\
      \widgraph{1\textwidth}{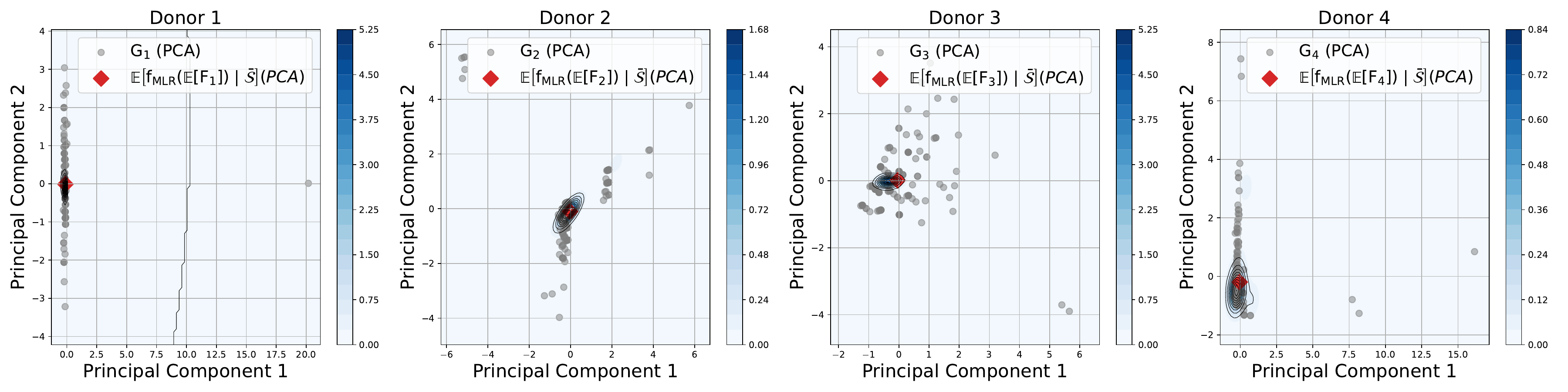}
      \\ (c) monocytes to NK cells \\
      \widgraph{1\textwidth}{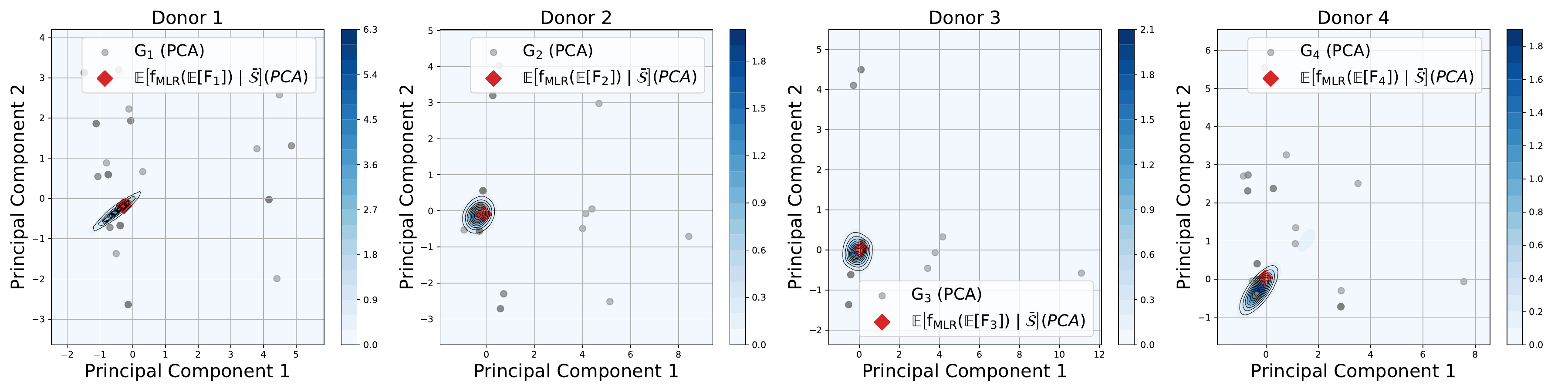}
      \\ (d) NK cells to monocytes
    \end{tabular}
  \end{center} 
  \caption{ \footnotesize{
      Same as Figure \ref{fig:prediction_1} for more pairs. }}
  \label{fig:prediction_2}
\end{figure}

\begin{figure}[!t]
\begin{center}
    \begin{tabular}{cccc}
  \widgraph{1\textwidth}{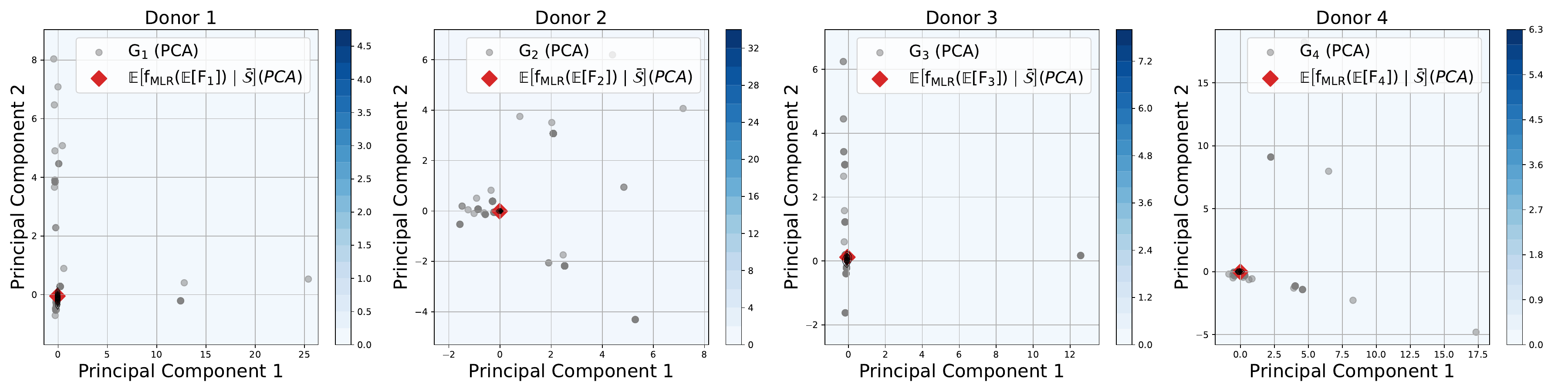} \\
  (a) NK to T cells \\
  \widgraph{1\textwidth}{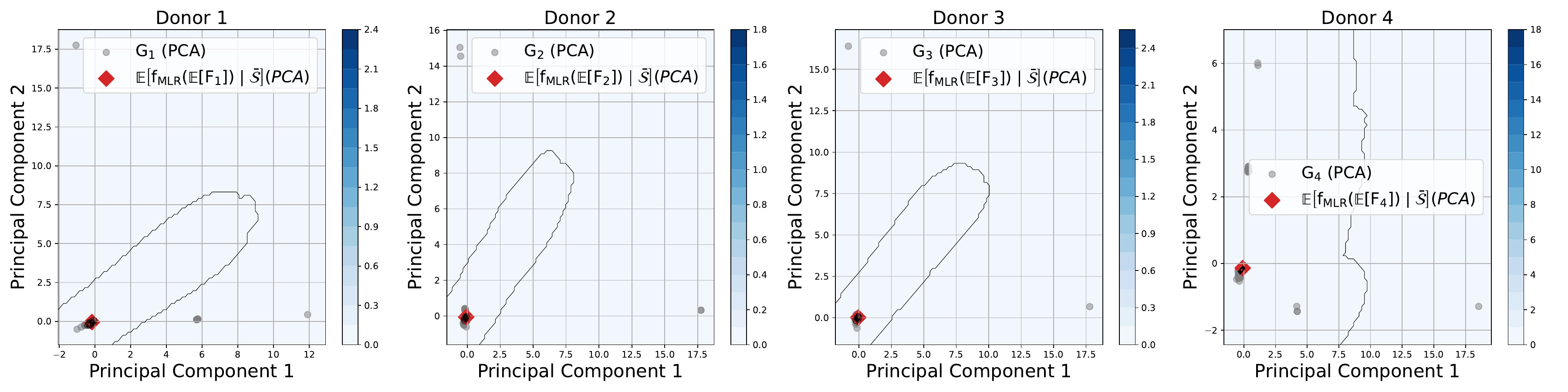} \\
  (b) T cells to NK cells \\
  \widgraph{1\textwidth}{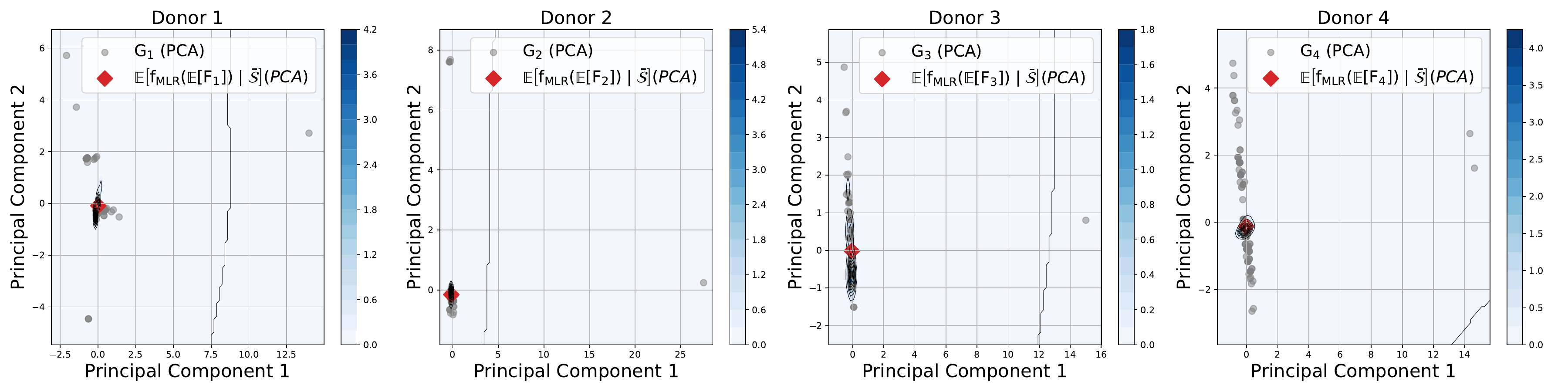} \\
  (c) B to NK cells \\
  \widgraph{1\textwidth}{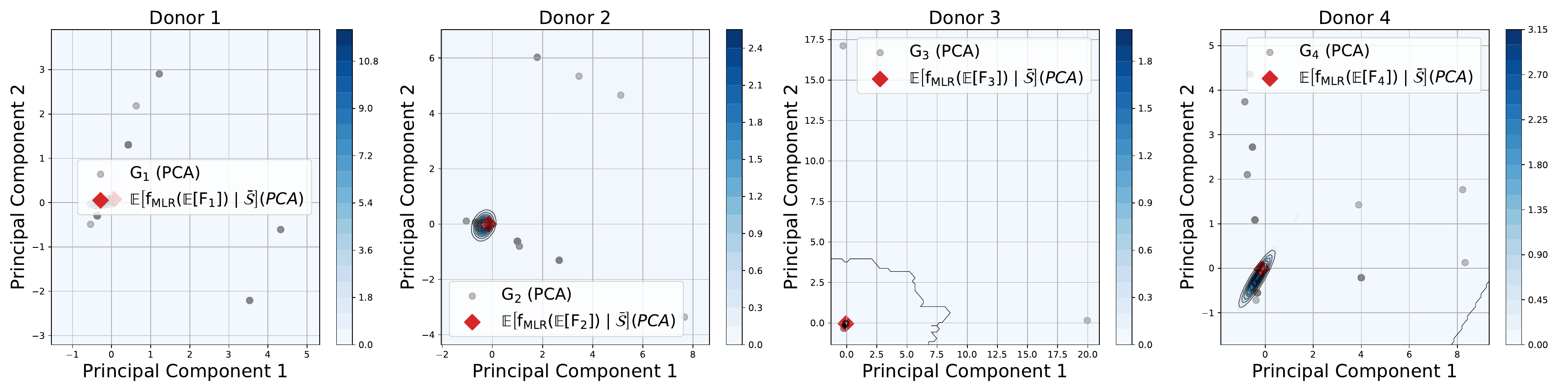} \\
  (d) NK to B cells
  \end{tabular}
  \end{center}
  \caption{
    \footnotesize{
      Same as Figure \ref{fig:prediction_1} for more pairs. }}
  \label{fig:prediction_3}
\end{figure}

 \begin{figure}[!t]
\begin{center}
    \begin{tabular}{c}
\widgraph{1\textwidth}{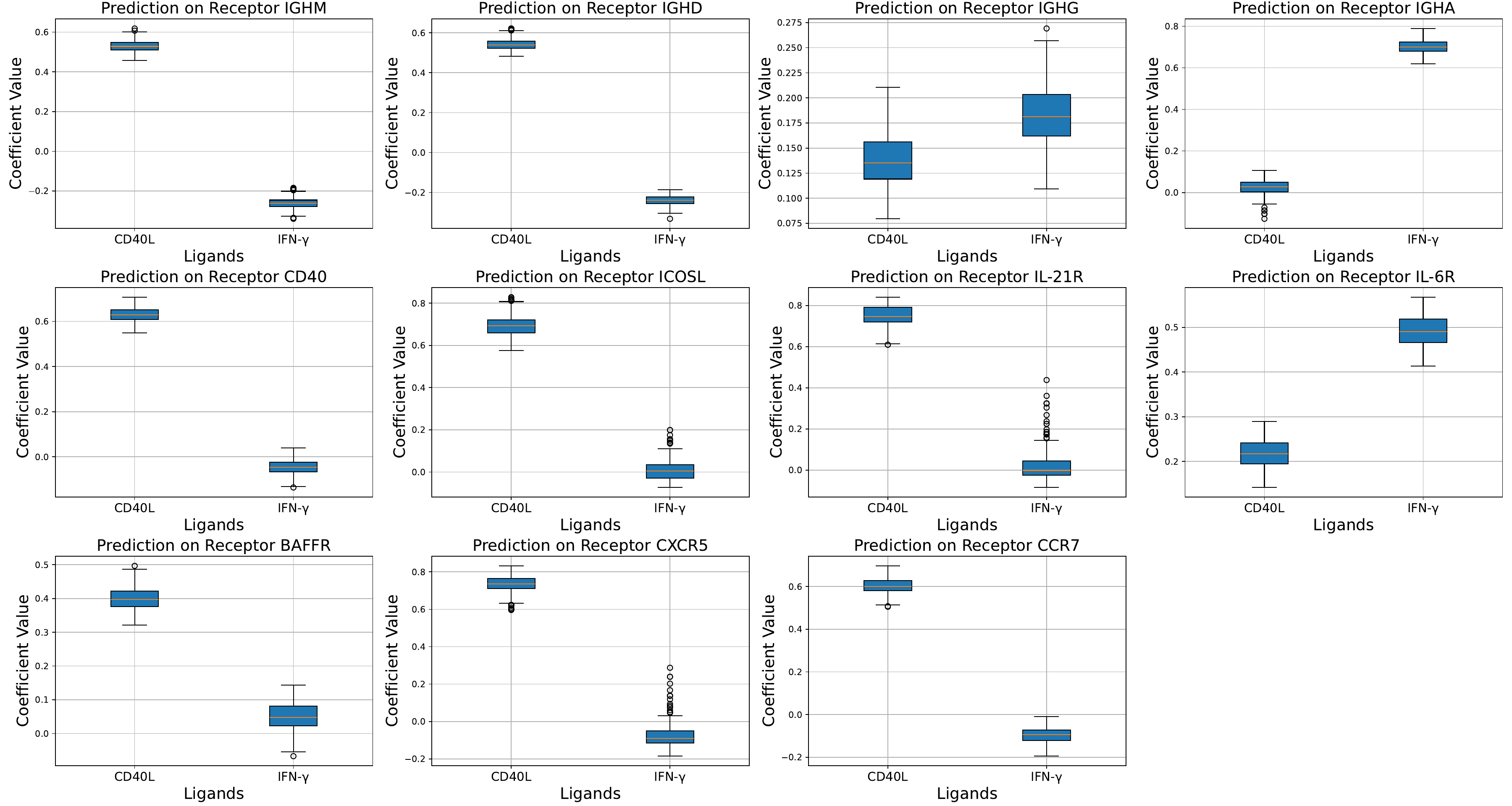}
  \end{tabular}
  \end{center}
  \caption{
  \footnotesize{Boxplots of  posterior samples for  the
    coefficients of the linear regression matrix $A_e$ for the ordered pair
    of cell types: T to B cells.  
}
} 
  \label{fig:T_to_B_coeff}
\end{figure}

 \begin{figure}[!t]
\begin{center}
    \begin{tabular}{c}
\widgraph{1\textwidth}{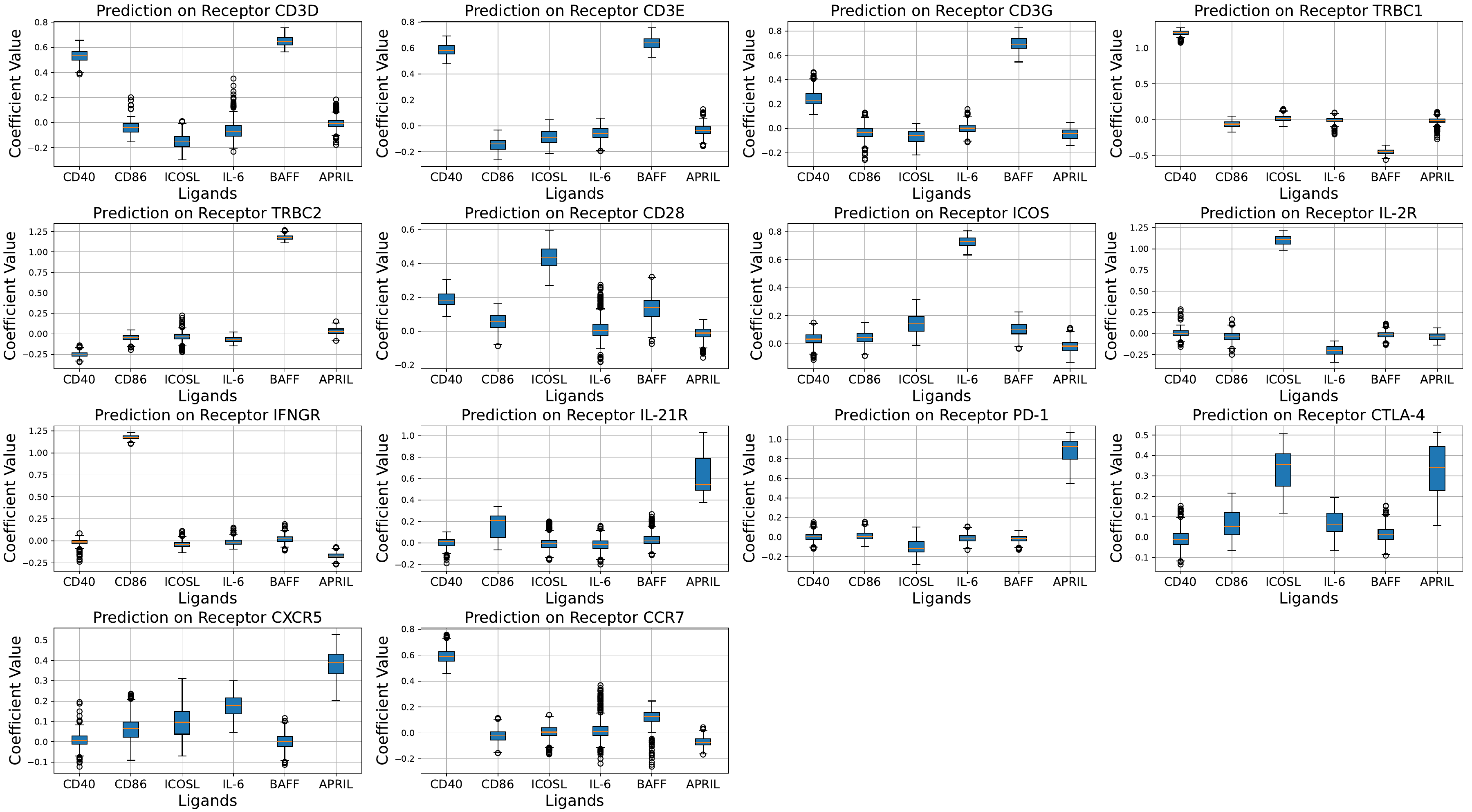}
  \end{tabular}
  \end{center}
  \caption{
    \footnotesize{
      Same as Figure \ref{fig:T_to_B_coeff}
      for B to T cells.
}
} 
  \label{fig:B_to_T_coeff}
\end{figure}

 \begin{figure}[!t]
\begin{center}
    \begin{tabular}{c}
\widgraph{1\textwidth}{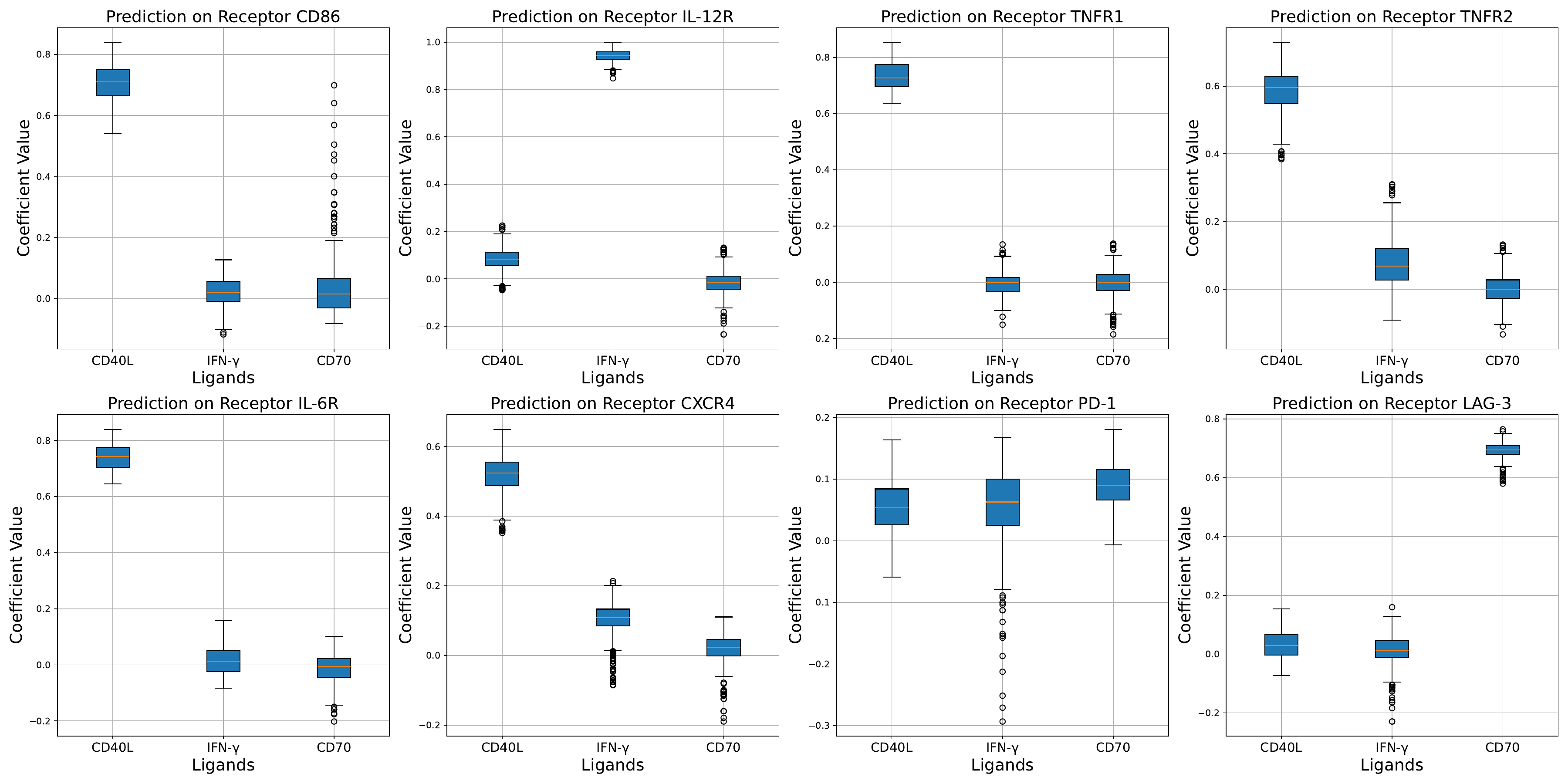}
  \end{tabular}
  \end{center}
  \caption{
    \footnotesize{
      Same as Figure \ref{fig:T_to_B_coeff}
      for T cells to monocytes
}
} 
  \label{fig:T_to_M_coeff}
\end{figure}

 \begin{figure}[!t]
\begin{center}
    \begin{tabular}{c}
\widgraph{1\textwidth}{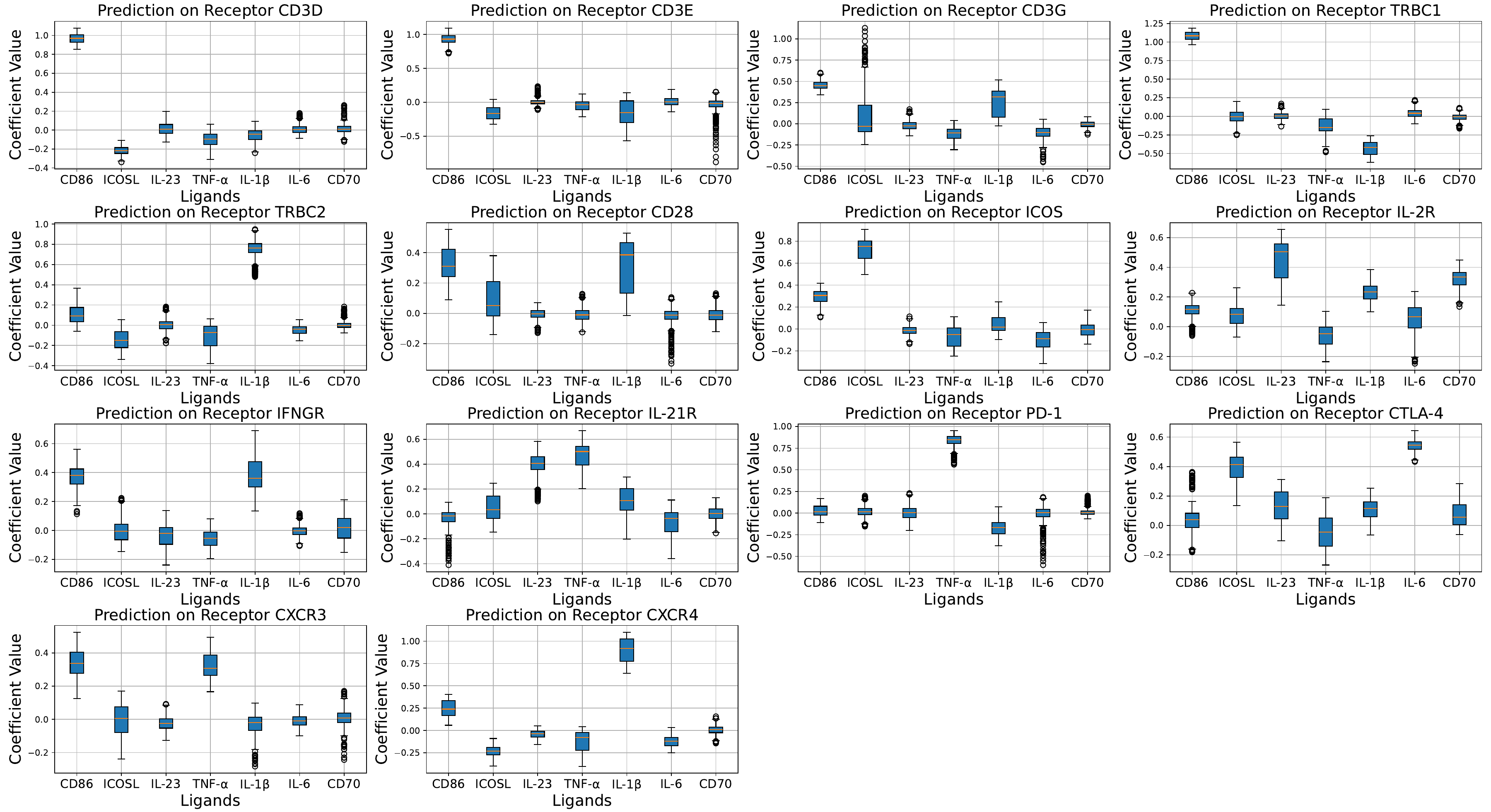}
  \end{tabular}
  \end{center}
  \caption{
    \footnotesize{
      Same as Figure \ref{fig:T_to_B_coeff} for
      monocytes to T cells
}
} 
  \label{fig:M_to_T_coeff}
\end{figure}

 \begin{figure}[!t]
\begin{center}
    \begin{tabular}{c}
\widgraph{1\textwidth}{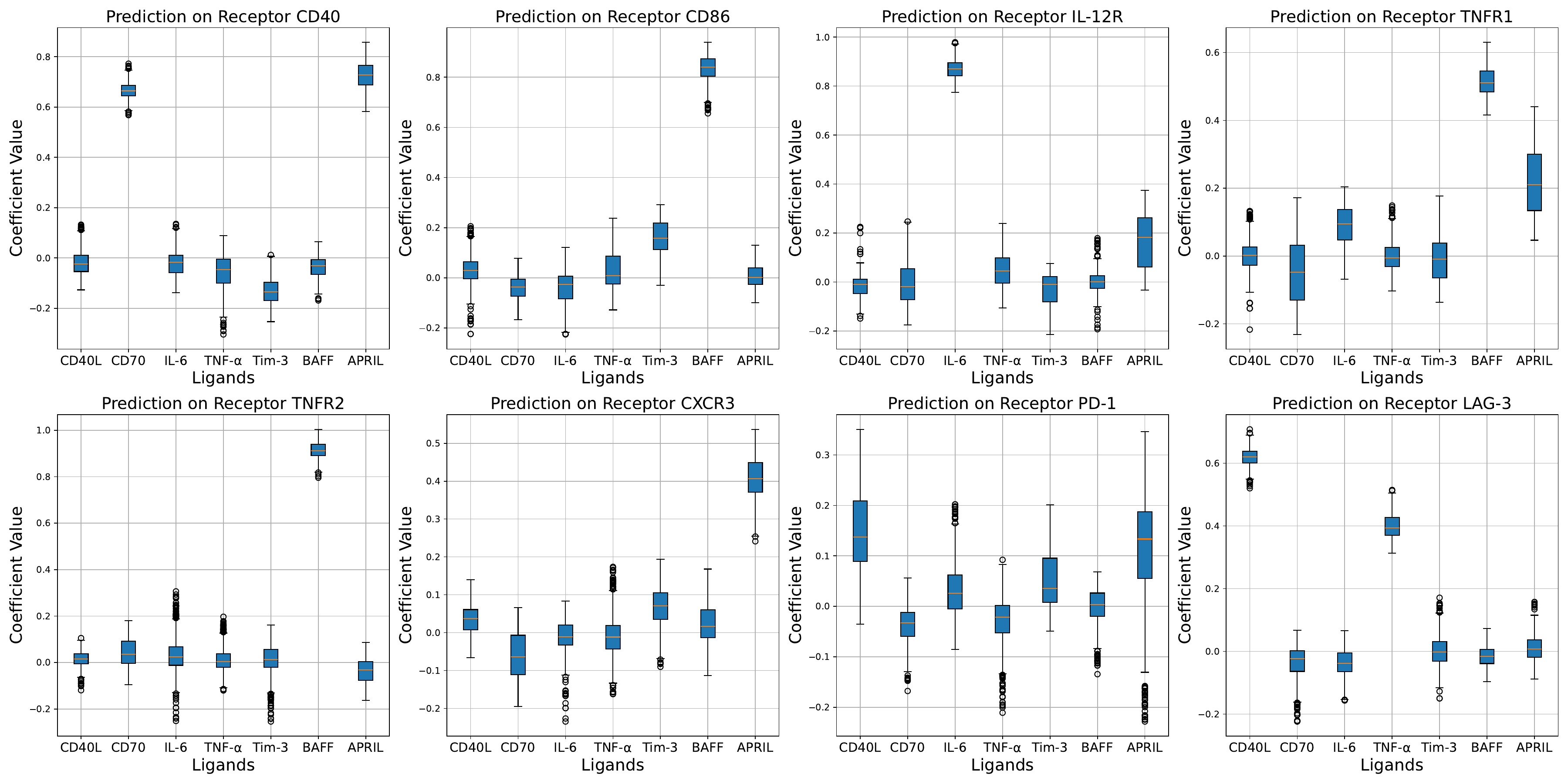}
  \end{tabular}
  \end{center}
  \caption{
    \footnotesize{
      Same as Figure \ref{fig:T_to_B_coeff} for
      B cells to monocytes
}
} 
  \label{fig:B_to_M_coeff}
\end{figure}

 \begin{figure}[!t]
\begin{center}
    \begin{tabular}{c}
\widgraph{1\textwidth}{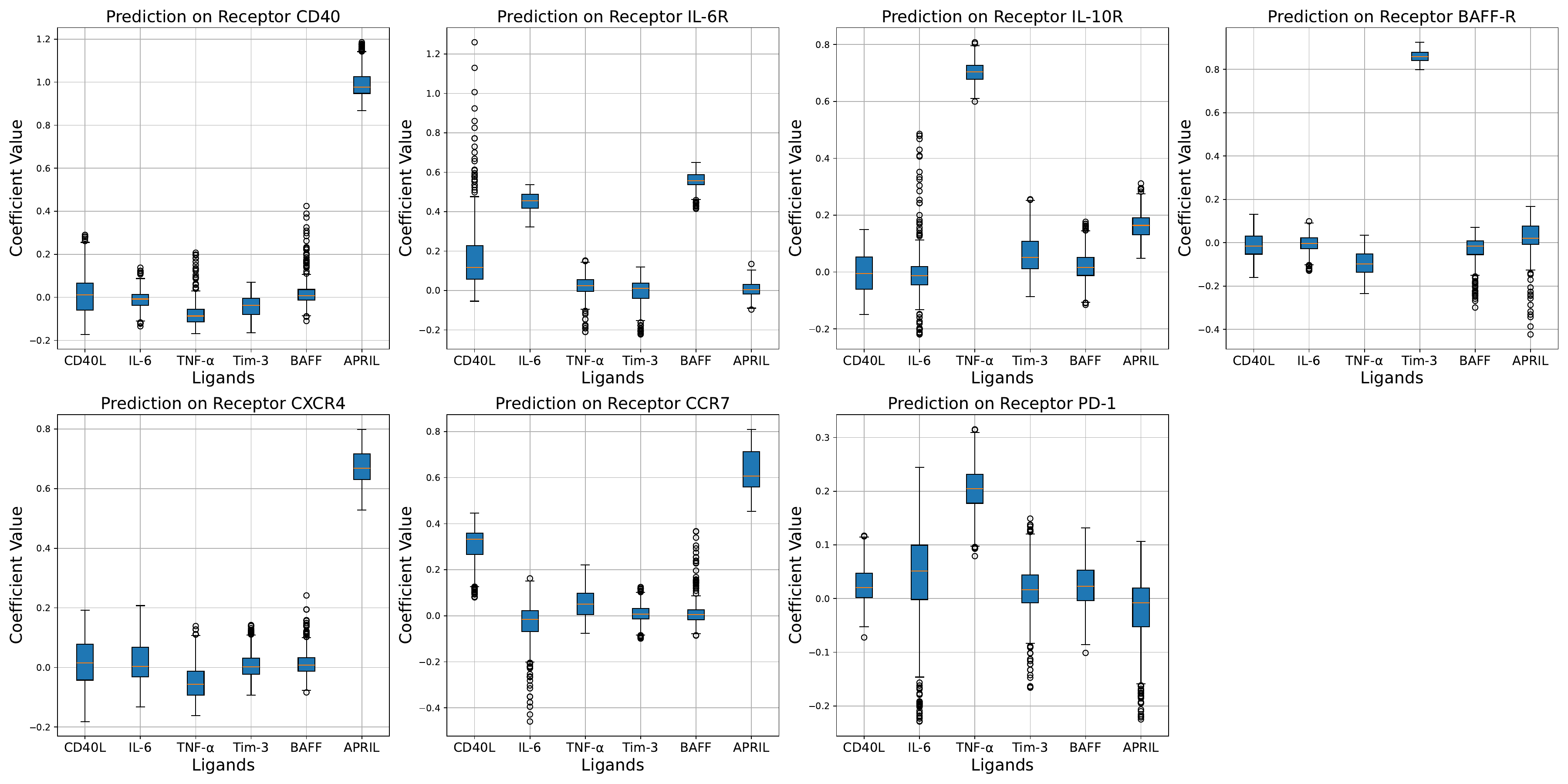}
  \end{tabular}
  \end{center}
  \caption{
    \footnotesize{
      Same as Figure \ref{fig:T_to_B_coeff} for
      monocytes to B cells 
}
} 
  \label{fig:M_to_B_coeff}
\end{figure}

 \begin{figure}[!t]
\begin{center}
    \begin{tabular}{c}
\widgraph{1\textwidth}{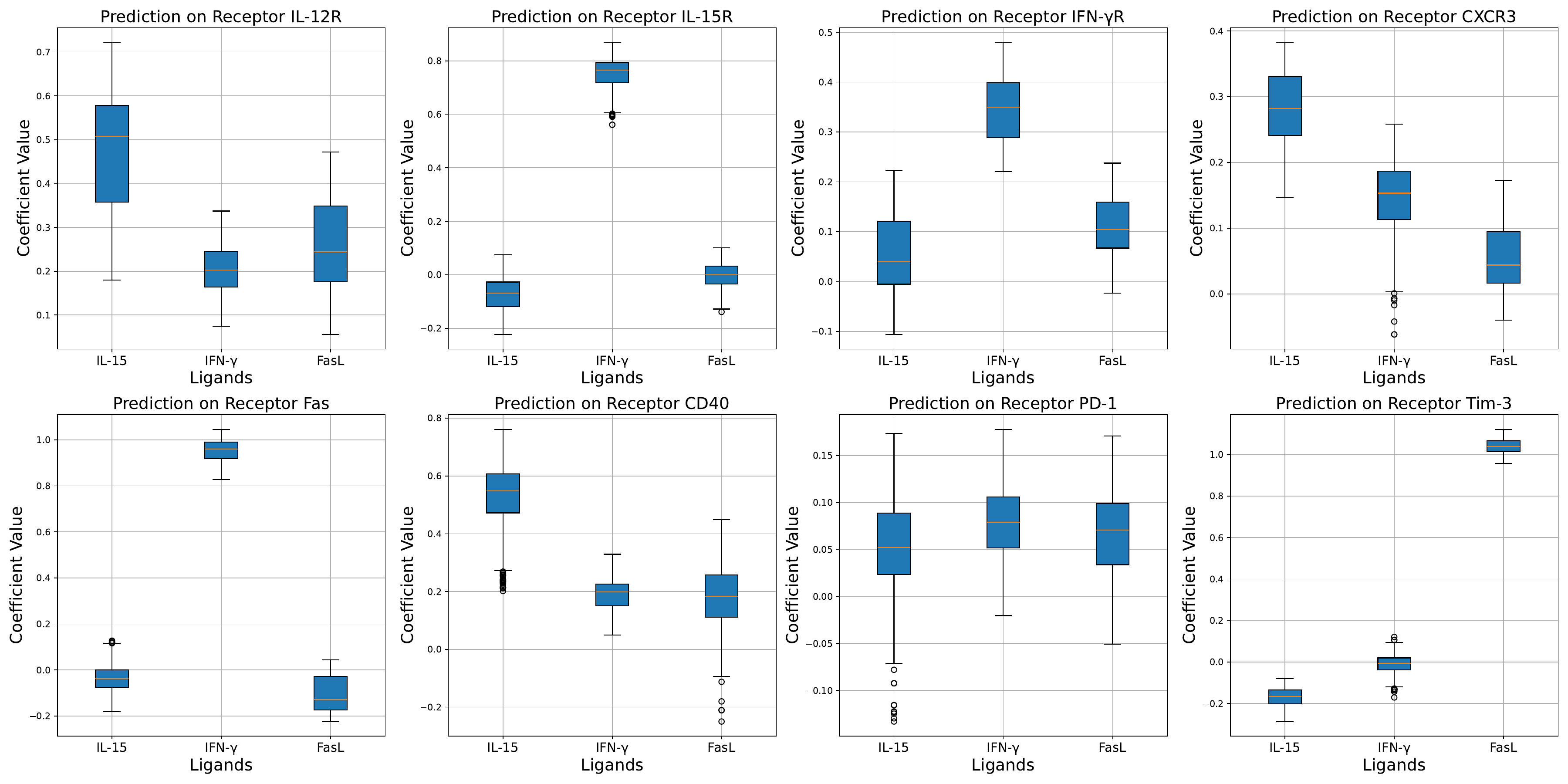}
  \end{tabular}
  \end{center}
  \caption{
    \footnotesize{
          Same as Figure \ref{fig:T_to_B_coeff} for
          NK cells to monocytes
}
} 
  \label{fig:NK_to_M_coeff}
\end{figure}

 \begin{figure}[!t]
\begin{center}
    \begin{tabular}{c}
\widgraph{1\textwidth}{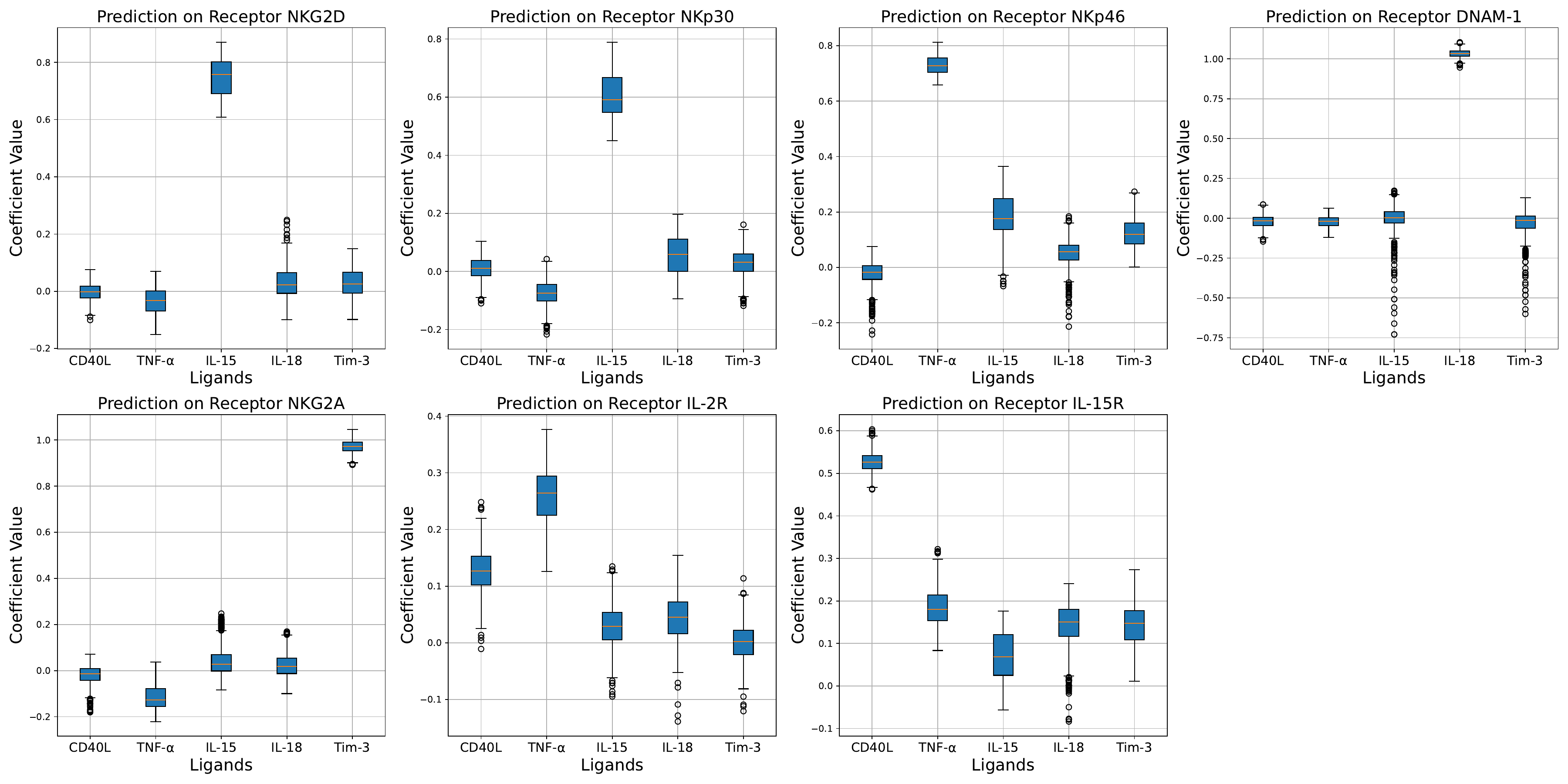}
  \end{tabular}
  \end{center}
  \caption{
    \footnotesize{
      Same as Figure \ref{fig:T_to_B_coeff} for
      monocytes to NK cells 
}
} 
  \label{fig:M_to_NK_coeff}
\end{figure}

 \begin{figure}[!t]
\begin{center}
    \begin{tabular}{c}
\widgraph{1\textwidth}{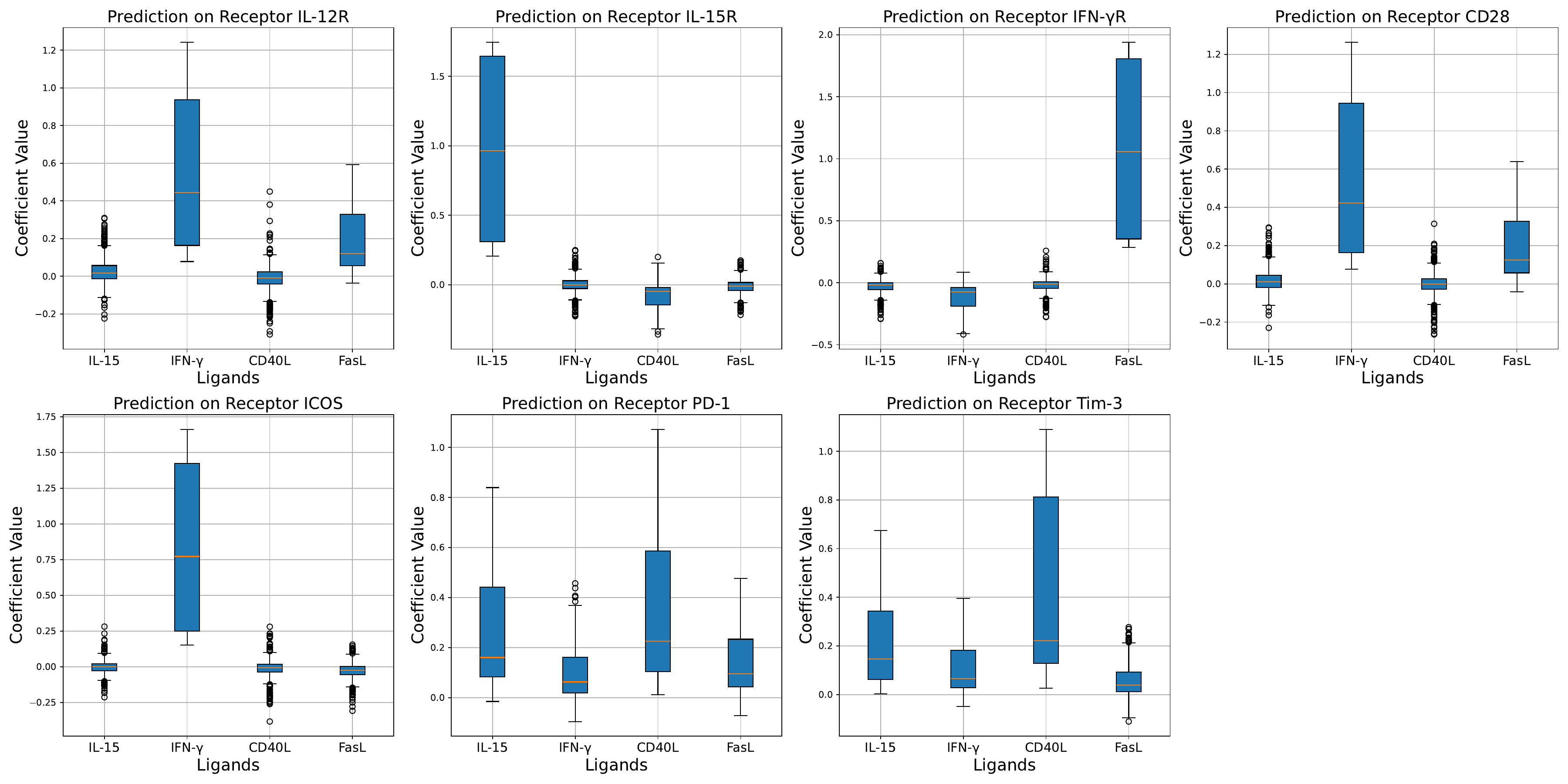}
  \end{tabular}
  \end{center}
  \caption{
    \footnotesize{
      Same as Figure \ref{fig:T_to_B_coeff} for
      NK to T cells
}
} 
  \label{fig:NK_to_T_coeff}
\end{figure}

 \begin{figure}[!t]
\begin{center}
    \begin{tabular}{c}
\widgraph{1\textwidth}{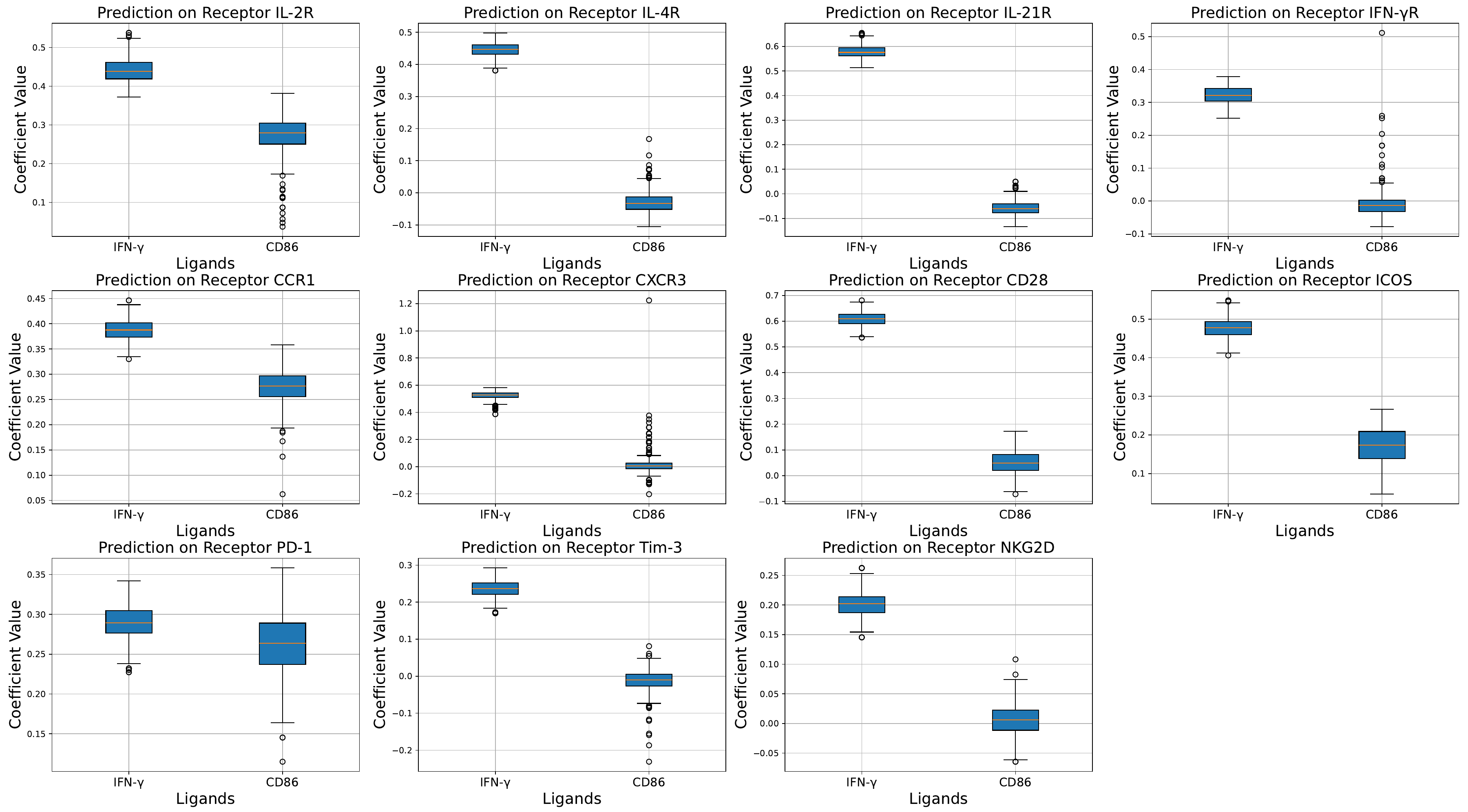}
  \end{tabular}
  \end{center}
  \caption{
    \footnotesize{
      Same as Figure \ref{fig:T_to_B_coeff} for
      T to NK cells
}
} 
  \label{fig:T_to_NK_coeff}
\end{figure}

 \begin{figure}[!t]
\begin{center}
    \begin{tabular}{c}
\widgraph{1\textwidth}{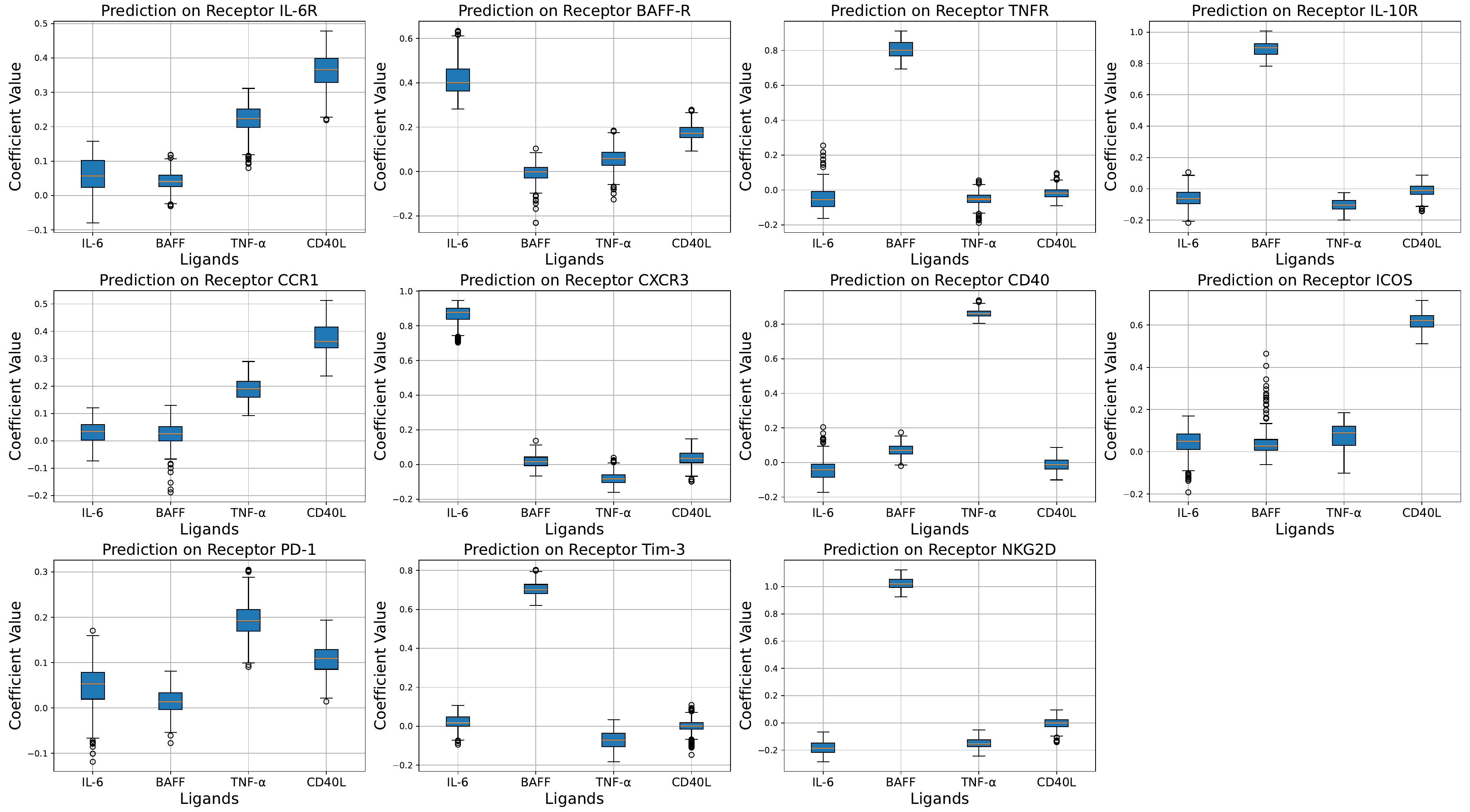}
  \end{tabular}
  \end{center}
  \caption{
    \footnotesize{
      Same as Figure \ref{fig:T_to_B_coeff} for
    B to  NK cells 
}
} 
  \label{fig:B_to_NK_coeff}
\end{figure}

 \begin{figure}[!t]
\begin{center}
    \begin{tabular}{c}
\widgraph{1\textwidth}{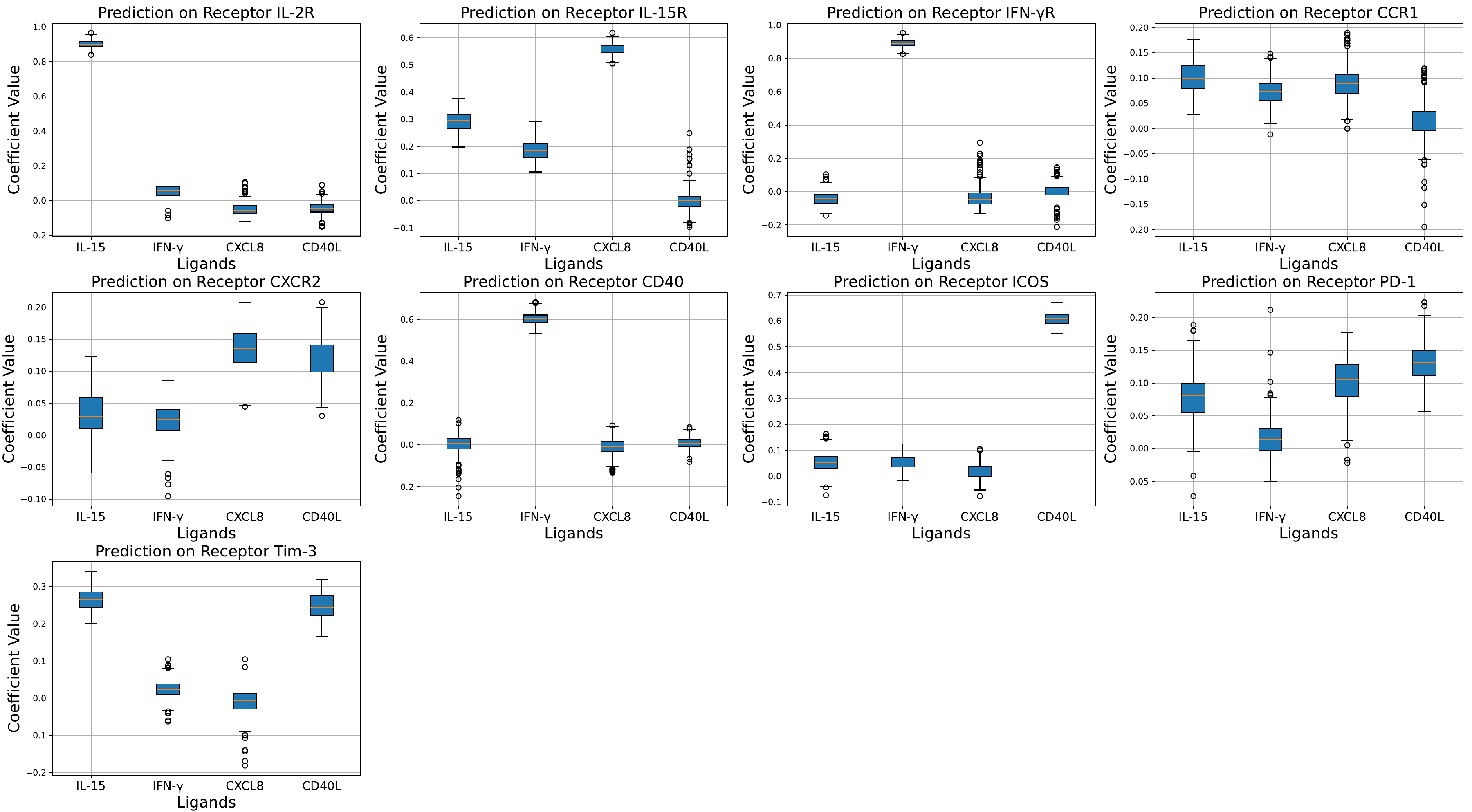}
  \end{tabular}
  \end{center}
  \caption{
    \footnotesize{
      Same as Figure \ref{fig:T_to_B_coeff} for
    NK to  B cells 
}
} 
  \label{fig:NK_to_B_coeff}
\end{figure}

\section{Semi-Simulation}
\label{subsec:semi_simulation_appendix}
Figure~\ref{fig:Semi_simulated_graph_null_RPEs} shows
boxplots of RPEs for all pairs of cell types under the
``no edge'', the ``full graph'', and the ``sparse graph'' setting.
From the low RPE values we conclude that the proposed Bayesian DDR
provides acceptable fits for nearly all pairs of cell types across all
three settings.
 Finally, Figure \ref{fig:coefficient_semi_simulation}
summarizes  posteriors of square errors over all coefficients
of $A_{eij}$, $b_e$. 
Consistently small values across all three simulation scenarios 
indicate that Bayesian DDR provides a good fit to the semi-simulated
data across all settings. 
\begin{figure}[!t]
\begin{center}
\begin{tabular}{c}
       \widgraph{1\textwidth}{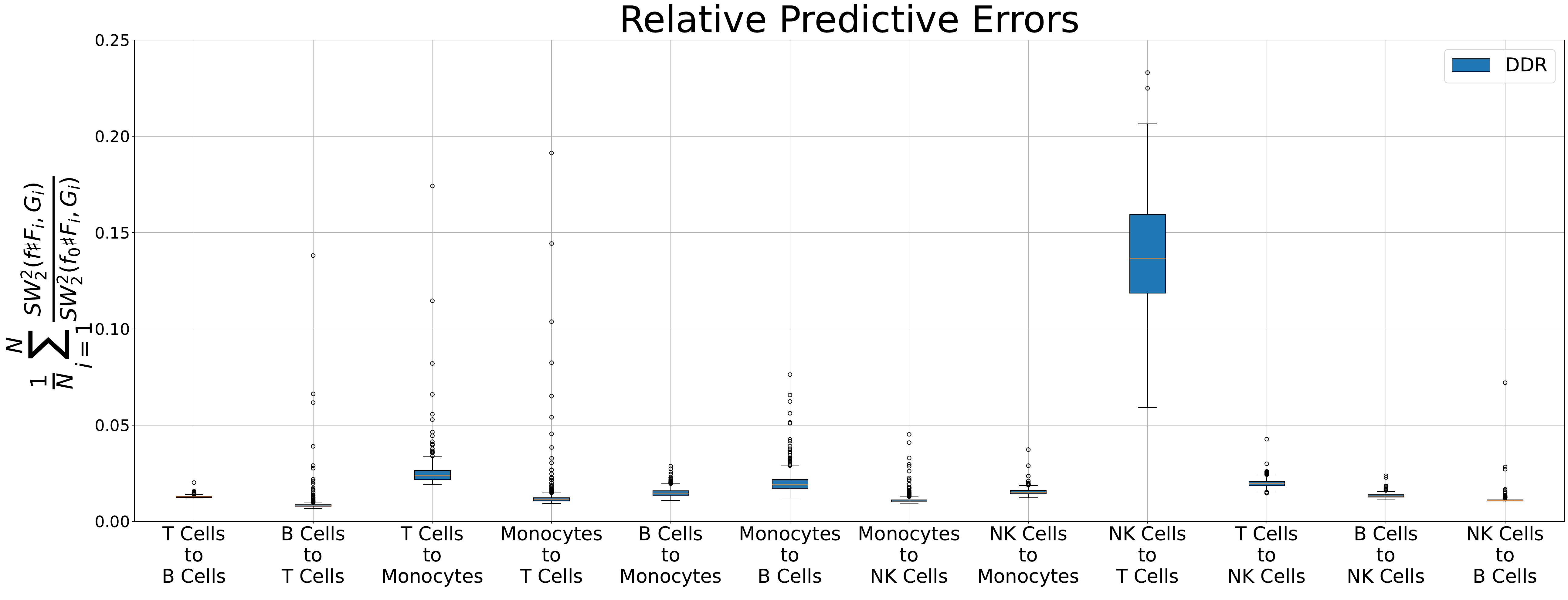} \\
       \widgraph{1\textwidth}{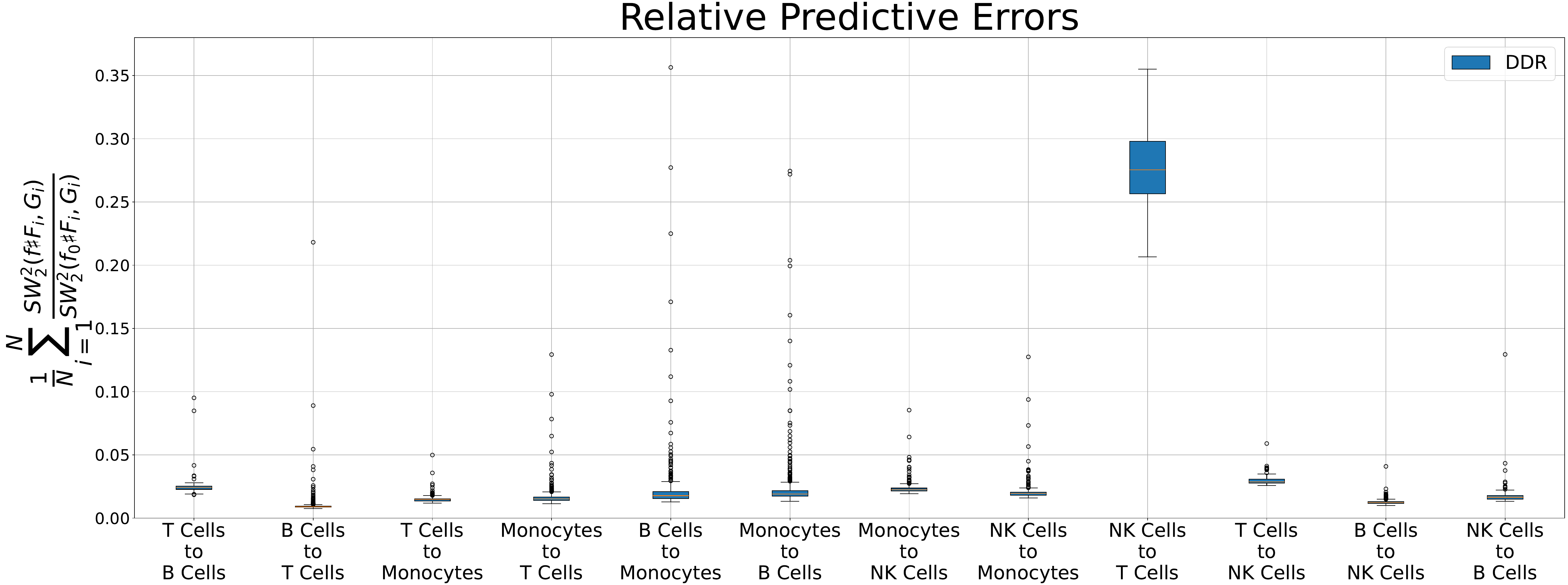}\\
       \widgraph{1\textwidth}{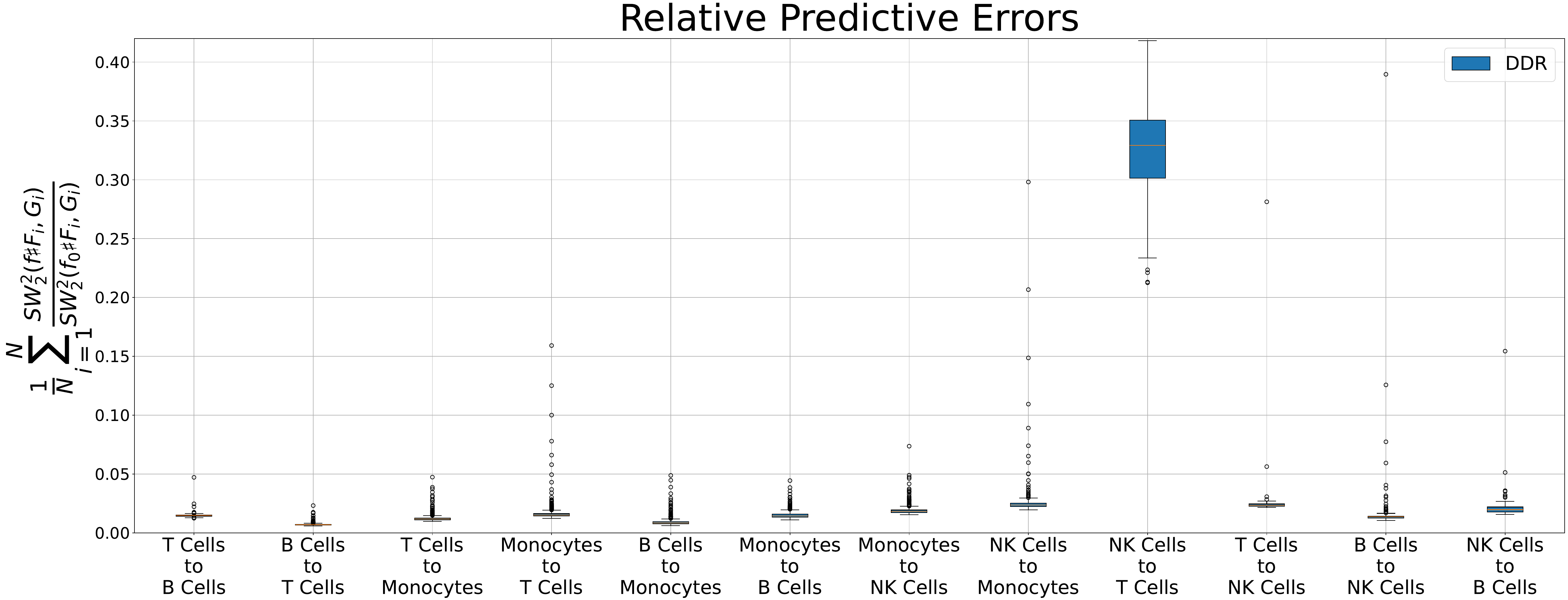}
  \end{tabular}
  
  \end{center}
  \caption{
    \footnotesize{Boxplots of relative predictive
      errors~\eqref{eq:RPE} under the Bayesian DDR for each pair of cell types (on the
      horizontal axis) in the ``no edge'' (top), the ``full graph''
      (middle), and the ``sparse edge'' simulation scenarios.
}
} 
  \label{fig:Semi_simulated_graph_null_RPEs}
\end{figure}

\begin{figure}[!t]
\begin{center}
    \begin{tabular}{c}
\widgraph{1\textwidth}{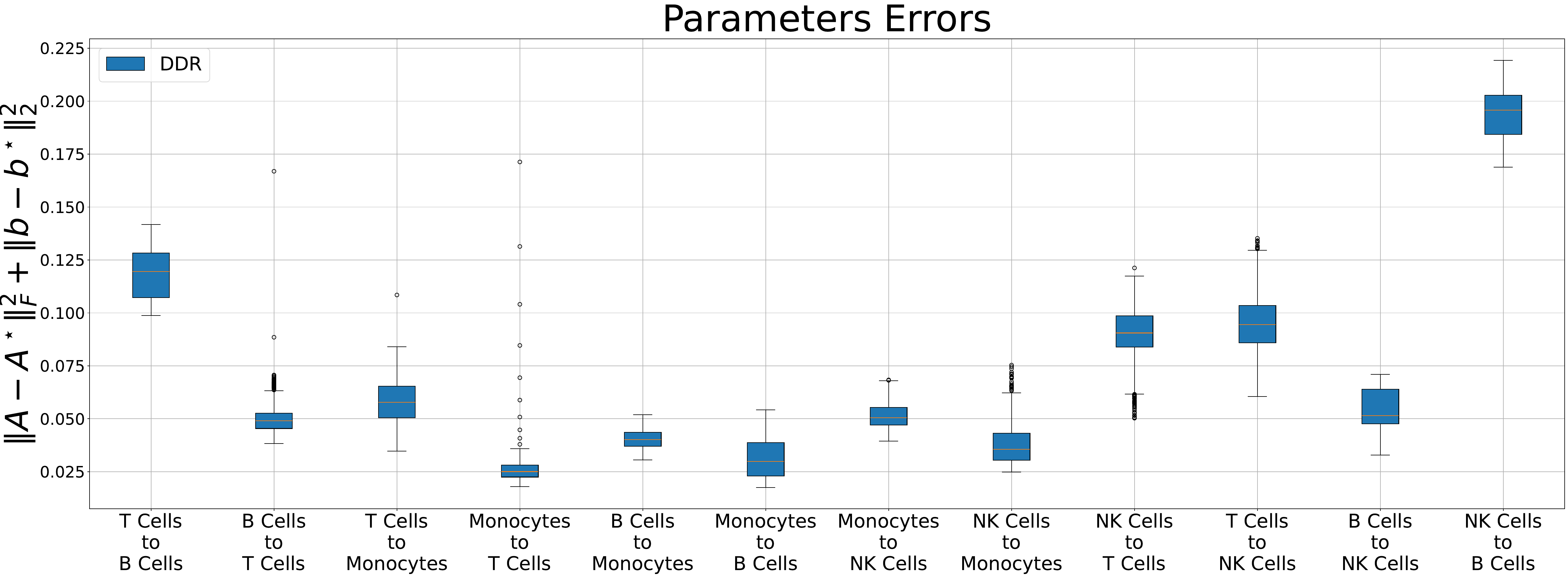} \\
\widgraph{1\textwidth}{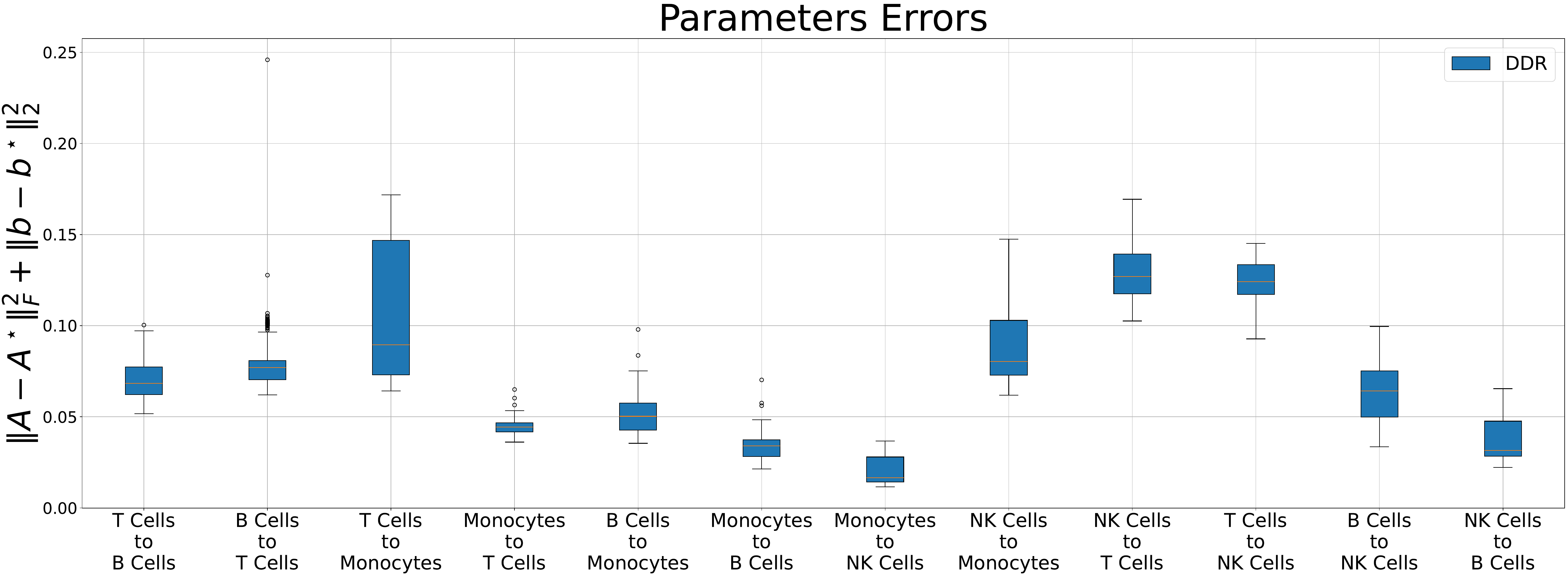}  \\
\widgraph{1\textwidth}{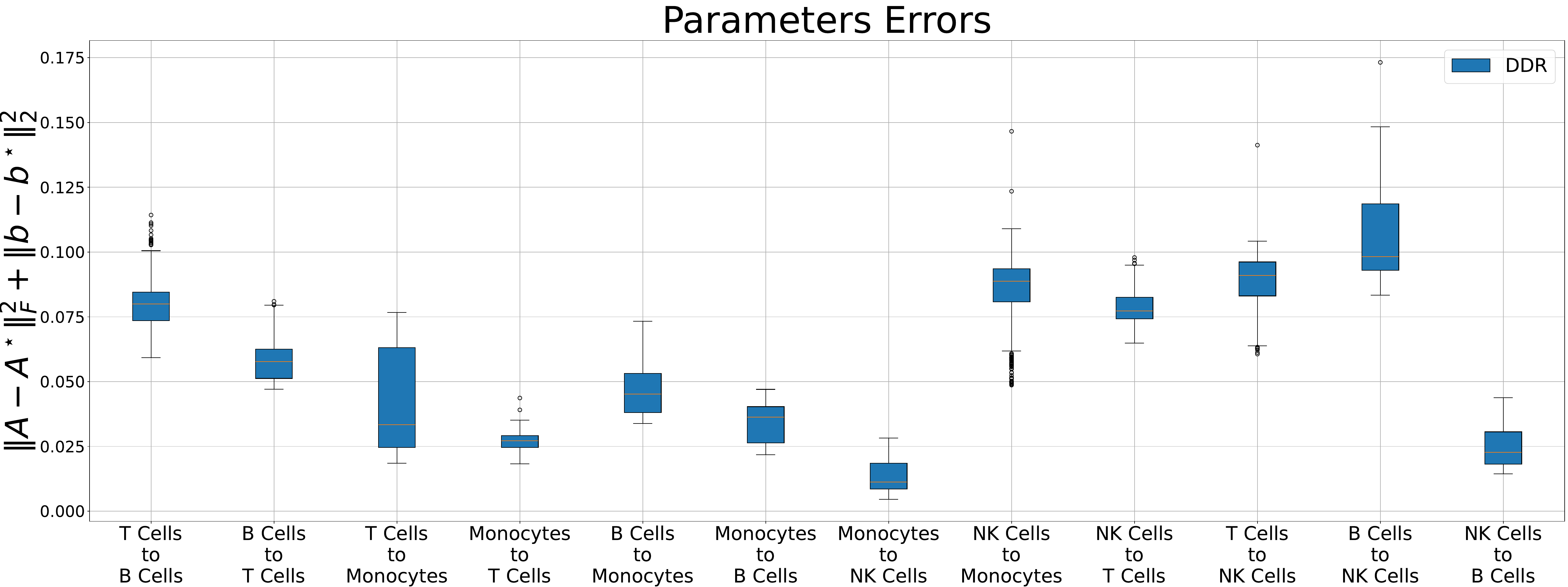}  
  \end{tabular}
  \end{center}
  \caption{
    \footnotesize{Boxplots of  the posterior of squared errors $\|A-A^\star\|_F^2 +\|b-b^\star \|_2^2$,
 over all coefficients of $A_{e}$ and $b_e$, arranged as in Figure
        \ref{fig:Semi_simulated_graph_null_RPEs}. 
}
} 
  \label{fig:coefficient_semi_simulation}
\end{figure}

\end{document}